\documentclass[acmsmall, authorversion]{acmart}

\usepackage[default]{ezlib}

\usepackage{tikz}
\usepackage{xcolor}

\tikzset{
    semithick,
    x=1.05em,
    y=1.05em,
    font=\small
}

\newcommand{\drawJob}[4]{%
    \draw[draw=#1!70!black, fill=#1!30!white]
        #2 ++(-0.5, 0) rectangle ++(1, #3)
        ++(-0.5, 0) node[above=-0.4em] {\strut job #4};
}

\apptocmd{\quote}{\itshape}{}{}


\newcommand{\fcfs}{{\textnormal{FCFS}}}
\newcommand{\nudge}{{\textnormal{Nudge}}}
\newcommand{\srpt}{{\textnormal{SRPT}}}
\newcommand{\nudgek}{{\textnormal{Nudge-K}}}
\newcommand{\nudgem}{{\textnormal{Nudge-M}}}
\newrobustcmd{\boost}[1]{%
    \ifbool{mmode}{%
        {\ifblank{#1}{\textnormal{Boost}}{#1\textnormal{-Boost}}}%
    }{%
        \ifblank{#1}{Boost}{$#1$-Boost}%
    }%
}
\newrobustcmd{\cheat}[1]{%
    \ifbool{mmode}{%
        {\ifblank{#1}{\textnormal{Cheat}}{#1\textnormal{-Cheat}}}%
    }{%
        \ifblank{#1}{Cheat}{$#1$-Cheat}%
    }%
}

\newcommand{\subfigwidth}{0.475\textwidth}

\makeatletter
\if@ACM@authorversion
\newcommand{\authorversion@skip}{0.4in}
\ifdefstring{\ACM@format}{sigconf}{
    \renewcommand{\authorversion@skip}{0.1in}
}{}
\let\old@ACM@footnote\footnotetextcopyrightpermission
\renewcommand{\footnotetextcopyrightpermission}[1]{\old@ACM@footnote{#1\vspace{\authorversion@skip}}}
\fi
\makeatother

\begin{document}

\setcopyright{acmlicensed}
\acmJournal{POMACS}
\acmYear{2024} \acmVolume{8} \acmNumber{2} \acmArticle{27} \acmMonth{6}\acmDOI{10.1145/3656011}

\begin{CCSXML}
<ccs2012>
<concept>
<concept_id>10002944.10011123.10011674</concept_id>
<concept_desc>General and reference~Performance</concept_desc>
<concept_significance>500</concept_significance>
</concept>
<concept>
<concept_id>10002950.10003648.10003688.10003689</concept_id>
<concept_desc>Mathematics of computing~Queueing theory</concept_desc>
<concept_significance>500</concept_significance>
</concept>
<concept>
<concept_id>10003033.10003079.10003080</concept_id>
<concept_desc>Networks~Network performance modeling</concept_desc>
<concept_significance>500</concept_significance>
</concept>
<concept>
<concept_id>10010147.10010341.10010342</concept_id>
<concept_desc>Computing methodologies~Model development and analysis</concept_desc>
<concept_significance>300</concept_significance>
</concept>
<concept>
<concept_id>10011007.10010940.10010941.10010949.10010957.10010688</concept_id>
<concept_desc>Software and its engineering~Scheduling</concept_desc>
<concept_significance>300</concept_significance>
</concept>
</ccs2012>
\end{CCSXML}

\ccsdesc[500]{General and reference~Performance}
\ccsdesc[500]{Mathematics of computing~Queueing theory}
\ccsdesc[500]{Networks~Network performance modeling}
\ccsdesc[300]{Computing methodologies~Model development and analysis}
\ccsdesc[300]{Software and its engineering~Scheduling}

\keywords{%
    scheduling;
    response time;
    sojourn time;
    tail latency;
    service level objective (SLO);
    M/G/1 queue;
    light-tailed distribution;
    FCFS;
    Boost scheduling}

\title{Strongly Tail-Optimal Scheduling in the Light-Tailed M/G/1}

\author{George Yu}
\affiliation{%
    \institution{Cornell University}
    \department{School of Operations Research and Information Engineering}
    \city{Ithaca}
    \state{NY}
    \country{USA}}
\author{Ziv Scully}
\affiliation{%
    \institution{Cornell University}
    \department{School of Operations Research and Information Engineering}
    \city{Ithaca}
    \state{NY}
    \country{USA}}
\begin{abstract}
    % Full paper version

We study the problem of scheduling jobs in a queueing system, specifically an M/G/1 with light-tailed job sizes, to asymptotically optimize the response time tail. This means scheduling to make $\mathbf{P}[T > t]$, the chance a job's response time exceeds~$t$, decay as quickly as possible in the $t \to \infty$ limit. For some time, the best known policy was First-Come First-Served (FCFS), which has an asymptotically exponential tail: $\mathbf{P}[T > t] \sim C e^{-\gamma t}$. FCFS achieves the optimal \emph{decay rate}~$\gamma$, but its \emph{tail constant}~$C$ is suboptimal. Only recently have policies that improve upon FCFS's tail constant been discovered. But it is unknown what the optimal tail constant is, let alone what policy might achieve it.

In this paper, we derive a closed-form expression for the optimal tail constant~$C$, and we introduce \emph{$\gamma$-Boost}, a new policy that achieves this optimal tail constant. Roughly speaking, $\gamma$-Boost operates similarly to FCFS, but it pretends that small jobs arrive earlier than their true arrival times. This significantly reduces the response time of small jobs without unduly delaying large jobs, improving upon FCFS's tail constant by up to 50\% with only moderate job size variability, with even larger improvements for higher variability. While these results are for systems with full job size information, we also introduce and analyze a version of $\gamma$-Boost that works in settings with partial job size information, showing it too achieves significant gains over FCFS. Finally, we show via simulation that $\gamma$-Boost has excellent practical performance.

\end{abstract}

\maketitle

\section{Introduction, background, and key ideas}
\label{sec:intro}
Service Level Objectives (SLOs) for practical queueing systems often relate to the \emph{tail} of the system's response time distribution~$T$. The tail is the function that maps an amount of time~$t$ to $\P{T > t}$, the probability that a job's response time~$T$ exceeds~$t$, where a job's \emph{response time} is the amount of time between the job's arrival and departure.

% For instance, an SLO may promise that a response time quantile~$t_q$, the threshold such that $\P{T > t_q} = 1 - q$, is small, typically for values of $q$ near~$1$. Meeting such quantile targets amounts to ensuring $\P{T > t}$ decays quickly as $t$ increases.
% Reflecting the goal of preventing especially large response times, SLOs typically relate to high response time quantiles, such as the 99th percentile~$t_{0.99}$.

Motivated by the problem of meeting SLOs, we consider the problem of scheduling jobs to minimize the tail $\P{T > t}$ in the M/G/1 queue. We actually focus on \emph{asymptotically} minimizing the tail, optimizing the decay of $\P{T > t}$ in the $t \to \infty$ limit. This is an extensively studied problem in queueing theory \citep{boxma_tails_2007, wierman_tailoptimal_2012, nuyens_preventing_2008, nair_tailrobust_2010, zwart_sojourn_2000, grosof_nudge_2021, vanhoudt_stochastic_2022, scully_characterizing_2020, scully_when_2024, scully_new_2022, gavish_markovian_1977, nuyens_largedeviations_2006, borst_impact_2003, friedman_fairness_2003, friedman_protective_2003, stolyar_largest_2001, nunez-queija_queues_2002} for a number of reasons:
\* Optimizing the tail $\P{T > t}$ for any particular value of~$t$ is seldom the sole design objective. Instead, one generally hopes to achieve low $\P{T > t}$ for a range of values of~$t$.
\* Because practical SLOs relate to high-quantile response times, meeting those SLOs corresponds to optimizing $\P{T > t}$ for large values of~$t$.
\* Optimizing $\P{T > t}$ for fixed finite~$t$ appears to be theoretically intractable, but there has been promising recent progress on asymptotic improvements in the $t \to \infty$ limit \citep{grosof_nudge_2021, vanhoudt_stochastic_2022}.
\*/

In this paper, we study the M/G/1 with light-tailed job size distributions. We propose a new policy, called \emph{\boost{\gamma}}, and prove it has \emph{asymptotically optimal response time tail} in a sense made precise in \cref{sec:intro:background} below. This resolves a \emph{significant open problem in queueing theory} \citep{boxma_tails_2007, wierman_tailoptimal_2012}. Moreover, \boost{\gamma} has excellent practical performance, as illustrated in \cref{fig:intro:performance}.

The rest of this section gives background on the problem of asymptotically optimal tail scheduling, with discussion of prior work integrated throughout, and describes the main ideas behind our solution. See \cref{sec:intro:contributions} for a summary of our contributions and an outline of the rest of the paper.

\subsection{Background on weak and strong tail optimality}
\label{sec:intro:background}

% We begin with some background on scheduling to optimize tail performance in the M/G/1. We denote the job size

Consider an M/G/1 with job size distribution~$S$. For now, we primarily focus on the \emph{full-information} setting where job sizes are known to the scheduler, but some of our results apply more broadly to \emph{partial-information} settings (\cref{sec:model:scheduler}).

Let $T_\pi$ denote the response time distribution under policy~$\pi$. Following \citet{boxma_tails_2007}, we say a policy~$\pi$ is \emph{weakly tail-optimal} if there exists a constant~$c \geq 1$ such that
\[
    \sup_{\pi'} \limsup_{t \to \infty} \frac{\P{T_\pi > t}}{\P{T_{\pi'} > t}} = c.
\]
If additionally $c = 1$, we say $\pi$ is \emph{strongly tail-optimal}.

Whether a scheduling policy is weakly tail-optimal depends critically on whether the job size distribution~$S$ is heavy-tailed or light-tailed. If $S$ is heavy-tailed, then several preemptive policies like Shortest Remaining Processing Time (SRPT) and Least Attained Service (LAS) are known to be weakly tail-optimal and conjectured to to be strongly tail-optimal \citep{boxma_tails_2007, wierman_tailoptimal_2012}. In fact, we observe in \cref{sec:heavy-tailed_strong_optimality} that a result of \citet{wierman_tailoptimal_2012} implies strong tail optimality of SRPT, LAS, and other policies for an important class of heavy-tailed distributions. % One interesting observation is that in the heavy-tailed case, optimizing tail asymptotics aligns with optimizing mean response time~$\E{T}$. For instance, SRPT minimizes mean response time \citep{schrage_proof_1968} when job sizes are known, and LAS has good mean response time in the heavy-tailed case when job sizes are unknown \citep{kamphorst_heavytraffic_2020, nuyens_foreground_2008}. The main takeaway is that
The problem of achieving strong tail optimality is thus largely solved in the heavy-tailed case.

In this work, we focus on the case of light-tailed job size distributions~$S$, specifically so-called \emph{class~I} distributions \citep{abate_asymptotics_1997, abate_waitingtime_1994} (\cref{def:light_tailed_dist}), for which strong tail optimality is a significant open problem \citep{boxma_tails_2007, wierman_tailoptimal_2012}. For some time, the only common policy known to be weakly tail-optimal was First-Come First-Served (FCFS), which has asymptotically exponential response time tail. That is,
% \citep{boxma_tails_2007, wierman_tailoptimal_2012}. But recently, a family of policies known as \emph{Nudge} dethroned FCFS, achieving superior asymptotic tail performance \citep{grosof_nudge_2021, vanhoudt_stochastic_2022} (see \cref{sec:intro:nudge}). This demonstrates that FCFS is not strongly tail-optimal. But it is not known whether Nudge or any other policy is strongly tail-optimal. What we do know is that unlike in the heavy-tailed setting, one cannot simply minimize mean response time and hope to achieve good tail performance. In fact, SRPT is in some sense \emph{tail-pessimal} in the light-tailed case \citep{nuyens_preventing_2008}. See \cref{sec:intro:challenge} below for further discussion.
%
% In this work, we focus on a specific class of light-tailed job size distributions known as  When $S$ is class~I, it is known that FCFS and other weakly optimal policies have asymptotically exponential response time tail: in the $t \to \infty$ limit \citep{boxma_tails_2007, wierman_tailoptimal_2012},
\[
    \P{T_\fcfs > t} \sim C_\fcfs \exp(-\gamma t),
\]
where $\gamma > 0$ is a constant called the \emph{decay rate}, and $C_\fcfs > 0$ is a constant we call FCFS's \emph{tail constant}. Both $\gamma$ and $C_\fcfs$ depend on $S$ and the system's arrival rate.

It is known that no policy can achieve asymptotic decay rate greater than~$\gamma$ \citep{stolyar_largest_2001, boxma_tails_2007}, so we can measure the performance of a weakly tail-optimal policy~$\pi$ by its tail constant
\[
    \label{eq:intro:pi_asymptotics}
    C_\pi = \lim_{t \to \infty} \exp(\gamma t) \, \P{T > t}.
\]
The question of finding a strongly tail-optimal policy thus amounts to minimizing~$C_\pi$ over all policies~$\pi$. Until recently, it was conjectured that FCFS may be strongly tail-optimal, but recent progress has improved upon FCFS's tail constant \citep{grosof_nudge_2021, vanhoudt_stochastic_2022} (\cref{sec:intro:nudge}). This prompts a question:
\begin{quote}
    What is the smallest possible tail constant~$C_\pi$, and what policy~$\pi$ achieves it?
\end{quote}

\subsection{Obstacle: prioritizing short jobs without delaying long jobs}
\label{sec:intro:challenge}

As explained in \cref{sec:intro:background}, optimizing tail asymptotics with light-tailed job sizes is an open problem, in contrast to the heavy-tailed case. Why is the light-tailed case so much more difficult? The main obstacle is that there is a tension between prioritizing short jobs and delaying long jobs. This tension is best illustrated by contrasting two policies, FCFS and SRPT.

Suppose a ``tagged'' job of random size~$S$ arrives to a steady-state system and observes \emph{work}~$W$, meaning the total remaining service time of jobs in the system is~$W$.
\* FCFS serves jobs in the order they arrive. This means the tagged job's response time is $T_\fcfs = W + S$. In particular, the job's response time is unaffected by future arrivals.
\* SRPT always preemptively serves the job of least remaining service time. This means the tagged job may not need to wait for all of the work~$W$ to be completed before entering service. But future arrivals of size less than~$S$ may be prioritized over the tagged job.
\*/
The reason SRPT is good in the heavy-tailed setting is that the amount of work from future arrivals that delays the tagged job, which we denote by~$R_\srpt(S)$, has a lighter tail than~$W$. But in the light-tailed setting, $R_\srpt(S)$ is heavier-tailed than~$W$, with a decay rate less than~$\gamma$. See \citet{nuyens_preventing_2008}, who prove these results for SRPT and a class of related policies, for details.

The takeaway of the above comparison is that for light-tailed job size distributions, strictly prioritizing short jobs delays long jobs too much for good tail performance. But prioritizing short jobs is essentially the only tool we have for improving response times. The question is thus: how should one \emph{partially} prioritize short jobs to improve tail performance?

A number of works have studied scheduling with some sort of partial priority, whether by having just a few priority buckets \citep{chen_scheduling_2021, harchol-balter_sizebased_2003, marin_sizebased_2020} or by dynamically changing priority over time \citep{fajardo_controlling_2015, fajardo_waiting_2017, stanford_waiting_2014}. While the tail asymptotics of most of these policies have not been formally studied, they seem unlikely to be weakly tail-optimal. This is because they still have the property that a sufficiently long tagged job might be delayed by a constant fraction of future arrivals. A result of \citet[Proposition~9.9]{scully_when_2024} suggests this should lead to worse decay rate, though their result does not directly apply to all of the other policies cited. \Citet[Theorem~5.5]{scully_when_2024} also show that no policy in the recently proposed class of ``SOAP'' policies \citep{scully_new_2022, scully_soap_2018, scully_soap_2018a} can improve upon FCFS's tail constant, because all SOAP policies other than FCFS have decay rate worse than~$\gamma$.\footnote{%
    \Citet{scully_when_2024} actually consider only a subset of SOAP policies, but one can generalize the relevant part of their argument to cover all SOAP policies.%
}

\subsection{Nudge: a promising but limited first step}
\label{sec:intro:nudge}

The first improvement upon FCFS's tail constant was through the \emph{Nudge} family of policies, introduced by \citet{grosof_nudge_2021} and expanded upon by \citet{vanhoudt_stochastic_2022} and \citet{charlet_tail_2024}. In its simplest variant, Nudge creates two classes of jobs, \emph{short} and \emph{long}, then runs FCFS with a small modification, illustrated in \cref{fig:intro:example:nudge}:
\* When a short job arrives, it is allowed to pass in front of up to~$K$ large jobs, where $K$ is a fixed constant.
\* Each large job can be passed by a limited number of short jobs. Different variants of Nudge differ in exactly how the limiting works. The two most important variants are the following:
\** \emph{Nudge-K} \citep{vanhoudt_stochastic_2022}: Each large job can be passed only once.
\** \emph{Nudge-M} \citep{charlet_tail_2024}: Small jobs only pass large jobs that are within the $K$ most recent arrivals.
\**/
When $K = 1$, Nudge-K and Nudge-M coincide, and are called simply ``Nudge'' \citep{grosof_nudge_2021}.
\*/

\begin{figure}
    \centering
    \begin{subfigure}[t]{0.3\linewidth}
        \centering
        \begin{tikzpicture}
    \providecommand{\aX}{0}
    \providecommand{\aY}{3.5}
    \providecommand{\sX}{3.25}
    \providecommand{\sY}{1.75}

    \draw (\aX, 0) node[below, align=center] {earlier\\arrival};
    \draw (\aY, 0) node[below, align=center] {later\\arrival};

    \drawJob{orange}{(\aX, 0)}{\sX}{X}
    \drawJob{teal}{(\aY, 0)}{\sY}{Y}

    \draw[thick, teal!60!black, densely dashed] (\aY - 0.5, \sY/2) edge[bend right=6em, ->] (\aX - \aY + 0.5, \sY/2);

    \node at (\aX, \sX/2) {$s_\mathrm{X}$};
    \node at (\aY, \sY/2) {$s_\mathrm{Y}$};

    % Vertical alignment
    \draw (2, -3.5) ++(0, -0.35) node[below=-0.4em] {\strut};
\end{tikzpicture}
        \Description{A cartoon of a large job X and a small job Y. Job X arrives earlier than job Y, but because Y is small, it gets to overtake X, as indicated by a dotted arrow.}
        % \caption{Nudge always serves Z before~Y but after~X, regardless of the arrival times.}
        \caption{Nudge always serves X after~Y, regardless of their arrival times.}
        \label{fig:intro:example:nudge}
    \end{subfigure}%
    \hfill
    \begin{subfigure}[t]{0.3\linewidth}
        \centering
        \begin{tikzpicture}
    \providecommand{\aX}{1}
    \providecommand{\aY}{3}
    \providecommand{\bX}{2.75}
    \providecommand{\bY}{6}
    \providecommand{\sX}{3.25}
    \providecommand{\sY}{1.75}
    \providecommand{\timeline}{4}

    \drawJob{orange}{(\aX, 0)}{\sX}{X}
    \drawJob{teal}{(\aY, 0)}{\sY}{Y}

    \draw[->, thick, draw=orange!70!black] (\aX, 0) -- ++(0, -1.5) -- node[midway, above=-0.3em] {\strut$b(s_\mathrm{X})$} ++(-\bX, 0);
    \draw[->, thick, draw=teal!70!black] (\aY, 0) -- ++(0, -3) -- node[midway, above=-0.3em] {\strut$b(s_\mathrm{Y})$} ++(-\bY, 0);

    \draw[<->] (-\timeline, -3.5) node[left] {\strut} -- ++(2*\timeline, 0) node[right] {\strut time};

    \draw (\aX, -3.5) -- ++(0, -0.35) node[below=-0.4em] {\strut$a_\mathrm{X}$};
    \draw (\aY, -3.5) -- ++(0, -0.35) node[below=-0.4em] {\strut$a_\mathrm{Y}$};

    % \draw (\aX, -1.5) ++(-\bX, 0) -- ++(0, -0.35) node[below] {\strut$a_\mathrm{X}'$};
    % \draw (\aY, -1.5) ++(-\bY, 0) -- ++(0, -0.35) node[below] {\strut$a_\mathrm{Y}'$};

    \node at (\aX, \sX/2) {$s_\mathrm{X}$};
    \node at (\aY, \sY/2) {$s_\mathrm{Y}$};
\end{tikzpicture}
        \Description{A cartoon of a large job X and a small job Y on a timeline. Each job is drawn at its arrival time, and an arrow coming out of each job indicates the job's boost. Job Y gets a bigger boost than job X, because smaller jobs get bigger boosts. Job X arrives slightly before job Y, so job Y's boosted arrival time is earlier than job X's boosted arrival time.}
        % \caption{Boost may serve Z before both X and~Y if the arrival times are close together.}
        \caption{Boost serves X after Y if their arrival times are close together.}
        \label{fig:intro:example:boost_pass}
    \end{subfigure}%
    \hfill
    \begin{subfigure}[t]{0.3\linewidth}
        \centering
        \input{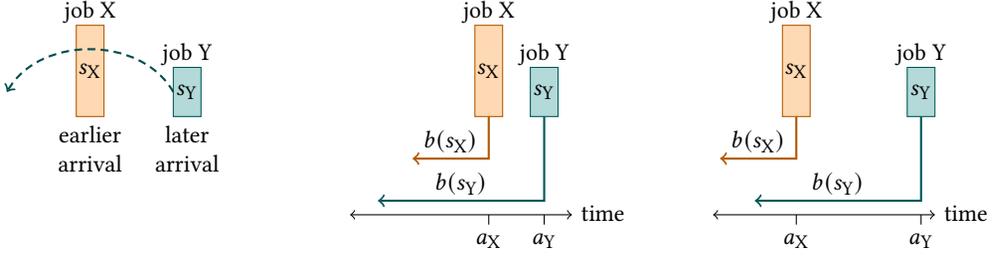}
        \Description{A cartoon of a large job X and a small job Y on a timeline. Each job is drawn at its arrival time, and an arrow coming out of each job indicates the job's boost. Job Y gets a bigger boost than job X, because smaller jobs get bigger boosts. Job X arrives long before job Y, so job Y's boosted arrival time is later than job X's boosted arrival time, despite Y's bigger boost.}
        % \caption{Boost may serve Z after both X and~Y if the arrival times are far apart.}
        \caption{Boost serves X before Y if their arrival times are far apart.}
        \label{fig:intro:example:boost_no_pass}
    \end{subfigure}%
    % \caption{Comparison between now Nudge and Boost each handle two long jobs~X and~Y arriving before a third job~Z. Suppose that Z arrives before X enters service. Nudge decides the order to serve the jobs based only on the arrival order, as shown in~(a). In contrast, Boost uses not just the arrival order but also the \emph{arrival times}, as shown in (b) and~(c).}
    \caption{Comparison between how Nudge and Boost each handle a long job~X arriving before a short job~Y. Suppose that Y arrives before X enters service. Nudge decides the order to serve the jobs based only on the arrival order, as shown in~(a). In contrast, Boost uses not just the arrival order but also the respective \emph{arrival times}, as shown in (b) and~(c). Notation: job~$i$'s arrival time is~$a_i$, its size is~$s_i$, and its boost is~$b(s_i)$.}
    \label{fig:intro:example}
\end{figure}

Nudge is a family of policies rather than a single policy, because there are many ways to decide which jobs are short and which are long; one can vary the parameter~$K$; and one can choose between Nudge-K, Nudge-M, and other variants \citep{charlet_tail_2024}. One can think of these parameters as controlling the degree to which short jobs are prioritized over long jobs. For instance, larger values of~$K$ further prioritize short jobs, and for a fixed value of~$K$, Nudge-K is more conservative about letting short jobs pass long jobs than Nudge-M.

Recent progress on Nudge has yielded several improvements to the best known tail constant. \Citet{grosof_nudge_2021} introduce Nudge with $K = 1$ and show that with appropriate tuning, Nudge achieves $C_\nudge < C_\fcfs$, thus demonstrating that FCFS is not strongly tail-optimal. In fact, they show that with sufficiently conservative tuning, Nudge \emph{stochastically improves} upon FCFS, meaning $\P{T_\nudge > t} \leq \P{T_\fcfs > t}$ for all $t > 0$.

Building on this progress, \citet{vanhoudt_stochastic_2022} introduces Nudge-K for $K \geq 2$ and, for any given split between small and large jobs, characterizes the value of $K$ that minimizes~$C_\nudgek$. The optimal value is generally neither $K = 1$ nor $K = \infty$. This reflects the fact that while short jobs should get some priority, giving them too much priority hurts long jobs, and thereby the tail constant~$C_\nudgek$.

Concurrently with this work, \citet{charlet_tail_2024} introduce the Nudge-M variant and show several results about it. The most important of these is that for any given split between small and large jobs, Nudge-M with the optimal value of~$K$ achieves the \emph{minimum possible tail constant} out of any variant of Nudge. That is, Nudge-M is strongly tail-optimal among Nudge policies. \Citet{charlet_tail_2024} also characterize the value of~$K$ that leads to this minimal~$C_\nudgem$, showing that it coincides with the value that minimizes~$C_\nudgek$.

While Nudge is significant due to its improving upon FCFS, there are two reasons to believe that Nudge can also be improved upon. First, it seems likely that it would help to have finer-grained distinctions between job sizes, as opposed to grouping them into just two classes, and it may be beneficial to allow jobs to move many spots in the queue. As an extreme example, if a job were size~$0$, it would make sense to let it jump straight to the front of the queue. Second, while Nudge makes use of the order in which jobs arrive, it does not make use of the \emph{amounts of time between arrivals}. For instance, suppose a job~X arrives before a shorter job~Y. If the time between the arrivals is very small, as in \cref{fig:intro:example:boost_pass}, it may make sense to serve the shorter~Y before~X. But if there is a long interarrival time between X and~Y, as in \cref{fig:intro:example:boost_no_pass}, it may make sense to keep X in front of~Y, because Y will not have been waiting as long by the time the system starts serving~X.

\subsection{Our answer: Boost}

Motivated by the limitations of Nudge discussed above, we define \emph{\boost{}}, a new family of scheduling policies. In the full-information setting where job sizes are known to the scheduler, an instance of \boost{} is specified by a \emph{boost function} $b : \bbR_+ \to \bbR$, where $b(s)$ is called the \emph{boost} of a job of size~$s$. The rough idea is that \boost{} acts like FCFS, except it pretends that a job of size~$s$ arrives $b(s)$ time earlier than it actually does. Specifically, if a job of size~$s$ arrives at time~$a$, we define its \emph{boosted arrival time} to be
\[
    \textnormal{boosted arrival time} = \textnormal{arrival time} - \textnormal{boost} = a - b(s).
\]
\boost{} then follows one scheduling rule: \emph{prioritize jobs from least to greatest boosted arrival time}. See \cref{fig:intro:example:boost_pass, fig:intro:example:boost_no_pass} for an illustration. Notice that \boost{}, unlike Nudge, takes into account not just the arrival order but also the arrival times.

One can define preemptive and nonpreemptive versions of \boost{}, depending on whether the priority rule is applied at every moment in time or only when a job completes. The distinction turns out not to affect \boost{}'s tail asymptotics, so our results apply to both versions.

The boost function~$b$ determines how \boost{} balances the tension between prioritizing short jobs and prioritizing jobs that have been waiting a long time. For example, setting $b(s) = 0$ reduces the policy to FCFS, whereas setting $b(s) = r/s$ for a large constant~$r$ results in prioritizing jobs nearly entirely based on their size, similar to SRPT. We therefore ask: what boost function is best?

\begin{figure}
    \begin{subfigure}{\subfigwidth}
        \includegraphics[width=\linewidth]{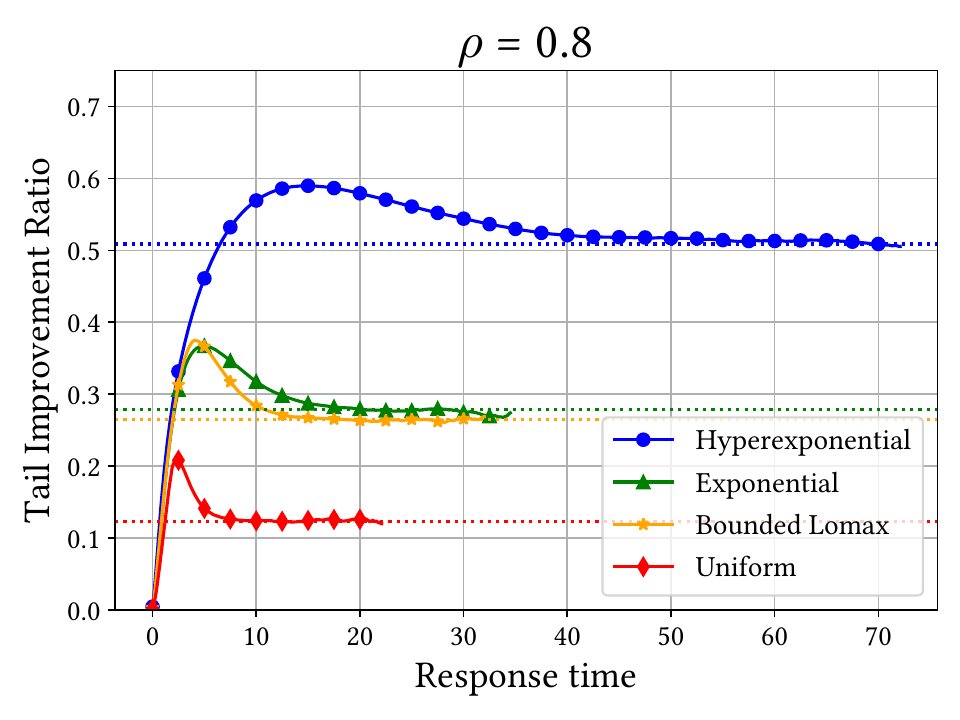}
        \caption{\boost{}'s tail improvement ratio for several job size distributions.}
        \label{fig:intro:performance:distributions}
    \end{subfigure}%
    \hfill
    \begin{subfigure}{\subfigwidth}
        \includegraphics[width=\linewidth]{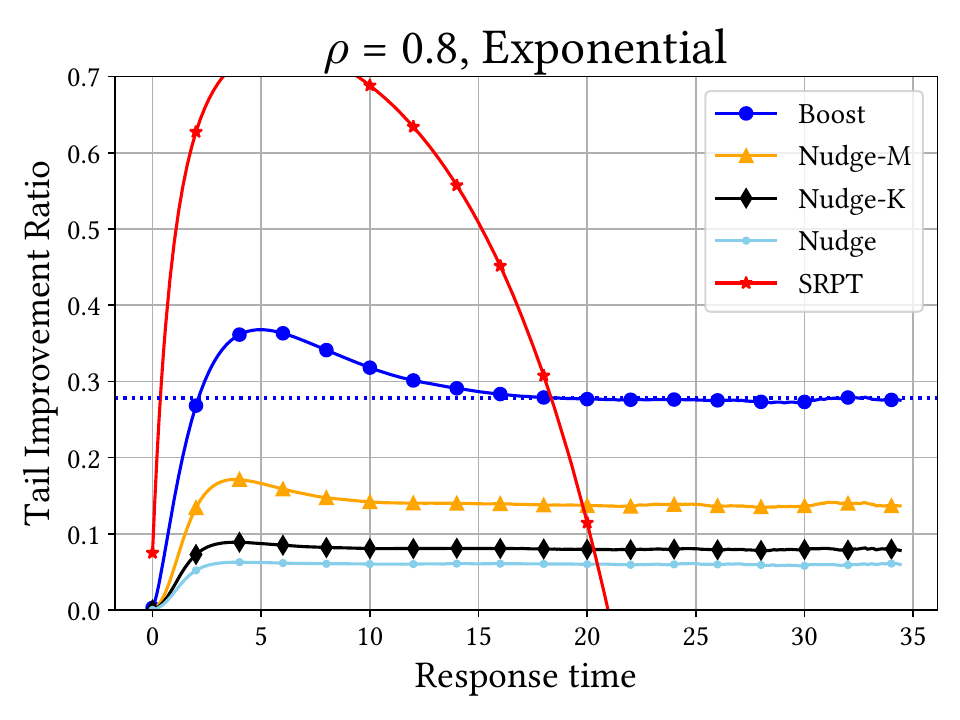}
        \caption{Several policies' tail improvement ratios for an exponential job size distribution.}
        \label{fig:intro:performance:policies}
    \end{subfigure}%
    % \caption{Empirical performance (higher is better) of \boost{}, specifically the strongly tail-optimal \boost{\gamma}, (a)~on several job size distributions, and (b)~compared to two other policies, Nudge and SRPT. The plots show \emph{tail improvement ratio} $1 - \P{T_\pi > t}/\P{T_\fcfs > t}$ as a function of~$t$. Dotted lines indicate the asymptotic improvement ratio $1 - C_\pi/C_\fcfs$. The load is $\rho = 0.8$, and the mean job size is $\E{S} = 1$. See \cref{sec:simulation} for additional details on the job size distributions and other simulation parameters.}
    \caption{Empirical performance (higher is better) of \boost{}, specifically the strongly tail-optimal \boost{\gamma}, (a)~on several job size distributions, and (b)~compared to two other policies, Nudge (and the~K and~M variants with optimal parameter $K$) and SRPT. The plots show \emph{tail improvement ratio} $1 - \P{T_\pi > t}/\P{T_\fcfs > t}$ as a function of~$t$. Dotted lines indicate the asymptotic tail improvement ratio $1 - C_\pi/C_\fcfs$. The load is $\rho = 0.8$, and the mean job size is $\E{S} = 1$. See \cref{sec:simulation} for additional details on the job size distributions and other simulation parameters.}
    \Description[Comparison of Boost on several size distributions, and compared to other policies.]{Plot (a) shows the performance improvement of \boost{\gamma} compared to FCFS for Hyperexponential, Exponential, Bounded Lomax, and Uniform distributions. Improvement amounts range from over 50\% (Hyperexponential) to about 12\% (Uniform). Plot (b) compares the performance improvements of \boost{\gamma} to Nudge (including Nudge-K and Nudge-M) variants. \boost{\gamma} performs better than all Nudge variants. The x-axis for both plots is the response time. The y-axis for both plots is the tail improvement ratio $1 - \P{T_\pi > t}/\P{T_\fcfs > t}$.}
    \label{fig:intro:performance}
\end{figure}

We prove two main theoretical results about \boost{}. First, we find an \emph{explicit formula} for its tail constant~$C_\boost{}$ in terms of the boost function~$b$ (\cref{thm:general_transform}). Second, we study a particular version of \boost{}, which we call \emph{\boost{\gamma}}, where the boost function is
\[
    \label{eq:intro:optimal_boost}
    b_\gamma(s) = \frac{1}{\gamma} \log \frac{1}{1 - \exp(-\gamma s)}.
\]
We show that \boost{\gamma} is \emph{strongly tail-optimal}, meaning $C_\boost{\gamma} \leq C_\pi$ for every other scheduling policy~$\pi$ (\cref{thm:strong_tail_optimality}). This solves the open problem of finding a strongly tail-optimal policy, as well as the problem of characterizing the best possible tail constant $\inf_\pi C_\pi$.

Of course, strong tail optimality is a theoretical property, and one would hope that pursuing it as an objective yields a policy with good practical performance. We confirm via simulation that this is indeed the case for \boost{\gamma}. Observe in \cref{fig:intro:performance} that \boost{\gamma}'s improvement over FCFS is often even better than one would predict from the asymptotic tail constants.

Above, we have focused on the case of full job size information, but \boost{} and \boost{\gamma} can also be defined for systems with partial job size information. The specific partial-information model we consider has multiple types of jobs, each with a distinct \emph{label}, and the scheduler knows each job's label, but not its size. In this partial-information setting, a job's boost is a function of its label rather than its size. Our analysis of \boost{}'s tail constant also applies to the partial-information setting (\cref{thm:general_transform}). We also show that partial-information \boost{\gamma} achieves better tail constant than any other \boost{} policy (\cref{sec:partial-information_boost_optimality}) and the previous state-of-the-art, Nudge-M (\cref{sec:nudge_vs_boost}).

\subsection{Key idea: relate strong tail optimality to an easier scheduling problem}

Where does the boost function in \cref{eq:intro:optimal_boost} come from, and how does one show that the resulting \boost{\gamma} policy is strongly-tail optimal? Our key idea is to relate the problem of minimizing the tail constant~$C_\pi$ to a more traditional scheduling problem involving a type of weighted cost.

We begin by considering the following alternative characterization of $C_\pi$, which follows from final value theorem \citep[Theorem~4.3]{grosof_nudge_2021}:
\[
    C_\pi = \lim_{\theta \to \gamma} \frac{\gamma - \theta}{\gamma} \E{\exp(\theta T_\pi)}.
\]
There is thus a vague sense in which minimizing $C_\pi$ is equivalent to minimizing $\E{\exp(\gamma T_\pi)}$. This is only an informal statement because, as one can deduce from \cref{eq:intro:pi_asymptotics}, we have $\E{\exp(\gamma T_\pi)} = \infty$ for all policies~$\pi$, even those that are weakly tail-optimal.

While minimizing the always-infinite quantity $\E{\exp(\gamma T_\pi)}$ is not a well-posed problem in the M/G/1, it is analogous to a well-posed problem in \emph{deterministic single-machine scheduling} \citep{pinedo_scheduling_2016, lenstra_elements_2020}. Consider an arbitrary finite batch of jobs $\calI = \{(a_1, s_1), \dots, (a_n, s_n)\}$. Here $a_i$ is the arrival time of job~$i$, and $s_i$ is its size. Additionally, let $d_{\pi, i}$ be the departure time of job~$i$ under policy~$\pi$, and let the \emph{$\theta$-cost} of policy~$\pi$ be
% \[
    $K_\pi(\theta, \calI) = \sum_{i = 1}^n \exp(\theta (d_{\pi, i} - a_i))$.
% \]
Minimizing $\E{\exp(\gamma T_\pi)}$ is analogous to minimizing $\gamma$-cost $K_\pi(\gamma, \calI)$ in the deterministic setting.

For $\theta < 0$, minimizing $\theta$-cost is actually a variation of a classic single-machine scheduling problem: minimizing total weighted discounted completion time \citep{pinedo_scheduling_2016, lenstra_elements_2020}, where job~$i$'s weight is $\exp(-\theta a_i)$. This problem is hard, but only because of arrival times. In the \emph{batch relaxation}, in which we allow job~$i$ to be served even before time~$a_i$, the optimal policy is an index policy called \emph{Weighted Discounted Shortest Processing Time} (WDSPT) \citep[Theorem~3.1.6]{pinedo_scheduling_2016}. To clarify, the arrival times~$a_i$ still matter in the batch relaxation, because they determine the weights $\exp(-\theta a_i)$.

Because $\gamma > 0$, one can view minimizing $\gamma$-cost as an instance of minimizing total weighted discounted completion time, but with a \emph{negative discount rate}. To the best of our knowledge, this variant of the problem has not been considered in the literature. Nevertheless, essentially the same proof as in the standard positive-discount case shows that a version of WDSPT is optimal in the negative-discount case.\footnote{%
    See, for instance, the interchange argument in \citet[Theorem~3.1.6]{pinedo_scheduling_2016}. We believe this result may be folklore, but we sketch a proof in \cref{sec:batch} for completeness.}
The \boost{\gamma} policy arises from finding a function~$b_\gamma$ such that WDSPT is equivalent to serving jobs in order of increasing boosted arrival time $a_i - b_\gamma(s_i)$.

Above, we have focused on the full-information case, but nearly the same reasoning works in the partial-information case. The difference is that we base the optimal boost function on Weighted Discounted Shortest \emph{Expected} Processing Time (WDSEPT) \citep[Theorem~10.1.3]{pinedo_scheduling_2016} instead of WDSPT.

\subsection{Technical challenge: translating from the batch relaxation to the M/G/1}

The fact that \boost{\gamma} minimizes $\gamma$-cost in the batch relaxation is a promising sign that it is strongly tail-optimal, meaning $C_\boost{\gamma} = \inf_\pi C_\pi$, in the M/G/1. There are two significant obstacles between this intuition and a proof of \boost{\gamma}'s strong tail optimality.

The first obstacle is that the batch relaxation allows jobs to be served at any time, whereas in the M/G/1, jobs cannot be served before they arrive. We therefore need to show that adjusting \boost{\gamma}'s schedule from the batch relaxation to not serve jobs before they arrive does not significantly degrade its performance. If we consider an arbitrarily long sequence of arrivals, this might not be true, so the first step is to figure out how to split up the M/G/1's infinite sequence of arrivals into finite batches. It turns out that using \emph{busy periods} as batches works well. The main technical challenge then becomes showing that ``honest'' \boost{\gamma}, which only serves jobs after they arrive, is nearly as good as ``cheating'' \boost{\gamma}, which is allowed to serve any job in the current busy period, even if it has not arrived yet.

The second obstacle is that minimizing steady-state mean $\gamma$-cost $\E{\exp(\gamma T_\pi)}$ is not a well-posed problem in the M/G/1, because the expectation is infinite for all policies. Instead, we must make do with the fact that for any weakly tail-optimal policy~$\pi$, mean $\theta$-cost $\E{\exp(\theta T_\pi)}$ is finite for all $\theta < \gamma$. We therefore work with $\theta \to \gamma$ limits of $\theta$-cost throughout the paper, as opposed to working directly with $\gamma$-cost. The main technical challenge is to show that \boost{\gamma} is near-optimal for minimizing not just mean $\gamma$-cost but also mean $\theta$-cost, provided $\theta$ is close enough to~$\gamma$.

Above, we have focused on the full-information case, and for good reason: we have not been able to generalize part of this argument to the partial-information case. The issue has to do with a subtle difference between the traditional stochastic batch setting \citep[Section~10.1]{pinedo_scheduling_2016}, which assumes independent job sizes, and the instances that arise from busy periods, which can have subtle dependencies between jobs' sizes (\cref{sec:reduction}). Nevertheless, we show \boost{\gamma} outperforms all other versions of \boost{} (\cref{sec:partial-information_boost_optimality}) and Nudge-M (\cref{sec:nudge_vs_boost}) in the partial-information setting.

\subsection{Contributions}
\label{sec:intro:contributions}

In this work, we present the \emph{first strongly tail-optimal scheduling policy}, namely \boost{\gamma}, for the M/G/1 with light-tailed job size distributions. This solves a significant open problem in queueing and scheduling theory. We also study \boost{} more generally in both theory and simulation, making the following specific contributions:
\* (\cref{sec:model}) We propose \emph{\boost{}}, a new family of scheduling policies that balance the tradeoff between prioritizing short jobs and prioritizing jobs that have been waiting a long time.
\* (\cref{sec:analysis}) We theoretically analyze \boost{}, giving an explicit formula for its tail constant $C_\boost{}$ in terms of the boost function used (\cref{thm:general_transform}).
\* (\cref{sec:batch}) We draw a new connection between minimizing the tail constant in the M/G/1 and a batch scheduling problem with \emph{negative discounting}. We solve the batch scheduling problem using \emph{\boost{\gamma}}, a specific instance of \boost{}.
\* (\cref{sec:optimality}) In the full-information setting, we prove \boost{\gamma} is \emph{strongly tail-optimal} in the M/G/1 with light-tailed job size distribution (\cref{thm:strong_tail_optimality}).
\* (\cref{sec:simulation}) We show in simulation that \boost{\gamma} has excellent practical performance, improving upon FCFS's tail performance by more than 50\% in some cases. We observe that \boost{\gamma}'s performance is robust to using the wrong value of~$\gamma$ or noisily estimated job sizes.
\*/
We also make an observation about strong tail optimality in the heavy-tailed case in \cref{sec:heavy-tailed_strong_optimality}, though it follows nearly immediately from known results.

\section{System model and Boost policies}
\label{sec:model}
% We consider two models in this paper. The primary model is an M/G/1 queue, and most of this section is devoted to describing that. The second model is a \emph{batch scheduling} model, described in \ziv{ref subsection}, which plays a role in our optimality proof. Throughout the paper, notation and terminology refers to the M/G/1 model unless otherwise noted.

% \subsection{Basic M/G/1 modeling assumptions}

We consider an M/G/1 queue with arrival rate $\lambda$, job size distribution~$S$, and load $\rho = \lambda \E{S}$. We make the standard assumption that $\rho < 1$, ensuring stability, and we assume that $S > 0$ almost surely to avoid trivial jobs of size~$0$. We assume that $S$ is light-tailed, considering the following specific class of light-tailed distributions initially identified by \citet{abate_waitingtime_1994}.

\begin{definition}\label{def:light_tailed_dist}
    A distribution $S$ is \emph{class~I} if its moment generating function's leftmost singularity
    \[
        \theta^* = \sup\curlgp{\theta \in \bbR \mid \E{\exp(\theta S) < \infty}},
    \]
    which may be $\infty$, satisfies $\theta^* > 0$ and $\lim_{\theta \to \theta^*} \E{\exp(\theta S)} = \infty$. In informal discussion, ``light-tailed'' is understood to refer to class~I unless otherwise stated.
\end{definition}

The main metric we are concerned with is \emph{response time}, the amount of time between a job's arrival and departure. We denote the response time distribution under scheduling policy~$\pi$ by~$T_\pi$. Thanks to the ``PASTA'' property of Poisson arrivals \citep{wolff_poisson_1982}, we can interpret $T_\pi$ as the response time of a random ``tagged'' job arriving to a steady-state system. We discuss the details of the scheduling policies~$\pi$ we consider in \cref{sec:model:scheduler}.

The quantity that has the largest impact on the system's response time is the \emph{work}, the total remaining processing time of jobs currently in the system. We denote the steady-state amount of work in the M/G/1 by~$W$. This amount is the same under all non-idling (aka work-conserving) scheduling policies.

Our main objective is to find the strongly tail-optimal scheduling policy, defined in \cref{sec:intro:background} and recalled below.

\begin{definition}
    \label{def:tail-optimal}
    A scheduling policy~$\pi$ is \emph{weakly tail-optimal} if there exists finite $c \geq 1$ such that
    \[
        \sup_{\pi'} \limsup_{t \to \infty} \frac{\P{T_\pi > t}}{\P{T_{\pi'} > t}} = c,
    \]
    If additionally $c = 1$, we say $\pi$ is \emph{strongly tail-optimal}.
\end{definition}

The supremum in \cref{def:tail-optimal} ranges over all preemptive scheduling policies $\pi'$ that have access to full information about the sizes and arrival times of all arrivals. In particular, $\pi'$ could in principle use information about \emph{future} arrivals. However, none of the policies we consider use this information (aside from the ``cheating'' policy introduced in \cref{sec:model:cheating}), and we achieve strong tail optimality without it.

\subsection{Asymptotic tails}
\label{sec:model:tails}

The key property of class~I distributions is that they ensure that the work~$W$ has asymptotically exponential tail. Specifically, there exist constants $\gamma > 0$ and $C_W > 0$ such that \citep[equation~(2)]{grosof_nudge_2021}\footnote{%
    Throughout the paper, $\theta \to \gamma$ limits are understood as being limits from below, seeing as $\E{\exp(\theta W)} = \infty$ for all $\theta \geq \gamma$.%
}
\[
    \label{eq:model:tail_constant_work}
    C_W = \lim_{t \to \infty} \exp(\gamma t) \P{W > t} = \lim_{\theta \to \gamma} \frac{\gamma - \theta}{\gamma} \E{\exp(\theta W)},
\]
with the equivalence of the two limits being due to final value theorem \citep[Theorem~4.3]{grosof_nudge_2021}. We call $\gamma$ the \emph{decay rate} and $C_W$ the \emph{tail constant} of the work distribution~$W$. It is known that $\gamma$ is the least positive real solution to
\[
    \label{eq:model:gamma}
    \gamma = \lambda(\E{\exp(\gamma S)} - 1).
\]
When $S$ is class~I, $\gamma$ is a simple pole of $W$'s moment generating function $\theta \mapsto \E{\exp(\theta W)}$ \citep{abate_waitingtime_1994, abate_asymptotics_1997}, regardless of the arrival rate~$\lambda$. Our results likely generalize to other combinations of $\lambda$ and~$S$ for which this is the case.

We define the \emph{tail constant of scheduling policy~$\pi$}, denoted $C_\pi$, in the same way as the tail constant of the work distribution:\footnote{%
    The limits below may not exist, so strictly speaking, we should define lower and upper constants using $\liminf_{t \to \infty}$ and $\limsup_{t \to \infty}$ in place of $\lim_{t \to \infty}$. But the limits exist for all policies we consider, so we omit this additional complexity.%
}
\[
    \label{eq:model:tail_constant}
    C_\pi = \lim_{t \to \infty} \exp(\gamma t) \P{T_\pi > t} = \lim_{\theta \to \gamma} \frac{\gamma - \theta}{\gamma} \E{\exp(\theta T_\pi)}.
\]
As an example, FCFS's tail constant is easily shown to be
\[
    C_\fcfs = C_W \E{\exp(\gamma S)},
\]
where finiteness of $\E{\exp(\gamma S)}$ follows from~\cref{eq:model:gamma}.

\subsection{Scheduling model and what information the scheduler has}
\label{sec:model:scheduler}

We consider both \emph{nonpreemptive} scheduling, where once a job begins service, it will complete without interruption, and \emph{preemptive} scheduling, where jobs may be paused in the middle of service. In the latter case, we assume a standard preempt-resume model in which jobs may be paused and resumed without delay, overhead, or loss of progress.

We wish to study both the \emph{full-information} setting, in which the scheduler learns each job's exact size (aka service time) when it arrives; as well as \emph{partial-information} settings, in which the scheduler has some limited but incomplete information about job sizes.
For instance, perhaps there are two types of arrivals, each with its own size distribution, but we do not know the size of any particular job.\footnote{%
    If there is no information at all to distinguish different jobs from one another, then, at least among nonpreemptive policies, there is no way to improve upon FCFS, due to $T_\fcfs$ being minimal in the convex order \citep{shanthikumar_convex_1987}. Investigating whether preemptive policies could improve upon FCFS in this setting is an interesting future direction.%
}

To capture a wide range of information settings, we use the flexible \emph{label-size pair} model from recent work on M/G/1 scheduling \citep{scully_new_2022, scully_soap_2018}. In this model, each job has an i.i.d. pair $(L, S)$ of a \emph{label}~$L$ and size~$S$. The space of possible labels, denoted~$\bbL$, can be arbitrary, and there may be an arbitrary joint distribution between labels and sizes. For example:
\* To model known job sizes, let $\bbL = \bbR_+$ and $L = S$.
\** We call this case the \emph{full-information} setting.
\* To model a scenario with two types of jobs A and~B, where job types are known but job sizes are unknown, let $\bbL = \{\mathrm{A}, \mathrm{B}\}$, and define the joint distribution such that, for instance, $(S \given L = \mathrm{A})$ is the size distribution of type~A jobs.
\** We call any case where $L \neq S$ with positive probability, of which the above is one example, the \emph{partial-information} setting.
\*/
Of course, one can imagine more complicated label-size pair distributions. As a final example, perhaps some jobs are labeled with their exact size, while others are labeled only type~A or type~B. This can be modeled using $\bbL = \bbR_+ \cup \{\mathrm{A}, \mathrm{B}\}$.

We assume that the scheduler has access to each job's arrival time and label, but that it has no information about each job's size beyond what can be deduced from its label. That is, if a job is labeled~$l$, the scheduler knows its size is distributed as $(S \given L = l)$, but it does not learn the realization until the job is complete.

\subsection{Defining the Boost family of policies}

We introduce a new scheduling policy called \emph{\boost{}}. Strictly speaking, \boost{} is a family of scheduling policies, where an instance of the family is determined by a \emph{boost function} $b : \bbL \to \bbR_+$.\footnote{%
    One can in principle allow negative boosts, but we focus our analysis on the case where boosts are always nonnegative.%
}
The boost function maps each label $l \in \bbL$ to a quantity $b(l)$ called the \emph{boost} of a job with label~$l$.

\boost{} operates as follows. Suppose a job with label~$l$ has arrival time~$a$. We define the job's \emph{boosted arrival time} to be $a - b(l)$. \boost{} uses the same basic rule with any boost function: \emph{prioritize jobs in order from least to greatest boosted arrival time}.\footnote{%
    To clarify, at any moment in time, \boost{} is only aware of jobs whose arrival time~$a$ is in the past. There may be jobs that will arrive in the future with boosted arrival time $a - b(l)$ in the past. But \boost{} is not aware of these jobs yet and will not serve them prior to their arrival. \boost{} thus does not require knowledge of future arrival times to implement.}
As a trivial example, choosing $b(l) = 0$ reduces \boost{} to FCFS.

One can define preemptive or nonpreemptive versions of \boost{}. The preemptive version makes scheduling decisions continuously, always serving the job of least boosted arrival time. The nonpreemptive version, after it serves the first job in each busy period, makes scheduling decisions whenever a job completes, each time choosing the job of least boosted arrival time. One can also define intermediate versions where a job may be preempted by some, but not necessarily all, arrivals with lower boosted arrival time.

All of our theoretical results hold for the preemptive, nonpreemptive, and intermediate versions of \boost{} (\cref{rmk:preemption_irrelevant}). As such, we leave the exact preemption rule unspecified throughout our theoretical results. But for concreteness, the reader may safely imagine that nonpreemptive \boost{} is used throughout, and we use nonpreemptive \boost{} in our simulations (\cref{sec:simulation}).

There is one family of boost functions that is especially important, as they result in strong tail optimality.

\begin{definition}
    \label{def:boost_theta}
    For any $\theta > 0$, the \emph{\boost{\theta}} policy for label-size pair distribution $(L, S)$ is the version of \boost{} with the following boost function, which we call the \emph{$\theta$-optimal boost function}:
    \[
        b_\theta(l) = \frac{1}{\theta} \log \frac{\E{\exp(\theta S) \given L = l}}{\E{\exp(\theta S) \given L = l} - 1}.
    \]
    While the definitions of \boost{\theta} and $b_\theta$ depend on the label-size pair distribution $(L, S)$, we leave this implicit in our notation.
\end{definition}

The \boost{\gamma} policy alluded to in \cref{sec:intro} is simply \boost{\theta} with $\theta = \gamma$. In the full-information case where $L = S$, we have $\E{\exp(\gamma S) \given L = s} = \exp(\gamma s)$, and so $b_\gamma$ reduces to the formula in \cref{eq:intro:optimal_boost}.

We use the name ``\boost{}'' when referring to a version with generic boost function~$b$, and we use the name ``\boost{\theta}'' when referring to a version using the $\theta$-optimal boost function~$b_\theta$, with $\theta = \gamma$ being the most important case.

If one uses an overly aggressive boost function, such as boosting small jobs too much, then \boost{} may not be weakly tail-optimal, let alone strongly tail-optimal. Our results (\cref{thm:general_transform}) show that as long as
\[
    \label{eq:model:boost_okay}
    \E{b(L) (\exp(\gamma S) - 1)} < \infty,
\]
then \boost{} is indeed weakly tail-optimal. There are two important special cases where \cref{eq:model:boost_okay} holds.
\* It always holds if one uses the $\gamma$-optimal boost function~$b_\gamma$ (\cref{lem:optimal_work_crossing_finite}).
\* In the full-information case where $L = S$, one can show using finiteness of $\E{\exp(\gamma S)}$ that \cref{eq:model:boost_okay} holds if $b(s) \leq O(1/s)$ in the $s \to 0$ limit and $b(s) \leq O(1)$ in the $s \to \infty$ limit. The intuition behind why we need such a condition is that if small jobs have too large of a boost, then each large job could be overtaken by so many small jobs that we lose weak tail optimality.
\*/

To reduce clutter, we let $B = b(L)$ be the boost of a random job, meaning we have a joint distribution of boost-size pairs $(B, S)$. We write $B_\theta$ when using the $\theta$-optimal boost function~$b_\theta$.

\subsection{Lower bounding tool: ``cheating'' version of Boost}
\label{sec:model:cheating}

In order to prove that \boost{\gamma} is strongly tail-optimal, we need a lower bound on the possible tail constant~$C_\pi$ (\cref{sec:model:tails}) achievable by any policy~$\pi$. Our main tool for doing so is to define a ``cheating'' version of \boost{}, which we call \emph{\cheat{}}.

Like \boost{}, \cheat{} is defined by a boost function~$b$, it assigns each job a boosted arrival time in the same way, and it also prioritizes jobs in order from least to greatest boosted arrival time. The difference is that we allow \cheat{} to \emph{serve arrivals from the future}. Specifically, we allow \cheat{} to serve any job that will arrive in the current busy period, where a \emph{busy period} is a maximal interval of time during which the server is busy \citep{harchol-balter_performance_2013}.\footnote{%
    This definition is assuming a non-idling scheduling policy, where all such policies lead to the same notion of busy periods.%
}

We can view \cheat{} as being essentially the same policy as \boost{} for a modified ``cheating'' M/G/1, which differs from the standard M/G/1. To describe the cheating M/G/1, we distinguish between a job's \emph{arrival time}, which is the time it arrives in the standard system, and a job's \emph{release time}, which is the earliest moment in time at which a job is allowed to be served by the scheduler.
\* In the standard M/G/1, a job's release time is its arrival time.
\* In the cheating M/G/1, a job's release time is the start of the busy period containing its arrival time.
\*/
In the cheating M/G/1, we still define a job's response time to be its departure time minus its arrival time. The difference is that a job's departure time may now be less than its arrival time plus its size, so a job's response time may be less than its size, or even negative.

To clarify, for a given arrival sequence, the standard and cheating M/G/1 systems have the same busy periods. That is, one can imagine first deciding what the busy periods are using the standard M/G/1, and then ``retroactively'' moving release times to construct the corresponding cheating M/G/1. Note that this implies the cheating M/G/1 is ergodic, because it inherits the renewal cycles of the standard M/G/1.

With the above distinction in mind, strictly speaking, \cheat{} is best thought of as the same policy as \boost{} but for a modified system, namely the cheating M/G/1, as opposed to a different scheduling policy for the standard system. But for the purposes of notation, we treat it like a different scheduling policy. For instance, we denote the response time distribution of \cheat{} by~$T_\cheat{}$. As a reminder, due to how the cheating system works, we can have $T_\cheat{} < 0$ with positive probability.

We define \emph{\cheat{\theta}} analogously to \boost{\theta} (\cref{def:boost_theta}), namely as the version of \cheat{} using the $\theta$-optimal boost function~$b_\theta$. The significance of \cheat{\theta} is that, as we show in \cref{thm:optimal_batch_policy}, $\E{\exp(\theta T_\cheat{\theta})} \leq \E{\exp(\theta T_\pi)}$ for any policy~$\pi$ for the standard M/G/1.

\section{Analysis of Boost's Tail Constant}
\label{sec:analysis}

In this section, we analyze the tail constants of \boost{} and \cheat{} in terms of the boost function~$b$ and the system model parameters, namely $\lambda$ and $(L, S)$.

Our main result is the following characterization of $C_\boost{}$ and $C_\cheat{}$. Interestingly, we find that cheating has no effect on the tail constant. See \cref{pf:general_transform} for the proof.

\begin{theorem}\label{thm:general_transform}
    Consider an M/G/1 with class~I job size distribution, and consider the \boost{} policy with a fixed boost function~$b$. If \cref{eq:model:boost_okay} holds, then \boost{} and \cheat{} both have tail constant
    \[
        C_\boost{} = C_\cheat{} = C_W \E[\big]{\exp\gp{\gamma \gp{S - b(L)}}} \exp\gp[\big]{\lambda \E{b(L) \, (\exp(\gamma S) - 1)}}.
    \]
    In particular, $C_\boost{} < \infty$, so \boost{} is weakly tail-optimal.
\end{theorem}

% We prove \cref{thm:general_transform} in \cref{pf:general_transform}. This comes after describing our overall approach and doing some preliminary analysis. Throughout, recall that we write $B = b(L)$, and we often work with the distribution of $(B, S)$ pairs.

\subsection{Approach: tagged job analysis}
\label{sec:analysis:tagged_job}

To derive the tail constants of \boost{} and \cheat{}, we require bounds on $T_\boost{}$ and $T_\cheat{}$. We obtain these by considering a pair of M/G/1 systems, one standard and one cheating, experiencing the same arrival process. The only difference is that the release times in the standard M/G/1 coincide with arrival times, whereas in the cheating M/G/1, all jobs in a given busy period have release time at the start of the busy period (\cref{sec:model:cheating}). We assume both systems are stationary processes.

We use a \emph{tagged job} analysis \citep{harchol-balter_performance_2013}, which is a common technique in analyzing complex scheduling policies \cite{scully_soap_2018, scully_soap_2018a, grosof_srpt_2018, fajardo_controlling_2015, stanford_waiting_2014, fajardo_waiting_2017, scully_new_2022, grosof_nudge_2021, vanhoudt_stochastic_2022, grosof_optimal_2023, charlet_tail_2024}. A typical tagged job analysis uses the observation that, due to the ``PASTA'' property of Poisson arrivals \citep{wolff_poisson_1982}, to analyze a policy~$\pi$'s response time distribution~$T_\pi$, we can analyze the response time of a single ``tagged'' job which arrives at an arbitrary time, such as time~$0$. We then interpret $T_\pi$ as the random response time the tagged job experiences. There are three \emph{independent} sources of randomness that contribute to $T_\pi$:
\* The tagged job's label-size pair, which is drawn from $(L, S)$.
\* Aspects of the M/G/1's state, such as its work~$W$, at time~$0$, which is drawn from the system's stationary distribution.
\* Arrivals that occur after time~$0$.
\*/

In our approach, instead of assuming without loss of generality that the tagged job's arrival time is time~$0$, we assume the \emph{boosted arrival time} is time~$0$. The tagged job's response time is still determined by the same three sources of randomness listed above. However, the interpretation of the latter two sources changes, e.g., some of the arrivals after time~$0$ arrive before the tagged job.

The only subtlety to check is that the system state at time~$0$, and in particular the work~$W$, still has the stationary distribution. This is indeed the case. Consider the stationary work process $W_t$ as a function of time $t \in \bbR$. The key observation is that the tagged job's boost is independent of its arrival time. So if the tagged job arrives at time~$a$ and has boost $B = b(L)$, the work at the boosted arrival time $W_{a - B}$ is distributed according to the stationary distribution, because $\{W_t\}_{t \leq a}$ is stationary and independent of~$B$. As such, we can imagine $a - B = 0$ without loss of generality. One can use the framework of Palm calculus \citep{baccelli_elements_2003} to formalize this argument.

To summarize our approach and notation:
\* We analyze the response time of a tagged job with boosted arrival time~$0$.
\** Abusing notation slightly, we write $L$, $S$, $B = b(L)$, and $T_\pi$ for the tagged job's label, size, boost, and response time under policy~$\pi$.
\** This means the tagged job's arrival time is~$B$.
\* The amount of work in the (standard) M/G/1 at time~$0$, denoted~$W$, is distributed according to the standard stationary work distribution \citep{harchol-balter_performance_2013}.
\** We follow the convention that $W$ does not include the tagged job's size~$S$.
\* After time~$0$, new Poisson arrivals occur at rate~$\lambda$ with iid label-size pairs. 
\** To avoid ambiguity, when we need to discuss the label, size, and boost of a generic future arrival, we write $L'$, $S'$, and~$B'$ instead of (the identically distributed) $L$, $S$, and~$B$.
\** We call these arrivals ``new'' because their arrival times are after time~$0$, even if their arrival times are before the tagged job's arrival time~$B$.
\*/

\subsection{Bounding Boost's response time using crossing work}

% One system is a standard M/G/1 that employs the \boost{} policy, and the other one is the cheating M/G/1 system, with the same arrivals of label-size pairs, where a job with label-size pair~$(l, s)$ and arrival time $a$ in the \boost{} system has release time $\max(a - b(s), R)$ in the cheating system, i.e. the maximum of its boosted arrival time and its release time.

A critical step in bounding both $T_\cheat{}$ and $T_\boost{}$ is quantifying how much arriving work will ``boost past'' the tagged job, which we define formally below. We also define the complementary quantity for the arriving work that doesn't boost past the tagged job.
% We analyze this quantity, which we call the \emph{crossing work}, and define it formally below.

\begin{definition}\label{def:work_crossing}
    \leavevmode
    \*[subenv] The \emph{crossing work} arriving in~$(0, u)$, denoted~$V(u)$, is the amount of work due to jobs that have arrival time in $(0, u)$ and boosted arrival time in $(-\infty, 0]$.
    \* The \emph{non-crossing work} arriving in~$(0, u)$, denoted $\bar{V}(u)$, is the amount of work due to jobs that have arrival time in $(0, u)$ and boosted arrival time in $(0, \infty)$.
    \*/
\end{definition}

For example, $V(\infty)$ is all the work that ``boosts past'' time~$0$, meaning arriving after~$0$ but having boosted arrival time before~$0$. On the other extreme, $V(0)$ is simply~$0$.

Understanding the amount of crossing work that the tagged job experiences is key to understanding the response time of the tagged job under \boost{}. However, the exact crossing work is difficult to compute, as it depends on the amount of work $W$ at the tagged job's boosted arrival time, its size, as well as the sizes and labels of future arrivals. In particular, jobs with boosted arrival time before time~$0$ may still depart after the tagged job, if their arrival times are late enough. Under \cheat{}, since the release time of all jobs is the beginning of the busy period, the crossing work is given by $V(Z)$, where $Z$ is the random time denoting the end of the busy period. But $Z$ in turn depends on the same arrivals that $V(Z)$ is counting, making this a hard quantity to analyze.
% \george{Not sure if we need to define $I$ here. Maybe there's a cleaner way to say this?} but the relationship to \boost{}'s experience is unclear.

Due to these difficulties, instead of computing crossing work exactly, we find bounds on $T_\boost{}$ and $T_\cheat{}$ which are good enough for computing $C_\boost{}$ and $C_\cheat{}$.

% \ziv{The hard questions are:
% \* How much work that arrives after the tagged job's boosted arrival time gets to go in front of it?
% \* If there is $W$ at the tagged job's boosted arrival time, how much work will there be at the tagged job's arrival time? Under \cheat{}, how much of this work does the tagged job get prioritized over?
% \*/
% These are hard, which is why we don't get an exact MGF formula. Instead, we find bounds that are good enough for computing $C_\pi$ in the next lemma.}

% Observe that the crossing work is defined relative to a start time. % We first observe the following relationship between work in the \boost{} and cheating systems:

% \begin{lemma}\label{lem:work_cheating>work_legit}
%     Let $W$ be the amount of work in the \boost{} system at time $0$. Then the cheating system has work at least $W$ at time $0$.
% \end{lemma}

% \begin{proof}
%     Consider the busy period that time $0$ falls in for the \boost{} system, and let it have work $W$ at this time. This work is exactly the difference between all the work that has arrived so far in this busy period and the service attained. Since all the work arrived has arrival time earlier than $0$, it must have arrived already in the cheating system as well. Because this work arrives at the earliest at the release time for this busy period, and both systems are non-idle, the cheating system must have work at least $W$ at this time.
% \end{proof}

% Applying the observation in \cref{lem:work_cheating>work_legit}, we derive the following bounds for the response time of our tagged job:

\begin{lemma}\label{lem:boost_response_time_bounds}
    The tagged job's response times under \boost{} and \cheat{} are both lower-bounded by
    \[
        T_\boost{}, T_\cheat{} &\geq W - B + V(W) + S
    \]
    and, for all $u \geq 0$, upper-bounded by
    \[
        T_\boost{}, T_\cheat{} &\leq (W - \min\{B, u\})^+ + V(\infty) + S + \bar{V}(u) \, \1(W < \min\{B,u\}).
    \]
\end{lemma}

\begin{proof}
    We first recall our conventions from \cref{sec:analysis:tagged_job}. The tagged job's boosted arrival time is~$0$, true arrival time is~$B$, and departure time is $T_\pi + B$. There is $W$ work in the system at time~$0$.
    % Additionally, let $U_\pi$ be the time at which the tagged job enters service under policy~$\pi$.

    We first treat both bounds on $T_\boost{}$ before handling both bounds on $T_\cheat{}$. For $T_\boost{}$, we analyze the work done by the server between the tagged job's boosted arrival time~$0$ and its departure time $T_\boost{} + B$. Both the lower and upper bounds use the fact that the server must complete the following work between $0$ and $T_\boost{} + B$:
    \* $W$ from work present at time~$0$. This must be completed because all the work present at time~$0$ has priority over the tagged job. This is because their arrival times, and thereby boosted arrival times, are earlier than~$0$.
    \* $S$ from the tagged job itself.
    \* Some additional work from new arrivals that occur after time~$0$.
    \*/
    The main task is thus to bound the work from new arrivals.

    \paragraph{Lower bound on $T_\boost{}$.}
    By the first bullet above, work on $W$ must complete before the tagged job can enter service. Therefore, it cannot enter service before time $W$. In this time, there is at least $V(W)$ work from new arrivals that will have priority over the tagged job. Adding the required components, we have $T_\boost{} + B \geq W + V(W) + S$, as desired.

    \paragraph{Upper bound on $T_\boost{}$.}
    First, observe that decreasing the tagged job's boost from $B$ to $\hat{B} = \min\{B, u\}$ can only increase its response time, so it suffices to analyze the case with this reduced boost. In the remainder of this argument, $\hat{B}$ plays the role of~$B$.

    We first consider the fully preemptive case, in which the tagged job has priority over all jobs with boosted arrival time later than~$0$. In addition to the required completions of $W$ and $S$ before the tagged job's departure, it will need to complete at most $V(\infty)$ from new arrivals that occur after time~$0$.

    However, we must also account for time the server spends idle. No idling occurs after the tagged job arrives, but the server may be idle for some time during $[0, \hat{B}]$. This idle time is at most $(\hat{B} - W)^+$, so in the fully preemptive case,
    \[
        \label{eq:boost_bound_preemptive}
        T_\boost{} + \hat{B} \leq W + V(\infty) + S + (\hat{B} - W)^+.
    \]

    We now turn to the case where the server is not fully preemptive. The only change to the argument is that the tagged job may also have to wait for the remaining work of a lower-priority job, if one is in service when the tagged job arrives at time~$\hat{B}$. If there is such a job, let $R$ be its total size, and otherwise, let $R = 0$. It is clear that the tagged job's departure time is at most the right-hand side of \cref{eq:boost_bound_preemptive} plus~$R$. We bound $R$ with two observations:
    \* If $W \geq \hat{B}$, then at time~$\hat{B}$ when the tagged job arrives, the server is still working on jobs with priority over the tagged job, either from $W$ or from new arrivals. This means $R > 0$ only if $W < \hat{B}$.
    \* If $R > 0$, then a job with boosted arrival time after~$0$ is in service at time~$\hat{B}$. The job's true arrival time is therefore in $(0, \hat{B})$, so $R \leq \bar{V}(\hat{B}) \leq \bar{V}(u)$ (\cref{def:work_crossing}).
    \*/
    Combining these observations implies
    \[
        R \leq \bar{V}(u)\,\1(W < \hat{B}).
    \]
    Adding this to the fully preemptive bound in \cref{eq:boost_bound_preemptive} yields
    \[
        T_\boost{} + \hat{B} \leq W + V(\infty) + S + (\hat{B} - W)^+ + \bar{V}(u)\,\1(W < \hat{B}),
    \]
    which rearranges to the desired bound.

    % We now proceed to bound $T_\cheat{}$. In the lower bound, we show that the tagged job cannot begin service until sufficiently long after its boosted arrival time~$0$. In the upper bound, we are able to reduce \cheat{} to \boost{}.

    \paragraph{Lower bound on $T_\cheat{}$.}
    Recall throughout that $W$ refers to the amount of work in the standard M/G/1 at time~$0$, whose busy periods affect the release times in the cheating M/G/1 (\cref{sec:model:cheating}). We will show that the tagged job cannot begin service before at least time $W + V(W)$, which implies the desired bound. To do this, we first analyze the busy period related to $W$, the work present at time~$0$ in the system. This work belongs to a busy period, which we call BP, that started at time $-A < 0$. Let $U$ be the amount of work from jobs with arrival time in $[-A, 0)$. Then we have
    \[
        U - A = W.
    \]
    Therefore, the server will be busy during $[-A, -A + U] = [-A, W]$. As such, at least $V(W)$ work from new arrivals will arrive during $[0, W]$. This means there is at least $U + V(W)$ total work in the busy period with boosted arrival time at most~$0$.

    There are now two cases to consider for the tagged job. First, suppose the system is busy for all of $[0, B]$. Then the tagged job belongs to the same busy period BP as described above. In this period we know that there is at least $U + V(W)$ total work with priority over the tagged job, so it cannot begin service prior to
    \[
        -A + U + V(W) = W + V(W).
    \]

    Second, suppose that the system becomes idle at some time in $[0, B]$. This means that the tagged job belongs to a busy period after BP. Therefore the tagged job's release time, which is the earliest it can enter service, must be after the end of BP. But we know that BP ends at the earliest at $-A + U + V(W) = W + V(W)$, so the tagged job's release time must be at least $W + V(W)$.

    \paragraph{Upper bound on $T_\cheat{}$.}
    First, as in the $T_\boost{}$ upper bound, we reduce the tagged job's boost to $\hat{B} = \min\{B, u\}$. We then imagine a variant of \cheat{} that treats the tagged job especially poorly, forcing it to begin service at no earlier than its arrival time~$\hat{B}$. This clearly only increases the tagged job's response time. From here, the reasoning from the fully preemptive $T_\boost{}$ upper bound also applies to $T_\cheat{}$.
\end{proof}

\begin{remark}
    \label{rmk:preemption_irrelevant}
    The proof of \cref{lem:boost_response_time_bounds} above works regardless of whether we are considering a preemptive, nonpreemptive, or intermediate version of \boost{}. As such, our results apply regardless of the precise preemption rule used.
\end{remark}

\Cref{lem:boost_response_time_bounds} allows us to bound both $T_\boost{}$ and $T_\cheat{}$ with quantities that only depend on the standard M/G/1. It then suffices to show that the $\liminf$ and $\limsup$ of the lower and upper bounds, respectively, converge to the same number, namely the expression given in \cref{thm:general_transform}. To do so, we require that the crossing work terms in our bounds have finite moment generating function. This holds under our assumption \cref{eq:model:boost_okay}, which says, roughly speaking, that the boost isn't too large for too many labels.

\begin{lemma}\label{lem:work_crossing_expression}
    If \cref{eq:model:boost_okay} holds, then for all $u \in \bbR_+ \cup \{\infty\}$,
    \[
        \E{\exp\gp{\theta V(u)}} = \exp\gp[\big]{\lambda \E[\big]{\gp{\exp\gp{\theta S'} - 1} \min\{B', u\}}} < \infty.
    \]
\end{lemma}

\begin{proof}
    We can consider each new arrival to be a triple $(b', s', t)$ representing its boost~$b'$, size~$s'$, and arrival time~$t$. Let $X$ be the random set of triples corresponding to arrivals after time~$0$. We can write the crossing work $V(u)$ as
    \[
         V(u) = \sum_{(b', s', t)\in X} s'\1\gp{t \leq \min\curlgp{b', u}}.
    \]
    To compute $\E{\exp(\theta V(u))}$, we use Campbell's theorem \citep[Section~3.2]{kingman_poisson_1993} for the Laplace functional of a Poisson point process. In our case, the point process is~$X$, and its intensity measure is
    \[
        \mu(\calB \times \calS \times \d{t}) = \P{B' \in \calB, S' \in \calS} \, \lambda \d{t}.
    \]
    Campbell's theorem and a brief computation involving Tonnelli's theorem then imply\footnote{%
        The precondition of Campbell's theorem \citep[Section~3.2]{kingman_poisson_1993} is satisfied when the right-hand-side expression is finite, as $\exp\gp{\theta s' \1\gp{t \leq \min{b', u}}} - 1 \geq \theta s' \1\gp{t \leq \min\gp{b', u}}$.%
    }
    \[
        \E{\exp\gp{\theta V(u)}}
        &= \exp\gp[\bigg]{\int \gp[\big]{\exp\gp[\big]{\theta s' \1\gp{t \leq \min \curlgp{b', u}}} - 1} \, \mu(\d(b', s', t))}\\
        &= \exp\gp[\bigg]{\int \gp{\exp\gp{\theta s'} - 1} \, \1\gp{t \leq \min \curlgp{b', u}} \, \mu(\d(b', s', t))}\\
        &= \exp\gp[\bigg]{\E[\bigg]{\int_0^\infty \gp{\exp\gp{\theta S'} - 1} \, \1\gp{t\leq\min\curlgp{B', u}} \, \lambda \d{t}}}\\
        &= \exp\gp[\big]{\lambda\E[\big]{\gp{\exp(\theta S') - 1} \min\curlgp{B', u}}}.
        \qedhere
    \]
\end{proof}

\subsection{Tail Constant of Boost}

% With these lemmas, we can now prove \cref{thm:general_transform}.

\begin{proof}[Proof of \cref{thm:general_transform}]
    \label{pf:general_transform}
    Let $\pi$ be one of \boost{} or \cheat{}, and let the claimed tail constant be
    \[
        C &= C_W \E[\big]{\exp\gp{\gamma \gp{S - B}}} \exp\gp[\big]{\lambda \E{B (\exp(\gamma S) - 1)}}
        = C_W \E[\big]{\exp\gp{\gamma \gp{S - B}}} \, \E{\exp \gamma V(\infty)},
    \]
    where \cref{lem:work_crossing_expression} implies the second equality. We will show $C_\pi = C$. By \cref{lem:boost_response_time_bounds}, for all $u \geq 0$,
    \[
        \E{\exp\gp{\theta T_\pi}} &\geq \E[\big]{\exp\gp[\big]{\theta \gp[\big]{W - B + V(W) + S}}}, \\
        \E{\exp\gp{\theta T_\pi}} &\leq \E[\big]{\exp\gp[\big]{\theta \gp[\big]{(W - \min\{B,u\})^+ + V(\infty) + S + \1(W < \min\{B,u\})\bar{V}(u)}}}.
    \]
    % both $T_\boost{}$ and $T_\cheat{}$ can be lower and upper bounded, respectively, by $W - B + V(W) + S$ and $(W - B)^+ + S + V(\infty)$. This means that both $\E*{\exp\gp*{\theta T_\boost{}}}$ and $\E*{\exp\gp*{\theta T_\cheat{}}}$ fall within
    % \[
    %     \sqgp[\Big]{\E[\big]{\exp\gp[\big]{\theta \gp[\big]{W - B + V(W) + S}}}, \E[\big]{\exp\gp[\big]{\theta \gp[\big]{(W - B)^+ + V(\infty) + S}}}}.
    % \]
    \Cref{eq:model:tail_constant} thus implies
    \[
        \gamma C_\pi &\geq \liminf_{\theta\to\gamma}{} (\gamma - \theta) \E[\big]{\exp\gp[\big]{\theta \gp[\big]{W - B + V(W) + S}}}, \\
        \gamma C_\pi &\leq \limsup_{\theta\to\gamma}{} (\gamma - \theta) \E[\big]{\exp\gp[\big]{\theta \gp[\big]{(W - \min\{B,u\})^+ + V(\infty) + S + \1(W < \min\{B,u\})\bar{V}(u)}}}.
    \]
    It thus suffices to show that the lower bound is at least~$\gamma C$, and that the infimum over~$u$ of the upper bound is at most~$\gamma C$.

    We first analyze the lower bound $\gamma C_\pi$. For all $w \geq 0$, we have
    \[
        \MoveEqLeft
        \liminf_{\theta\to\gamma}{} (\gamma - \theta) \E[\big]{\exp\gp[\big]{\theta \gp[\big]{W - B + V(W) + S}}}\\
        &\geq \liminf_{\theta\to\gamma}{} (\gamma - \theta) \E*{\exp\gp*{\theta \gp*{W - B + V(W) + S}}\1\gp{W > w}}\\
        &\geq \liminf_{\theta\to\gamma}{} (\gamma - \theta) \E*{\exp\gp*{\theta \gp*{W - B + V(w) + S}}\1\gp{W > w}}\\
        &= \liminf_{\theta\to\gamma}{} (\gamma - \theta) \E*{\exp\gp*{\theta W}\1\gp{W > w}} \, \E*{\exp\gp*{\theta \gp{S - B}}} \, \E*{\exp{\theta V(w)}}\\
        &= \gamma C_W\E*{\exp\gp{\gamma\gp{S - B}}} \, \E*{\exp{\gamma V(w)}},
    \]
    where the last step follows from \cref{eq:model:tail_constant_work}. Since this holds for all~$w$, it also holds in the $w \to \infty$ limit. The monotone convergence theorem implies the limit is $\gamma C$, as desired.

    We now turn to the upper bound on $\gamma C_\pi$. We have
    \[
        \MoveEqLeft
        \limsup_{\theta\to\gamma}{} (\gamma - \theta) \E[\big]{\exp\gp[\big]{\theta \gp[\big]{(W - \min\{B,u\})^+ + V(\infty) + S + \bar{V}(u)\,\1(W < \min\{B, u\})}}} \\
        &= \limsup_{\theta\to\gamma}{} (\gamma - \theta) \E[\big]{\exp\gp[\big]{\theta \gp[\big]{W - \min\{B,u\} + V(\infty) + S}} \, \1(W \geq \min\{B,u\})}\\
        &\quad + \limsup_{\theta\to\gamma}{} (\gamma - \theta) \E[\big]{\exp\gp[\big]{\theta \gp[\big]{V(\infty) + S + \bar{V}(u)}} \, \1(W < \min\{B,u\})}\\
        &\leq \limsup_{\theta\to\gamma}{} (\gamma - \theta) \E[\big]{\exp\gp[\big]{\theta \gp[\big]{W - \min\{B,u\} + V(\infty) + S}}} \label{eq:almost_there:first} \\
        &\quad + \limsup_{\theta\to\gamma}{} (\gamma - \theta) \E[\big]{\exp\gp[\big]{\theta \gp[\big]{V(\infty) + S + \bar{V}(u)}}}. \label{eq:almost_there:second}
        % &\leq \limsup_{\theta\to\gamma} (\gamma - \theta) \E[\big]{\gp*{\exp\gp*{\theta \gp{W - \min\{B,u\}}} + 1}\exp\gp*{\theta \gp{V(\infty) + S + \1(W < \min\{B, u\})\bar{V}(u)}}}\\
        % &\leq \limsup_{\theta\to\gamma} (\gamma - \theta) \E[\big]{\exp\gp*{\theta \gp{W - \min\{B,u\} + V(\infty) + S + \1(W < \min\{B,u\})\bar{V}(u)}}}\\
        %     &\;\;\;\;+ \limsup_{\theta\to\gamma} (\gamma - \theta) \E[\big]{\exp\gp*{\theta\gp{V(\infty) + S + \1(W < \min\{B,u\})\bar{V}(u)}}}.
    \]
    We now compute the limits in \cref{eq:almost_there:first, eq:almost_there:second}. We first show the limit in \cref{eq:almost_there:second} vanishes, then we show the limit in \cref{eq:almost_there:first} has the desired value.

    Let $A(u) = V(u) + \bar{V}(u) \geq \bar{V}(u)$ be the total amount of work that arrives during~$(0, u)$. The limit in \cref{eq:almost_there:second} is bounded by
    \[
        \MoveEqLeft
        \limsup_{\theta\to\gamma}{} (\gamma - \theta) \E*{\exp\gp*{\theta\gp[\big]{V(\infty) + S + \bar{V}(u)}}}\\
        &= \limsup_{\theta\to\gamma}{} \E{\exp(\theta V(\infty))} \, \E{\exp(\theta S)} \, \E{\exp(\theta \bar{V}(u))}\\
        &\leq \limsup_{\theta\to\gamma}{} (\gamma - \theta) \E{\exp(\theta V(\infty))} \, \E{\exp(\theta S)} \, \E{\exp(\theta A(u))}\\
        &\leq \limsup_{\theta\to\gamma}{} (\gamma - \theta) \E{\exp(\gamma V(\infty))} \, \E{\exp(\gamma S)} \, \E{\exp(\gamma A(u))} \\
        &= 0,
    \]
    where the second line holds by independence of $V(\cdot)$, $\bar{V}(\cdot)$, and $S$; and the second-to-last line is possible because all three of the factors are finite.
    \* $\E{\exp(\gamma S)}$ is finite by the definition of~$\gamma$.
    \* $\E{\exp(\gamma V(\infty))}$ is finite by \cref{eq:model:boost_okay, lem:work_crossing_expression}.
    \* $\E{\exp(\gamma A(u))} = \exp(\lambda u (\E{\exp(\gamma S)} - 1)) = \exp(\gamma u)$, by a standard M/G/1 result \citep{harchol-balter_performance_2013} and the definition of~$\gamma$.
    \*/

    Because the limit in \cref{eq:almost_there:second} vanishes, $\gamma C_\pi$ is at most the limit in \cref{eq:almost_there:first}, so
    \[
        \gamma C_\pi
        &\leq \limsup_{\theta\to\gamma}{} (\gamma - \theta) \E*{\exp\gp*{\theta \gp{W - \min\{B,u\} + V(\infty) + S}}}\\
        &= \limsup_{\theta\to\gamma}{} (\gamma - \theta) \E{\exp\gp{\theta W}} \, \E{\exp\gp{\theta\gp{S - \min\{B,u\}}}} \, \E{\exp\gp{\theta V\gp{\infty}}}\\
        &= \gamma C_W\E{\exp\gp{\gamma\gp{S - \min\{B,u\}}}} \, \E{\exp\gp{\gamma V(\infty)}},
    \]
    where the last step follows from \cref{eq:model:tail_constant_work}. Because this holds for any $u$, it holds in the $u\to\infty$ limit. The monotone convergence theorem implies the limit is $\gamma C$, as desired.
\end{proof}

\section{A batch scheduling problem related to tail optimality}
\label{sec:batch}
In this section we show that \cheat{\theta} minimizes $\E{\exp\gp{\theta T_\pi}}$ in the full-information case. The key idea is that in this case, we can treat each busy period as a finite batch of jobs on which to minimize a cost function which corresponds to $\E{\exp\gp{\theta T_\pi}}$. We start by defining what a batch is and what the cost function is.

\begin{definition}
    A \emph{batch instance} $\calI = \{(a_1,s_1),\ldots,(a_n,s_n)\}$ is a finite batch of pairs of arrival times and job sizes.
\end{definition}

\begin{definition}
    The \emph{$\theta$-cost} of policy~$\pi$ on batch instance $\calI = \{(a_1, s_1), \dots, (a_n, s_n)\}$ is
    \[
        K_\pi(\theta, \calI) = \sum_{i = 1}^n \E{\exp\gp{\theta T_{\pi, i}}} = \sum_{i = 1}^n \E{\exp\gp{\theta (D_{\pi, i} - a_i)}},
    \]
    where $D_{\pi, i}$ is the departure time of job $i$ under policy $\pi$. The expectation is over any randomness in the policy.
\end{definition}

We now show that \cheat{\theta} minimizes $\theta$-cost across any finite batch of jobs, across all preemptive policies, and that this in turn minimizes $\E{\exp\gp{\theta T_\pi}}$.

\begin{theorem}\label{thm:optimal_batch_policy}
    In the full-information setting, the \cheat{\theta} policy minimizes $\theta$-cost. Specifically, for any batch instance~$\calI$,
    \[
        \label{eq:optimal_batch_policy:instance}
        K_{\cheat{\theta}}(\theta, \calI) = \min_\pi K_\pi(\theta, \calI),
    \]
    and therefore
    \[
        \label{eq:optimal_batch_policy:mg1}
        \E{\exp\gp{\theta T_{\cheat{\theta}}}} \leq \inf_\pi \E{\exp\gp{\theta T_\pi}}.
    \]
\end{theorem}

\begin{proof}
    We first show \cref{eq:optimal_batch_policy:instance}, namely that \cheat{\theta} minimizes $\theta$-cost for any batch instance~$\calI$. This can be seen by observing that \cheat{\theta} serves jobs in $\calI$ in the same order as the Weighted Discounted Shortest Processing Time policy (WDSPT), which is known to minimize $\theta$-cost for any batch instance $\calI$. This is because boosted arrival time $a_i - b_\theta(s_i)$ is a monotonic function, namely $1/\theta$ times the negative log, of WDSPT's priority index:
    \[
        a_i - b_\theta(s_i) = -\frac{1}{\theta} \log\gp*{\exp(-\theta a_i) \frac{\exp(\theta s_i)}{\exp(\theta s_i) - 1}}.
    \]
    The proof is an interchange argument identical to that of \citet[Theorem 3.1.6]{pinedo_scheduling_2016}, with the signs for the discount rate and objective flipped. Specifically, with negative discounting, we define discounted completion time as $\exp(\theta D_i) - 1$ instead of $1 - \exp(\theta D_i)$ to keep the sign positive.

    Having shown \cref{eq:optimal_batch_policy:instance}, we now turn to \cref{eq:optimal_batch_policy:mg1}. The key idea is to consider busy periods as batch instances. Specifically, by renewal-reward theorem \citep[Theorem~10.2.15]{bremaud_probability_2020},
    \[
        \E{\exp(\theta T_\pi)} = \frac{\E{K_\pi(\theta, \calB)}}{\E{|\calB|}},
    \]
    where $\calB$ is a random batch instance corresponding to an M/G/1 busy period, and $|\calB|$ is the number of jobs in the instance. The intuition is that the average $\theta$-cost per job is the expected total $\theta$-cost of all jobs in a busy period, namely $\E{K_\pi(\theta, \calB)}$, divided by the expected number of jobs in a busy period, namely $\E{|\calB|}$. But $\E{|\calB|} = \frac{1}{1 - \rho}$ is the same under all scheduling policies, so
    \[
        \E{\exp(\theta T_{\cheat{\theta}})}
        &= \frac{\E{K_{\cheat{\theta}}(\theta, \calB)}}{\E{|\calB|}} \\
        &= \frac{\E{\min_\pi K_\pi(\theta, \calB)}}{\E{|\calB|}} \\
        &\leq \inf_\pi \frac{\E{K_\pi(\theta, \calB)}}{\E{|\calB|}} \\
        &= \inf_\pi \E{\exp(\theta T_\pi)}.
        \qedhere
    \]
\end{proof}

\begin{remark}
    \label{rmk:partial-information_hard}
    While \cref{thm:optimal_batch_policy} treats only the full-information setting, an analogue of \cref{eq:optimal_batch_policy:instance} holds in the partial-information setting. The optimality is relative not to all policies, but \emph{nonpreemptive} policies that have access to only labels and arrival times. However, \cref{eq:optimal_batch_policy:instance} no longer implies \cref{eq:optimal_batch_policy:mg1} in the partial-information setting. See \cref{sec:reduction} for details.
\end{remark}

\section{Proof of Boost's Strong Tail Optimality}
\label{sec:optimality}

In this section, we prove strong tail optimality of \boost{\gamma} in the full-information setting. We do so by computing an explicit lower bound on the optimal tail constant, namely $C^*$ from \cref{thm:strong_tail_optimality} below, and showing that $C_\boost{\gamma} = C^*$. The fact that $C^*$ is a lower bound on the optimal tail constant in the full-information case follows from \cref{thm:optimal_batch_policy}.

\begin{theorem}\label{thm:strong_tail_optimality}
    Consider an M/G/1 with class~I job size distribution with a fixed label-size pair distribution $(L, S)$, and let
    \[
        C^* = \liminf_{\theta \to \gamma} \frac{\gamma - \theta}{\gamma} \E{\exp\gp{\theta T_{\cheat{\theta}}}}.
    \]
    \*[subenv] The tail constant of \boost{\gamma} is
    \[
        \label{eq:gamma_boost_star}
        C_{\boost{\gamma}}
        = C^*
        = C_W \E{\exp(\gamma S) - 1} \exp\gp[\big]{\lambda \E{b_\gamma(L) \, \gp{\exp(\gamma S) - 1}}}.
    \]
    \* In the full-information setting, namely when $L = S$, \boost{\gamma} is strongly tail-optimal.
    \*/
\end{theorem}

\begin{remark}
    While the strong tail optimality result is only for the full-information setting, the definition of $C^*$ and the fact that $C_{\boost{\gamma}} = C^*$ extend to the partial-information setting. Of course, the details of what the labels are affects the value of~$C^*$, with the minimum occurring for the full-information setting.

    However, based on \cref{rmk:partial-information_hard}, we conjecture that in the partial-information setting, $C^*$ is (a lower bound on) the optimal tail constant achievable with \emph{nonpreemptive} policies that have access to only labels and arrival times. \Cref{thm:strong_tail_optimality} would then imply that \boost{\gamma} achieves this optimal tail constant. As evidence for this conjecture, we show in \cref{sec:partial-information_boost_optimality} that \boost{\gamma} outperforms all other versions of \boost{}, and we show in \cref{sec:nudge_vs_boost} that \boost{\gamma} outperforms Nudge-M.
\end{remark}

% \ziv{Can we cut this? Also sorry I messed it up a bit before realizing we can probably cut it, so proofread again if keeping} We follow a similar setup as in \george{Ref:analysis section}. The only difference is that we now consider an infinite sequence of cheating M/G/1 systems parameterized by~$\theta$, all experiencing the same arrival process. The system corresponding to parameter value~$\theta$, which we call \emph{System~$\theta$}, differs from the others only in that it uses the \cheat{\theta} policy for its specific value of~$\theta$. We consider a sequence of parameters such that $\theta \to \gamma$. We assume that all systems are in steady-state. Because we consider multiple instances of from the same family of boost functions, we introduce some notation for referring to the $\theta$th system.

Before proving \cref{thm:strong_tail_optimality}, we introduce some notation for working with the $\theta$-optimal boost function, analogous to the notation used in \cref{sec:analysis}:
% In what follows below, recall that $b_\theta$ is the boost function given in \cref{eq:intro:optimal_boost}, using parameter $\theta$ instead of $\gamma$:
\* Recall that $b_\theta$ is the boost function given in \cref{eq:intro:optimal_boost}.
\* We write $B_\theta = b_\theta(L)$ to mean the boost of the tagged job using boost function~$b_\theta$.
\** Similarly $B'_\theta = b_\theta(L')$ is the boost of a generic future arrival.
\* We write $V_\theta(u)$ to refer to the crossing work using boost function~$b_\theta$.
\*/

Recall that for the constant to be well-defined, we require the crossing work $V(\infty)$ to have finite moment generating function. By \cref{lem:work_crossing_expression}, this amounts to showing \cref{eq:model:boost_okay}. We therefore first show that \cref{eq:model:boost_okay} holds for the $\gamma$-optimal boost function~$b_\gamma$.

% \ziv{Do we need this for all $\theta$ in the proof? [after refreshing myself on it] No! Because we have all other jobs use boost function $b_\gamma$}

\begin{lemma}\label{lem:optimal_work_crossing_finite}
    For all $\theta > 0$, for all labels $l \in \bbL$,
    \[
        \E{B_\theta' \gp*{\exp\gp{\theta S'} - 1}} < \frac{1}{\theta}.
    \]
    In particular, taking $\theta = \gamma$ implies the $\gamma$-optimal boost function satisfies \cref{eq:model:boost_okay}.
\end{lemma}

\begin{proof}
    Plugging in $B_\theta' = b_\theta(L')$ (\cref{def:boost_theta}) and rearranging, it suffices to show that with probability~$1$,
    \[
        \label{eq:optimal_boost_work_crossing_bound_step}
        \gp[\big]{\E{\exp\gp{\theta S'} \given L'} - 1} \log \frac{\E{\exp\gp{\theta S'} \given L'}}{\E{\exp\gp{\theta S'} \given L'} - 1} < 1.
    \]
    Letting $x = \E{\exp\gp{\theta S'} \given L'} - 1$, this holds because $x \log (1 + 1/x) \leq 1$ for all $x > 0$.
    % % First, we can consider $e^{\gamma S'}$ instead of $e^{\theta S'}$ to get an upper bound.
    % Rearranging \cref{eq:model:gamma} yields
    % \[
    %     \frac{\lambda}{\gamma} = \frac{1}{\E*{\exp\gp{\gamma S'} - 1}},
    % \]
    % and let $w(\bar{S}) = \E*{\exp\gp{\gamma \bar{S}}}$ and $g(\bar{S}) = w(\bar{S}) - 1$. Substituting for $\lambda$, we get
    % \[
    %     \E*{\frac{\gamma}{\E*{\exp\gp{\gamma \bar{S}} - 1}} B'_\gamma \gp{\exp\gp{\gamma S'} - 1}}.
    % \]
    % \george{We need a replacement for $S'$ here b/c $S'$ now means non-tagged job. Currently using $\bar{S}$...} Observe that we can write the above as taking the $g(\bar{S})$-weighted expectation, i.e.
    % \[
    %     \E_{\bar{S}}*{\gamma B'_\gamma}.
    % \]
    % Applying Jensen's inequality, we have
    % \[
    %     \E_{\bar{S}}*{\gamma B'_\gamma} &\leq \log\gp*{\E_{\bar{S}}*{\frac{1}{1 - \frac{1}{\E*{\exp\gp{\gamma S'}}}}} }\\
    %     &= \E*{\frac{\E*{\exp\gp{\gamma S'}} - 1}{\E*{\exp\gp{\gamma S'} - 1}} \gp*{\frac{1}{1 - \frac{1}{\E*{\exp\gp{\gamma S'} | L'}}}}}\\
    %     &= \frac{\E*{\exp\gp{\gamma S'}}}{\E*{\exp\gp{\gamma S'}} - 1}\\
    %     &= \frac{\lambda}{\gamma}\E*{\exp\gp{\gamma S'}}.
    % \]
    % We know that $\E*{\exp\gp{\gamma S'}}$ is finite, so we are done.
\end{proof}

Given \cref{lem:optimal_work_crossing_finite}, we can now employ \cref{thm:general_transform, lem:boost_response_time_bounds}. The last ingredient we need to prove \cref{thm:strong_tail_optimality} is to understand how the $\theta$-optimal boost function compares to the $\gamma$-optimal boost function. We observe below that as a function of~$\theta$, the $\theta$-optimal boost is actually monotonic, with larger~$\theta$ yielding smaller boosts.

\begin{lemma}\label{lem:optimal_boost_monotonicity}
    For all labels $l \in \bbL$, both $b_\theta(l)$ and $\theta b_\theta(l)$ are decreasing as functions of~$\theta > 0$.
\end{lemma}

\begin{proof}
    It suffices to prove that $\theta b_\theta(l)$ is decreasing, which follows from writing it as
    \[
        \theta b_\theta(l) = -\log \gp*{1 - \frac{1}{\E{\exp(\theta S) \given L = l}}}.
        \qedhere
    \]
\end{proof}

\begin{proof}[Proof of \cref{thm:strong_tail_optimality}.]
    We first note that the expression on the right-hand side of \cref{eq:gamma_boost_star} follows from \cref{thm:general_transform} and plugging in \cref{def:boost_theta} with $\theta = \gamma$. Specifically, the second factor from \cref{thm:general_transform} simplifies to
    \[
        \E{\exp(\gamma (S - B_\gamma))}
        &= \E{\E{\exp(\gamma S) \given L} \, \exp(-\gamma b_\gamma(L))} \\
        &= \E*{\E{\exp(\gamma S) \given L} \cdot \frac{\E{\exp(\gamma S) \given L} - 1}{\E{\exp(\gamma S) \given L}}} \\
        &= \E{\exp(\gamma S)} - 1.
    \]
    Also, note that in the full-information setting, $C^* \leq C_\pi$ for all policies~$\pi$, so strong tail optimality is implied by $C_{\boost{\gamma}} = C^*$.
    
    It thus remains only to compute $C^*$ to confirm that $C_{\boost{\gamma}} = C^*$. To do so, we consider a system using the \cheat{\theta} policy, bound $\E{\exp(\theta T_{\cheat{\theta}})}$, then compute the $\theta \to \gamma$ limit. We use the analysis from \cref{sec:analysis:tagged_job} with boost function~$b_\theta$ (\cref{def:boost_theta}).
    
    From \cref{lem:boost_response_time_bounds}, we have $T_\cheat{\theta} \geq W - B_\theta + V_\theta(W) + S$. We can lower bound this further using \cref{lem:optimal_boost_monotonicity}: if we were to change \emph{other jobs' boosts} from $B_\theta'$ to $B_\gamma'$, it would only improve the tagged job's response time. This means $V_\theta(w) \geq V_\gamma(w)$ for all~$w \geq 0$, so
    \[
        T_\cheat{\theta} \geq W - B_\theta + V_\gamma(W) + S.
    \]
    Therefore, for all~$w \geq 0$,
    \[
        \E*{\exp\gp*{\theta T_\cheat{\theta}}} &\geq \E*{\exp\gp*{\theta \gp*{W - B_\theta + V_\gamma(W) + S}}}\\
        &\geq \E*{\exp\gp*{\theta \gp{W - B_\theta + V_\gamma(W) + S}} \, \1\gp*{W > w}} \\
        &\geq \E*{\exp\gp*{\theta \gp{W - B_\theta + V_\gamma(w) + S}} \, \1\gp*{W > w}} \\
        &= \E{\exp\gp{\theta W} \, \1\gp*{W > w}} \, \E{\exp\gp{\theta \gp{S - B_\theta}}} \, \E{\exp\gp{\theta V_\gamma(w)}}.
    \]
    Taking the $\theta \to \gamma$ limit and applying the monotone convergence theorem, which applies thanks to the monotonicity of $\theta B_\theta$ (\cref{lem:optimal_boost_monotonicity}), yields
    \[
        \MoveEqLeft
        \liminf_{\theta\to\gamma}{} (\gamma - \theta) \E{\exp\gp{\theta T_\cheat{\theta}}} \\
        &\geq \liminf_{\theta\to\gamma}{} (\gamma - \theta) \E{\exp\gp*{\theta W}\1\gp*{W > w}} \, \E{\exp\gp{\theta \gp{S - B_\theta}}} \, \E{\exp\gp{\theta V_\gamma(w)}} \\
        &= \gp*{\lim_{\theta\to\gamma}{} (\gamma - \theta) \E{\exp\gp*{\theta W}\1\gp*{W > w}}} \gp*{\lim_{\theta\to\gamma} \E{\exp\gp{\theta \gp{S - B_\theta}}}} \gp*{\lim_{\theta\to\gamma}\E{\exp\gp{\theta V_\gamma(w)}}} \\
        \label{eq:optimality:almost_there}
        &= \gamma C_W\E{\exp\gp{\gamma\gp{S - B_\gamma}}}\E{\exp\gp{\gamma V_\gamma(w)}},
    \]
    where the second line follows from \cref{eq:model:tail_constant_work} and the fact that
    \[
        \lim_{\theta\to\gamma}{} \gp{\gamma - \theta} \E{\exp\gp{\theta W} \, \1(W \leq w)}
        \leq \lim_{\theta\to\gamma}{} \gp{\gamma - \theta} \exp(\gamma w)
        = 0.
    \]
    Finally, as in the proof of \cref{thm:general_transform}, we observe that because \cref{eq:optimality:almost_there} holds for all~$w$, it holds in the $w \to \infty$ limit, yielding
    \[
        C^* \geq C_W\E*{\exp\gp*{\gamma\gp*{S - B_\gamma}}}\E*{\exp\gp*{\gamma V_\gamma(\infty)}}.
    \]
    The right-hand side is exactly the tail constant~$C_\boost{\gamma}$ from \cref{thm:general_transform}, so $C^* \geq C_\boost{\gamma}$.

    It remains only to show $C^* \leq C_\boost{\gamma}$. In the full-information setting, this is immediate from \cref{thm:optimal_batch_policy}, which implies $C^* \leq C_\pi$ for all policies~$\pi$. In the partial-information setting, \cref{thm:optimal_batch_policy} does not apply, but by following essentially the same steps as the upper bound in \cref{thm:general_transform}, we obtain
    \[
        \gamma C^*
        &\leq
        % \limsup_{\theta\to\gamma}{} (\gamma - \theta) \E*{\exp\gp*{\theta \gp{W - \min\{B_\theta,u\} + V_{\theta}(\infty) + S}}}\\
        % &=
        \gp*{\lim_{\theta\to\gamma}{} (\gamma - \theta) \E{\exp\gp{\theta W}}}
            \gp*{\lim_{\theta\to\gamma}{} \E{\exp\gp{\theta\gp{S - \min\{B_\theta,u\}}}}}
            \gp*{\lim_{\theta\to\gamma}{} \E{\exp\gp{\theta V_{\theta}\gp{\infty}}}} \\
        &= \gamma C_W\E{\exp\gp{\gamma\gp{S - \min\{B_\gamma,u\}}}} \, \E{\exp\gp{\gamma V_{\gamma}(\infty)}}.
    \]
    This becomes $\gamma C_{\boost{\gamma}}$ in the $u \to \infty$ limit, as desired. The main difference from \cref{thm:general_transform} is that the second and third limits above involve a boost function that varies with~$\theta$. We compute the second limit with \cref{lem:optimal_boost_monotonicity} and the monotone convergence theorem, and we compute the third limit with \cref{lem:work_crossing_expression} and the bounded convergence theorem, which applies thanks to~\cref{eq:optimal_boost_work_crossing_bound_step}.
\end{proof}

% Finally, we conclude with a theorem that shows the improvement over \fcfs~is strict, assuming the labels have useful information. Formally,

% \begin{theorem}
%     Unless \boost{\gamma} assigns all labels (outside a set of measure zero) the same boost,
%     \[
%         C_\boost{\gamma} < C_\fcfs.
%     \]
% \end{theorem}

% \begin{proof}
%     The result of \cref{thm:general_transform,lem:optimal_work_crossing_finite} allows us to write \[
%         \frac{C_{\boost{\gamma}}}{C_{\fcfs}} \leq \frac{\E*{\exp\gp{\gamma \gp*{S - b_\gamma(L)}}} \frac{\lambda}{\gamma} \E{\exp\gp{\gamma S}} }{\E*{\exp\gp{\gamma S}}}.
%     \]
%     Recall that in \cref{lem:optimal_work_crossing_finite} the inequality arises from the use of Jensen's inequality, so that the inequality is strict unless the internal term is constant (i.e. $b_\boost{\gamma}(L)$ is constant across $L$).
    
%     Now consider the term $\E*{\exp\gp{\gamma \gp*{S - b_\gamma (L)}}}$. We have the following:
%     \[
%         \E*{\exp\gp{\gamma \gp*{S - b_\gamma(L)}}} &= \E*{\exp\gp{\gamma S} \exp\gp{-\gamma b_\gamma (L)}}\\
%         &= \E*{\exp\gp{\gamma S} \gp*{1 - \frac{1}{\E*{\exp\gp{\gamma S} \given L}}}}\\
%         &= \E*{ \E*{\exp\gp{\gamma S} \given L} \gp*{1 - \frac{1}{\E*{\exp\gp{\gamma S} \given L}}}}\\
%         &= \frac{\gamma}{\lambda}.
%     \]
%     Plugging in for these two terms, we get that
%     \[
%         \frac{C_{\boost{\gamma}}}{C_{\fcfs}} \leq 1,
%     \]
%     with strict inequality unless the boost is constant across labels, as desired.
% \end{proof}

\section{Simulations}
\label{sec:simulation}
We have shown that \boost{\gamma} achieves strong tail optimality, which is an asymptotic property. However, there remain unanswered questions about \boost{\gamma} that are important to practitioners. In this section, we explore the questions below via simulations. We focus by default on the full-information setting, but we address some partial-information settings in the last two questions.
\* (\cref{sec:simulation:is_it_good}) How well does \boost{\gamma} perform in practical regimes? Does it do as well as one would predict from its tail constant~$C_\boost{\gamma}$?
\* (\cref{sec:simulation:vs_others}) How does \boost{\gamma} compare to Nudge and SRPT?
\* (\cref{sec:simulation:variation}) In what settings does \boost{\gamma} offer the largest benefit? What role does the variance of the job size distribution play?
\* (\cref{sec:simulation:robustness}) Is \boost{\gamma} robust to misspecification, such as being given the wrong value of~$\gamma$ or noisily estimated job sizes?
\* (\cref{sec:simulation:unknown_sizes}) How well does \boost{\gamma} perform in the partial-information setting, where we have much coarser information about jobs' sizes?
\*/

\subsection{Boost in the full-information setting}
\label{sec:simulation:is_it_good}

\begin{figure}
  % \centering
  \begin{subfigure}{\subfigwidth}
    \includegraphics[width=\linewidth]{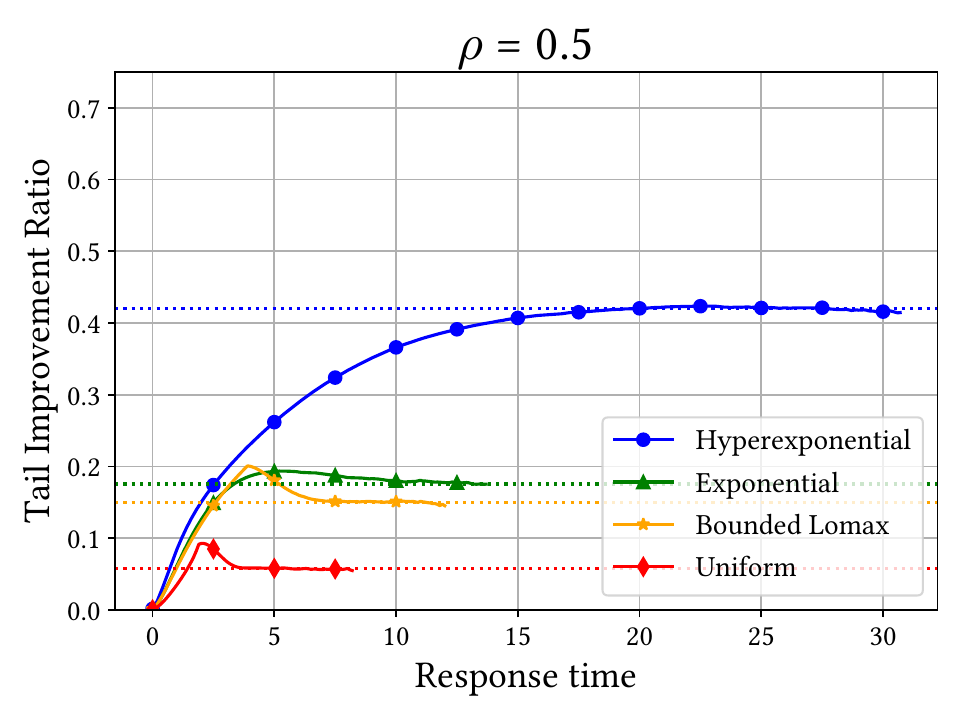}
    \Description{Graph showing tail improvement ratio of $\gamma$-Boost for four job size distributions: Hyperexponential, Exponential, Bounded Lomax, and Uniform. $x$ axis is response time (range 0 to 200), $y$ axis is tail improvement ratio (range 0 to 0.7). Hyperexponential: starts at y = 0, then slowly rises to an asymptote at about y = 0.42. Each other curve starts at y = 0, peaks early, then reduces as it converges to an asymptote. Exponential: peaks at about y = 0.2, asymptotes to about y = 0.18. Bounded Lomax: peaks at about y = 0.2, asymptotes to about y = 0.15. Uniform: peaks at about y = 0.1, asymptotes to about y = 0.05.}%
  \end{subfigure}%
  % \hfill
  % \begin{subfigure}{\subfigwidth}
  %   \includegraphics[width=\linewidth]{imgs/np_rho_0.8_plot.pdf}
  % \end{subfigure}%
  \hfill
  \begin{subfigure}{\subfigwidth}
    \includegraphics[width=\linewidth]{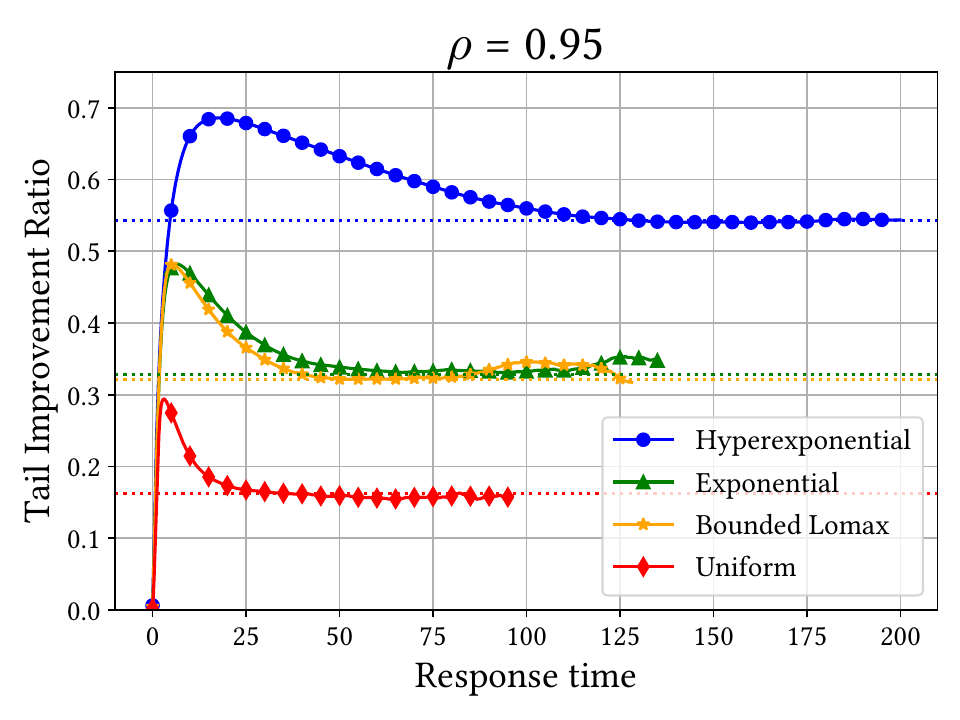}
    \Description{Graph showing tail improvement ratio of $\gamma$-Boost for four job size distributions: Hyperexponential, Exponential, Bounded Lomax, and Uniform. $x$ axis is response time (range 0 to 200), $y$ axis is tail improvement ratio (range 0 to 0.7). Each curve starts at y = 0, peaks early, then reduces as it converges to an asymptote. Hyperexponential: peaks at about y = 0.68, asymptotes to about y = 0.54. Exponential: peaks at about y = 0.48, asymptotes to about y = 0.33. Bounded Lomax: peaks at about y = 0.48, asymptotes to about y = 0.32. Uniform: peaks at about y = 0.29, asymptotes to about y = 0.16.}%
  \end{subfigure}%
  \caption{Empirical TIR of \boost{\gamma} over FCFS for several job size distributions~$S$, each with mean $\E{S} = 1$, at loads $\rho = 0.5, 0.95$. See \cref{fig:intro:performance:distributions} for $\rho = 0.8$. We use the same distributions as \citet[Fig.~2]{grosof_nudge_2021}, which are: Uniform($0$,$2$), Exponential, Hyperexponential with branches drawn from Exp($2$) and Exp($1/3$) and first branch probability $0.8$, and BoundedLomax with shape parameter $\alpha = 2$ and upper bound~$4$. The asymptotic TIR is computed with \cref{thm:general_transform} and plotted as a same color dotted line for each distribution. Simulations run with 50 million arrivals.}
  \label{fig:boost_vs_load}
\end{figure}

In \cref{fig:boost_vs_load}, we evaluate the performance of the optimal boost policy on common distributions by looking at the empirical Tail Improvement Ratio (TIR) with respect to \fcfs, that is, by looking at
% \[
    $\mathrm{TIR}(t) = 1 - {\P{T_\boost{\gamma} > t}}/{\P{T_\fcfs > t}}$.
% \]
\boost{}'s performance improves upon \fcfs's across a variety of job size distributions and loads. In all tested distributions, \boost{\gamma} achieves asymptotic performance equivalent to the TIR that our theory suggests. Moreover, these improvements can be significant: \boost{\gamma} improves the tail constant in the Exponential case by roughly 30\%, and improvements exceed 50\% for the Hyperexponential distribution case. Another observation of note is that across all loads and distributions, not only does \boost{\gamma} achieve the asymptotic performance suggested by theory, it also improves \emph{stochastically} over \fcfs, meaning $\mathrm{TIR}(t) > 0$ for all $t > 0$, not just the $t \to \infty$ limit. Moreover, in many cases the ``pre-asymptotic'' improvement actually exceeds the asymptotic TIR.

\subsection{Boost compared to other policies}
\label{sec:simulation:vs_others}

\begin{figure}
  \centering
  % \begin{subfigure}{\subfigwidth}
  %   % \includegraphics[width=\linewidth]{imgs/np_exp_comp.pdf}
  % \end{subfigure}%
  \begin{subfigure}{\subfigwidth}
    \includegraphics[width=\linewidth]{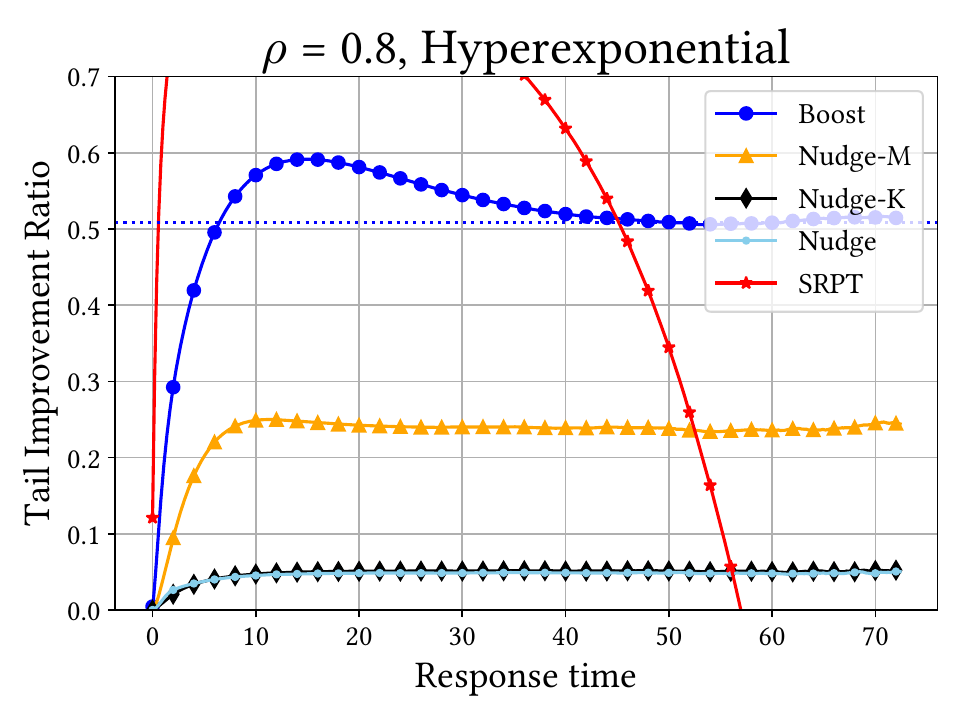}
    \Description{Graph showing tail improvement ratio of $\gamma$-Boost, Nudge-M, Nudge-K, Nudge, and SRPT for the Hyperexponential job distribution: $x$ axis is response time (range 0 to 80), $y$ axis is tail improvement ratio (range 0 to 0.7). SRPT's curve starts at y = 0, peaks above y = 0.7, then decays to negative infinity, going below y = 0 at x = 55. For the other policies, each of their curves starts at y = 0, peaks, then reduces as it converges to an asymptote. Boost: peaks at about y = 0.6, asymptotes to about y = 0.5. Nudge-M: peaks at about y = 0.25, asymptotes to about y = 0.25. Nudge-K: peaks at about y = 0.05, asymptotes to about y = 0.05. Nudge: peaks at about y = 0.05, asymptotes to about y = 0.05.}
  \end{subfigure}%
  \hfill
  \begin{subfigure}{\subfigwidth}
    \includegraphics[width=\linewidth]{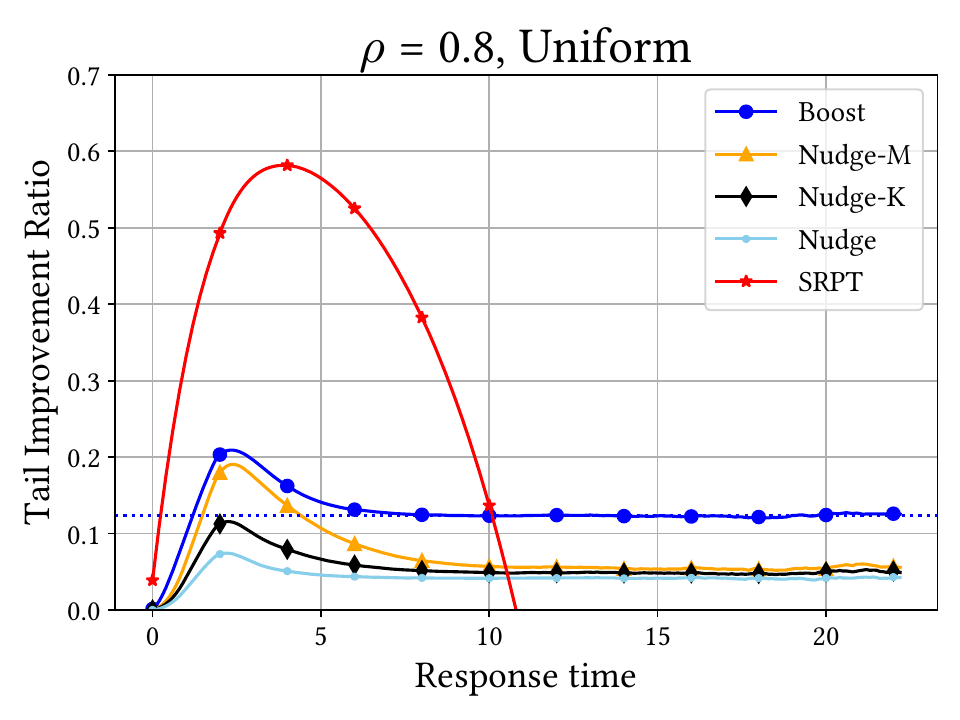}
    \Description{Graph showing tail improvement ratio of $\gamma$-Boost, Nudge-M, Nudge-K, Nudge, and SRPT for the Uniform job distribution: $x$ axis is response time (range 0 to 25), $y$ axis is tail improvement ratio (range 0 to 0.7). SRPT's curve starts at y = 0, peaks at y = 0.6, then decays to negative infinity, going below 0 around x = 11. For the other policies, each of their curves starts at y = 0, peaks, then reduces as it converges to an asymptote. Boost: peaks at about y = 0.2, asymptotes to about y = 0.12. Nudge-M: peaks at about y = 0.22, asymptotes to about y = 0.06. Nudge-K: peaks at about y = 0.12, asymptotes to about y = 0.06. Nudge: peaks at about y = 0.08, asymptotes to about y = 0.05.}
  \end{subfigure}
  % \hfill

  % \medskip

  % \hfill
  % \begin{subfigure}{\subfigwidth}
  %   % \includegraphics[width=\linewidth]{imgs/np_bp_comp.pdf}
  % \end{subfigure}
  % \caption{Comparison of empirical TIR of \boost{\gamma} against FCFS and Nudge, for two job size distributions, each with mean $1$, and with the same settings as in \citet[Fig. 2]{grosof_nudge_2021}. The left is Hyperexponential with branches drawn from Exp($2$) and Exp($1/3$), with first branch probability $0.8$ and the right is Uniform($0$,$2$). See \cref{fig:intro:performance:policies} for exponential. In each plot, the dotted horizontal line represents \boost{\gamma}'s asymptotic TIR for the respective distribution. Simulations run with 50 million arrivals.}
  \caption{Comparison of empirical TIR of \boost{\gamma} against FCFS and Nudge, for two job size distributions, each with mean $1$, and with the same settings as in \citet[Fig. 2]{grosof_nudge_2021}. The left is Hyperexponential with branches drawn from Exp($2$) and Exp($1/3$), with first branch probability $0.8$ and the right is Uniform($0$,$2$). See \cref{fig:intro:performance:policies} for exponential. In each plot, the dotted horizontal line represents \boost{\gamma}'s asymptotic TIR for the respective distribution. For Hyperexponential, we set~$K$ to the optimal value of~$8$ for Nudge-K/M, with type-$1$ and type-$2$ jobs set to jobs coming from the Exp($2$) and Exp($1/3$) branches respectively. (Nudge-K does perform slightly better than Nudge in this case, though this is barely visible on the plot.) For Uniform, we use the same small-large split as Nudge, where type-$1$ jobs are small (smaller than the mean of the distribution) and type-$2$ jobs are large, and set~$K$ to the optimal value of~$3$ for Nudge-K/M. Simulations run with 50 million arrivals.}
  % BoundedLomax with shape parameter $\alpha = 2$ and upper bound $4$, and . These are the same settings as in \citet[Fig. 2]{grosof_nudge_2021}. }
  \label{fig:boost_vs_other_policies}
\end{figure}

We evaluate the performance of boost policies against other policies, namely against Nudge, which is known to have better stochastic performance than \fcfs, and SRPT, which is tail-pessimal for Class~I distributions. In \cref{fig:intro:performance:policies, fig:boost_vs_other_policies}, we compare the TIR for all three policies. Following best practices from \citet[Fig. 2, Section 9]{grosof_nudge_2021}, we set the small-large threshold for Nudge to be at $\E*{S}$, with no medium or extra-large split, and examine performance under a variety of common job size distributions. We find that \boost{\gamma} has larger asymptotic TIR than Nudge and SRPT across all tested job size distributions. Moreover, \boost{\gamma} is stochastically better than Nudge across the distributions tested.

\subsection{Variation matters: how CoV affects asymptotic performance}
\label{sec:simulation:variation}

\begin{figure}
  \centering
  \begin{subfigure}{\subfigwidth}
    \includegraphics[width=\linewidth]{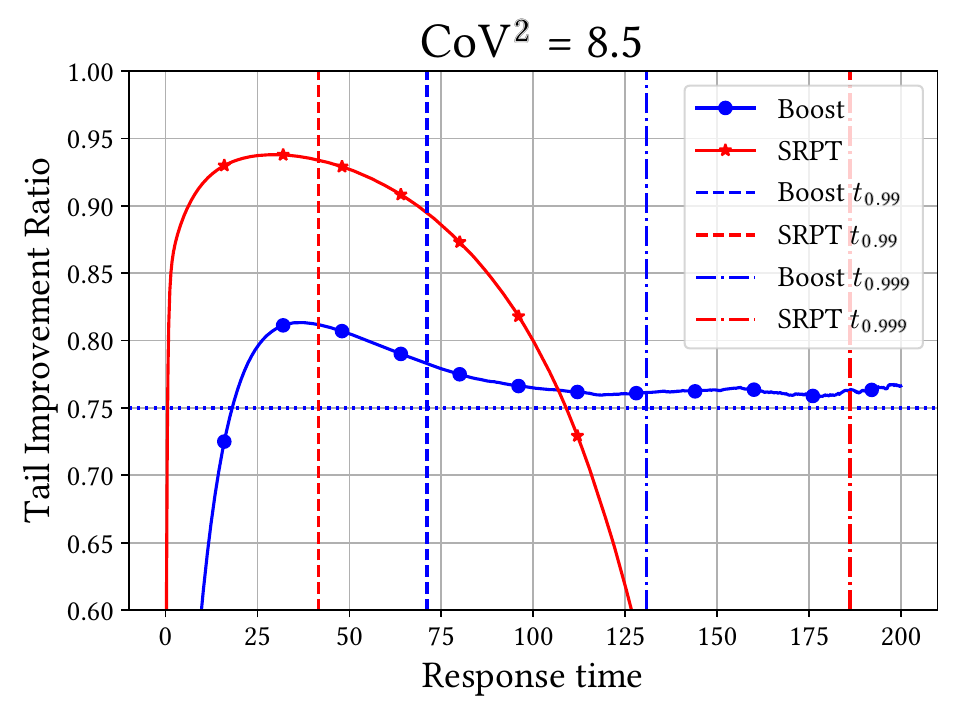}
    \Description{Graph showing tail improvement ratio of $\gamma$-Boost and SRPT. $x$ axis is response time (range 0 to 200), $y$ axis is tail improvement ratio (range 0.6 to 1). $\gamma$-Boost: enters the range at about x = 12, peaks at about 0.81, then asymptotes to about 0.75. SRPT: enters the range at about x = 0, peaks at about y = 0.94, has a broad peak, then reduces sharply, exiting the range about x = 125. Dotted lines show the 99th and 99.9th response time percentiles of both policies. SRPT's 99th percentile (about x = 40) is less than $\gamma$-Boost's (about x = 70), but $\gamma$-Boost's 99.9th percentile (about x = 130) is less than SRPT's (about x = 185).}%
  \end{subfigure}%
  \hfill
  \begin{subfigure}{\subfigwidth}
    \includegraphics[width=\linewidth]{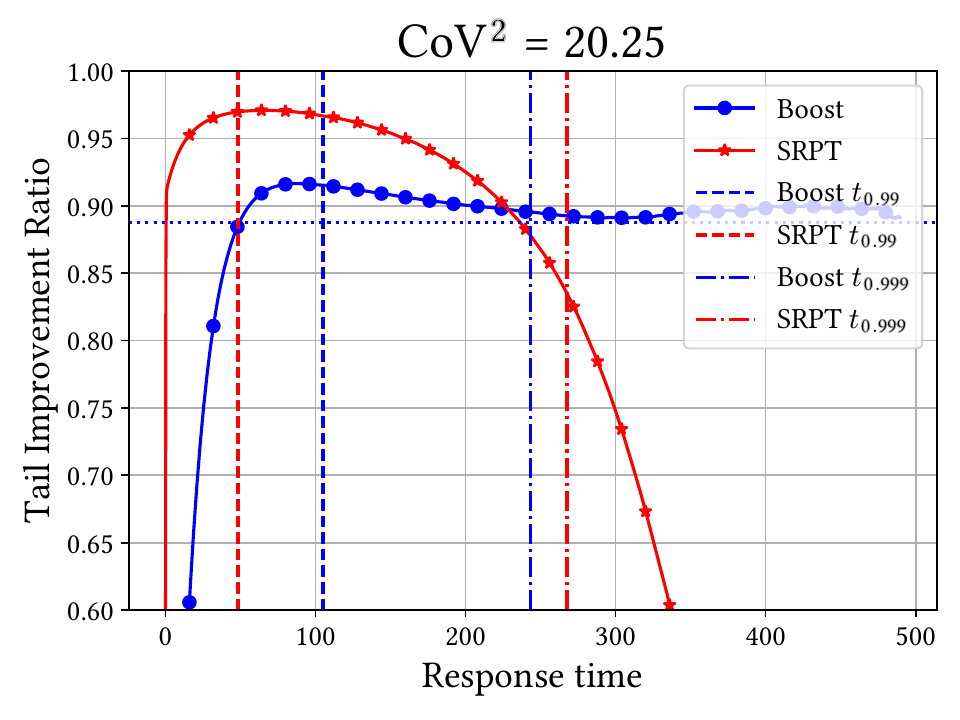}
    \Description{Graph showing tail improvement ratio of $\gamma$-Boost and SRPT. $x$ axis is response time (range 0 to 500), $y$ axis is tail improvement ratio (range 0.6 to 1). $\gamma$-Boost: enters the range at about x = 20, peaks at about 0.92, then asymptotes to about 0.89. SRPT: enters the range at about x = 0, peaks at about y = 0.97, has a broad peak, then reduces sharply, exiting the range about x = 335. Dotted lines show the 99th and 99.9th response time percentiles of both policies. SRPT's 99th percentile (about x = 50) is less than $\gamma$-Boost's (about x = 105), but $\gamma$-Boost's 99.9th percentile (about x = 245) is less than SRPT's (about x = 265).}
  \end{subfigure}
  \caption{A comparison of \boost{\gamma} and SRPT at high $\mathrm{CoV}$. In both plots, the distributions considered are Hyperexponential, mean $1$ distributions. For $\mathrm{CoV}^2 = 8.5$, we choose parameters Exp(4), Exp(1/6), and a first-branch probability of $p = 20/23$. For $\mathrm{CoV}^2 = 20.25$ we choose parameters Exp(8), Exp(1/12), and a first-branch probability of $p = 88/95$. Load is $\rho = 0.8$. Dashed vertical lines mark the $t_{0.99}$ response times of the two policies, and dash-dotted vertical lines mark the $t_{0.999}$ response time. The dotted horizontal line represents the theoretical asymptotic TIR for \boost{\gamma}. Observe how SRPT has lower $t_{0.99}$ response time but higher $t_{0.999}$ response time than \boost{\gamma}. Simulations run with 50 million arrivals.}
  \label{fig:c2_plots}
\end{figure}

What is important when scheduling for the tail in distributions? If job size distributions are highly variable, tail-pessimal policies for class~I distributions that are optimal for heavy-tailed distributions may still perform well for all but the highest tail percentiles. In particular, SRPT demonstrates strong performance for all but the highest tail percentiles. Therefore, while \boost{\gamma} is asymptotically optimal, a practitioner's choice of policy depends on how sensitive they are to tail response times and the properties of the job size distribution they face. In particular, one may care about tail percentiles that are ``pre-asymptotic''. In \cref{fig:c2_plots}, we show how this ``pre-asymptotic'' performance can depend on the variability of the job size distribution. The main takeaway is that for job size distributions with high variability, SRPT has great performance unless one cares about extremely high threshold (e.g. $t_{0.999}$) response times.

% \subsection{Choosing boost policy parameters.}
% We next evaluate the performance of the optimal boost policy against another boost policy, which uses a boost function of the form $r_\alpha(S) = \frac{1}{\alpha^2 S}$. For both policies, we examine how performance at $t_{0.99}$ varies as a function of their parameters ($\beta$ and $\alpha$, respectively). Letting $\alpha^*$ be the optimal choice of $\alpha$, found by numerical optimization, we compare performance as the parameters vary from $[\gamma / 2, 2\gamma]$ and $[\alpha^*/2, 2\alpha^*]$, respectively.

\subsection{Robustness of Boost}
\label{sec:simulation:robustness}

\begin{figure}
  \centering
  \begin{subfigure}{\subfigwidth}
    \includegraphics[width=\linewidth]{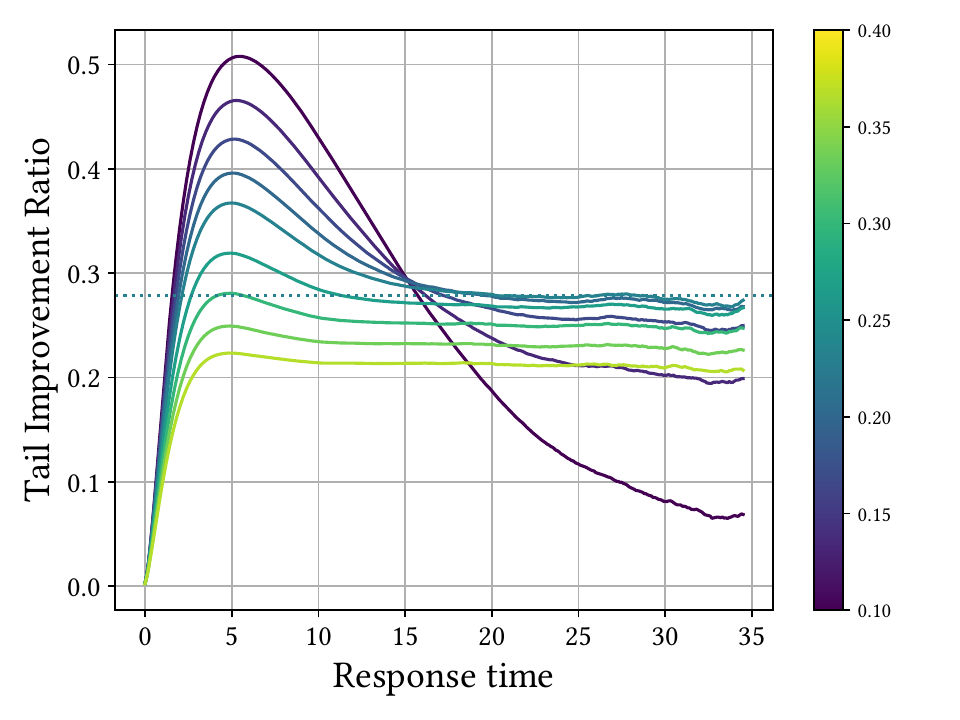}
    \caption{Sensitivity to misspecified $\gamma$.}
    \label{fig:misspecified_gamma}
    \Description{Graph showing tail improvement ratio of $\gamma$-Boost for different gamma parameters, ranging from 0.1 to 0.4: $x$ axis is response time (range 0 to 35), $y$ axis is tail improvement ratio (range 0 to 0.6). All curves peak at x = 5 and then asymptote after. Gamma between 0.1 to 0.15: the curves peak between y = 0.45 to 0.5. Higher gamma have lower peaks. The curves decay to an asymptote between y = 0.08 to 0.2. Higher gamma have higher asymptote. Gamma between 0.15 and 0.2: the curves peak between y = 0.32 to 0.43. Higher gamma have lower peaks. The curves decay to an asymptote between y = 0.25 and 0.27. Higher gamma have higher asymptote. Gamma between y = 0.2 and 0.4: the curves peak between y = 0.22 to 0.32. Higher gamma have lower peaks. The curves decay to an asymptote between y = 0.27 and 0.2. Higher gamma have lower asymptote.}
  \end{subfigure}%
  \hfill
  \begin{subfigure}{\subfigwidth}
    \includegraphics[width=\linewidth]{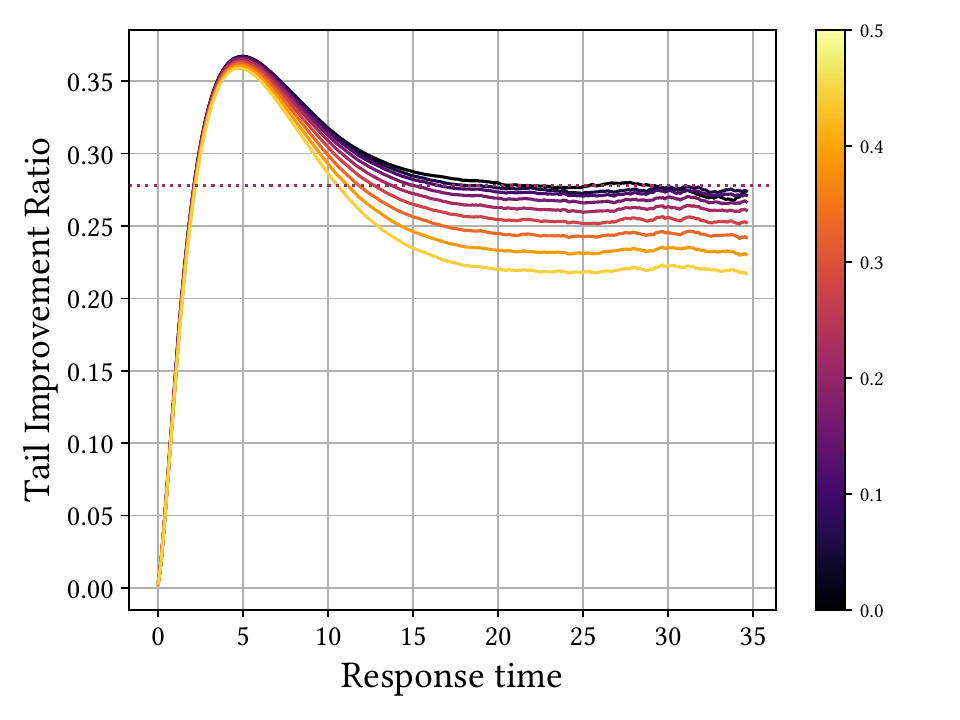}
    \caption{Sensitivity to noisy size estimates.}
    \label{fig:noisy_job_sizes}
    \Description{Graph showing tail improvement ratio of $\gamma$-Boost for different noise levels of job size estimates. The standard deviation ranges from 0 to 0.5: $x$ axis is response time (range 0 to 35), $y$ axis is tail improvement ratio (range 0 to 0.35). All curves peak at x = 5 and then asymptote after. Higher noise leads to lower asymptotes. The range of asymptote decreases from y = 0.27 to 0.22 as noise ranges from 0 to 0.5.}
  \end{subfigure}
  \caption{Performance of \boost{\gamma} with different types of noise. Job size distribution is Exponential with mean $1$, with $\rho = 0.8$. (a)~We plot the performance with misspecified $\gamma$. The theoretical asymptotic TIR is shown as a dotted horizontal line. We consider a range of $[\gamma/2, 2\gamma]$, where the optimal $\gamma = 0.2$. Observe that underestimating $\gamma$ leads to significantly worse TIR compared to overestimating $\gamma$. (b)~We plot the performance with noisy job size estimates. The noise is multiplicative and drawn i.i.d. for each job from a LogNormal distribution with mean $0$ and some standard deviation. We consider noise levels (i.e. standard deviations) ranging from $0$ to $0.5$. While performance is degraded from the theoretically optimal TIR as more noise is added, at $0.5$, the TIR is still above 20\%. In both (a) and~(b), simulations run with 50 million arrivals.}
  \label{fig:noise_plots}
\end{figure}

\boost{\gamma} sets a job's boost using its arrival time, size, and the decay rate parameter $\gamma$ for the job size distribution. In practice, one might only have access to noisy estimates of the job size distribution and of the job sizes.

If one does not know the exact job size distribution, then the parameter $\gamma$ may be misspecified. In this case, our takeaways suggest that one can conservatively set the decay rate parameter $\gamma$ higher than estimated to maintain good performance. For clarity, we denote the optimal policy as \boost{\gamma}, and assume that we set the parameter $\hat{\gamma}$ based on our noisy information. We can see in \cref{fig:misspecified_gamma} that using a conservative overestimate of the parameter, namely setting $\hat{\gamma} > \gamma$, has good performance. While $\hat{\gamma} < \gamma$ can still have good performance near $\gamma$, setting too low of a parameter leads to a faster decay in performance than setting too high of a parameter.

If one only has noisy estimates of the job sizes, using these estimates directly provides still good performance. In \cref{fig:noisy_job_sizes} we examine the sensitivity of \boost{\gamma} to noise in the labels. Using the noisy size as input to the boost policy maintains good performance, with TIR above 20\%.

\subsection{Boost in the partial-information setting}
\label{sec:simulation:unknown_sizes}

\begin{figure}
  \centering
  \begin{subfigure}{\subfigwidth}
    \includegraphics[width=\linewidth]{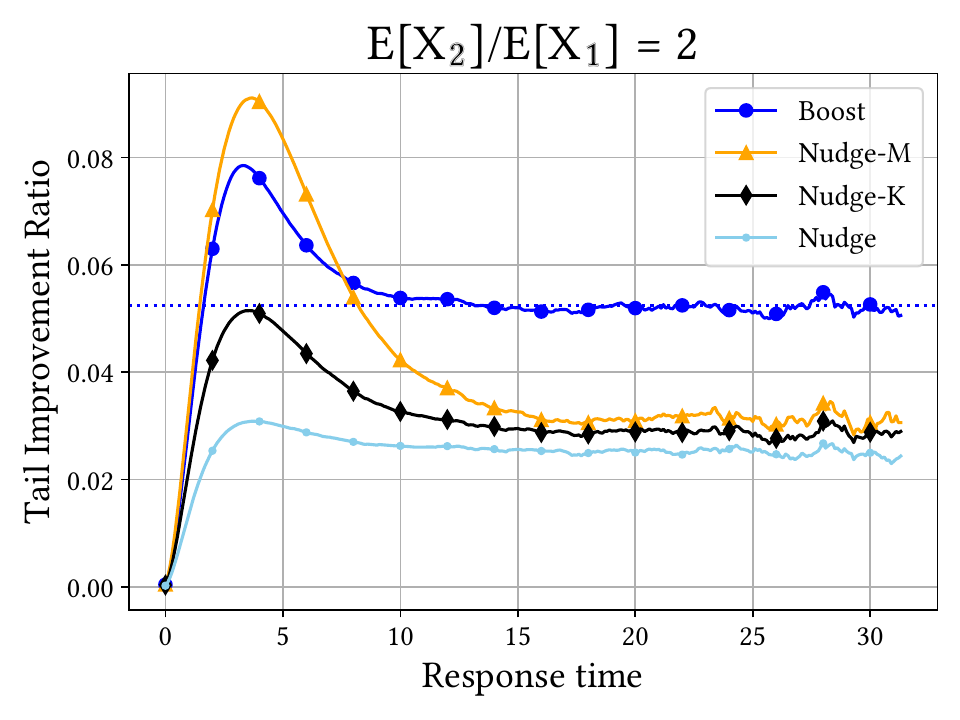}
    \Description{Graph showing tail improvement ratio of $\gamma$-Boost, Nudge-M, Nudge-K, and Nudge for the two-class Hyperexponential distribution with mean ratio 2: $x$ axis is response time (range 0 to 35), $y$ axis is tail improvement ratio (range 0 to 0.1). For all policies, each of their curves starts at y = 0, peaks around x = 4, then reduces as it converges to an asymptote. Boost: peaks at about y = 0.08, asymptotes to about y = 0.055. Nudge-M: peaks at about y = 0.09, asymptotes to about y = 0.03. Nudge-K: peaks at about y = 0.055, asymptotes to about y = 0.03, but below Nudge-M. Nudge: peaks at about y = 0.03, asymptotes to about y = 0.025.}
  \end{subfigure}%
  \hfill
  \begin{subfigure}{\subfigwidth}
    \includegraphics[width=\linewidth]{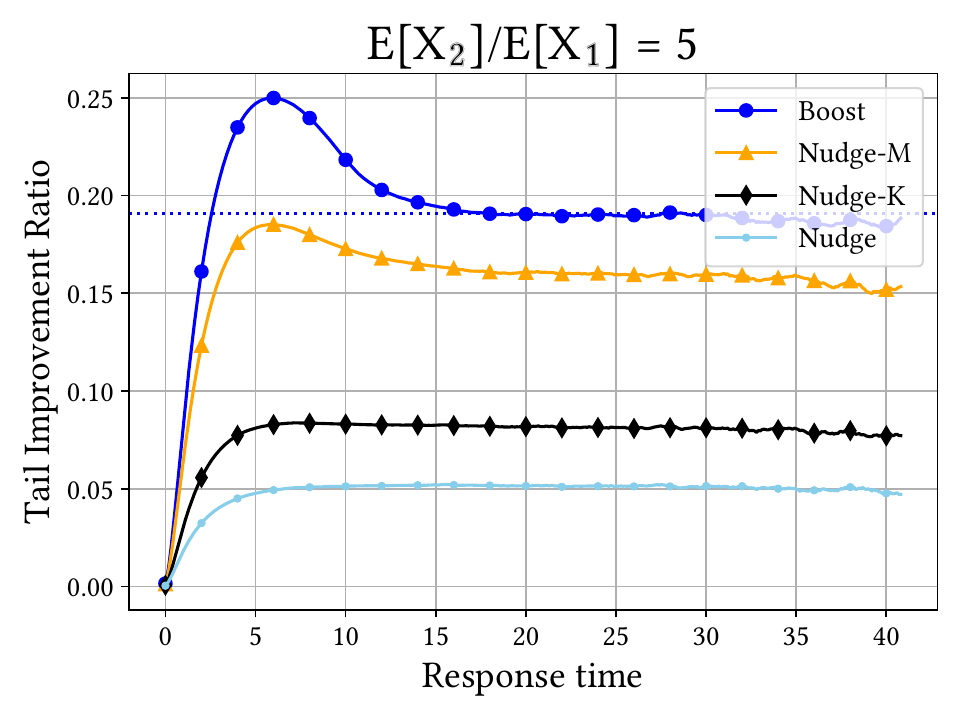}
    \Description{Graph showing tail improvement ratio of $\gamma$-Boost, Nudge-M, Nudge-K, and Nudge for the two-class Hyperexponential distribution with mean ratio 5: $x$ axis is response time (range 0 to 45), $y$ axis is tail improvement ratio (range 0 to 0.3). For all policies, each of their curves starts at y = 0, peaks around x = 5, then reduces as it converges to an asymptote. Boost: peaks at about y = 0.25, asymptotes to about y = 0.19. Nudge-M: peaks at about y = 0.19, asymptotes to about y = 0.15. Nudge-K: peaks at about y = 0.08, asymptotes to about y = 0.08. Nudge: peaks at about y = 0.05, asymptotes to about y = 0.05.}
  \end{subfigure}
  \caption{Comparison of \boost{\gamma}'s performance to Nudge-K's and Nudge-M's performance in the two-class partial-information setting. The mean of the job size distribution is $1$, and, following the settings in \citet[Fig.~1]{vanhoudt_stochastic_2022}, we set the load to be $\rho = 0.75$ and consider the case where both type-$1$ and type-$2$ jobs come from Exponential distributions and the probability of drawing from either branch is $1/2$. On the left, we consider the case where the ratio of the means type-2/type-1 is $2$; on the right it is $5$. The more separable the means are, the better the performance of \boost{\gamma} becomes. In both cases, the asymptotic TIR of the optimal boost, represented by the dotted line, is better than that of both Nudge-K and Nudge-M. In both cases we set the optimal value of $K$ for the Nudge policies ($K = 3$ in the first setting and $K = 5$ in the second setting). Simulations run with 50 million arrivals.}
  \label{fig:label_setting}
\end{figure}

In this section, we consider the case when job sizes are unknown, and one may only have coarse information, such as the class, of incoming jobs or rough thresholds on job sizes. We find that, as with the full-information setting, \boost{\gamma} demonstrates good practical performance, attaining the asymptotic TIR suggested by theory.

In \cref{fig:label_setting} we evaluate the performance of the optimal boost policy in one such limited-information setting. We consider the case where jobs can be \emph{type-1} or \emph{type-2}, where each type has a different distribution, and consider Nudge-K and Nudge-M for $K$ in the setup from \citet[Fig.~1]{vanhoudt_stochastic_2022}. We also extend this to the case where the distributions of the two job types are further separated. We find that performance increases the more distinguishable the distributions of the two job types are. We also find that under these settings, \boost{\gamma} has better performance than both Nudge-K and Nudge-M.

\section{Conclusion}
\label{sec:conclusion}
In this work, we introduce the \emph{\boost{}} family of scheduling policies, which provide a simple new approach to balancing the tradeoff between prioritizing short jobs and prioritizing jobs that have been waiting a long time. We prove that a policy in this new family, called \emph{\boost{\gamma}}, is a \emph{strongly tail-optimal scheduling policy} for the M/G/1 with light-tailed job size distributions, resolving a long-standing open problem in queueing theory. Our simulations show that in addition to achieving a theoretical milestone, \boost{\gamma} has excellent practical performance.

Our results on \boost{} reveal many promising future directions, some of which we outline below.

\subsection{Settings beyond the full-information M/G/1}

The most direct question prompted by our results is: what policy is strongly tail-optimal in the partial-information setting? We believe that \boost{\gamma} is the desired policy if one further restricts attention to nonpreemptive policies. But if one allows preemption, it is less clear what to do. For minimizing mean response time in the M/G/1, it is known that the \emph{Gittins policy} \citep{gittins_multiarmed_2011} is optimal for a wide variety of preemptive partial-information settings \citep{scully_gittins_2020, scully_gittins_2021}. Moreover, versions of the Gittins policy exist for the batch relaxation \citep{pinedo_scheduling_2016} with discounting. Perhaps a version of \boost{} based on the Gittins policy is needed for the preemptive partial-information setting.

\boost{} is also a natural candidate for systems beyond the M/G/1. One important direction would be multiserver systems. Recent results have shown that simple priority policies for optimizing mean response time in the M/G/1 generalize well to a wide variety of multiserver systems \citep{grosof_load_2019, grosof_srpt_2018, scully_new_2022, scully_gittins_2020, grosof_optimal_2022, grosof_optimal_2023, hong_performance_2024}. Can we leverage similar techniques to analyze \boost{} in the M/G/$k$ or dispatching settings? Similarly, one could ask about non-Poisson arrival processes. We suspect that \boost{\gamma} is also strongly-tail optimal in models like the G/G/1 with light-tailed job sizes, as long as we still have $\P{T_\fcfs > t} \sim C_\fcfs \exp(-\gamma t)$.

\subsection{Metrics beyond the tail constant}

One could reexamine the metric or metrics being optimized. We now know how to minimize mean response time, namely using SRPT, and the tail constant, namely using \boost{\gamma}. Can we characterize the Pareto frontier of achievable means and tail constants, and can we design policies to achieve the entire frontier? Because the tail constant is a purely asymptotic notion, we conjecture that, at least theoretically, it is possible to design a policy with mean response time arbitrarily close to SRPT's and tail constant arbitrarily close to \boost{\gamma}'s. But whether such a policy would have good tail performance in practice remains to be seen.

Other important metrics arise in systems with multiple priority classes, where we want to promise better performance to better priority classes. Can we design policies to optimize a \emph{weighted} tail constant, meaning a convex combination of each class's tail constant? It seems likely that reducing to a weighted batch problem would yield a version of \boost{\gamma} that accounts for the weights.

But what if we also want to balance a tradeoff between the \emph{decay rates} of the different classes' response times? This seems like a more challenging problem, but we expect it to be important for settings where the top priority classes must have very low response times. The results of \citet{stolyar_largest_2001} show how to balance decay rates optimally using what is essentially an accumulating priority policy \citep{fajardo_controlling_2015, fajardo_waiting_2017, stanford_waiting_2014}, albeit in a slightly weaker sense than our \cref{def:tail-optimal}. Combining \boost{} with accumulating priority thus seems like a promising direction for further exploration.

\subsection{Boost in practice}

We are excited about the potential of deploying \boost{} in real-world computer systems, due to its promising simulation performance and simplicity of implementation. The simplicity angle is especially important in high-performance environments such as network switches. \boost{} fits into the recently proposed Priority-In, First-Out (PIFO) scheduling abstraction \citep{sivaraman_programmable_2016}. One question is whether \boost{} is amenable to being approximated using methods like SP-PIFO \citep{alcoz_sppifo_2020}, which would allow it to be deployed more easily on common network hardware.

Another question is how to determine $\gamma$ in practice, given that the job size distribution may be unknown, or may change over time. We suspect that the best way to obtain $\gamma$ in practice is to directly measure the empirical decay rate of the work distribution. This seems feasible to do in many computer systems, and it seems simpler than trying to estimate the full job size distribution. Our simulations show that \boost{} is robust to some misspecification of $\gamma$, so such an empirical estimate would likely suffice.

Of course, computer systems are just one of many domains where one might consider deploying \boost{}. Different domains have different prioritization needs, which could be met by different boost functions. It is worth highlighting that \boost{} is weakly tail-optimal for any boost function satisfying the relatively mild condition of \cref{eq:model:boost_okay}. Even jobs that receive zero boost experience weakly optimal response time tail. This means that for systems with light-tailed job sizes, \boost{} offers a flexible framework for priority scheduling without sacrificing tail performance.

\begin{acks}
    We thank Isaac Grosof, Benny Van Houdt, and David Shmoys for helpful discussions. This work was supported by the \grantsponsor{nsf}{National Science Foundation}{https://www.nsf.gov} under grant nos. \grantnum{nsf}{CMMI-2307008}, \grantnum{nsf}{DMS-2023528}, and \grantnum{nsf}{DMS-2022448}.
\end{acks}

\bibliographystyle{ACM-Reference-Format}
\bibliography{refs}


\begin{thebibliography}{48}

%%% the Web address.

\ifx \showCODEN    \undefined \def \showCODEN     #1{\unskip}     \fi
\ifx \showDOI      \undefined \def \showDOI       #1{#1}\fi
\ifx \showISBNx    \undefined \def \showISBNx     #1{\unskip}     \fi
\ifx \showISBNxiii \undefined \def \showISBNxiii  #1{\unskip}     \fi
\ifx \showISSN     \undefined \def \showISSN      #1{\unskip}     \fi
\ifx \showLCCN     \undefined \def \showLCCN      #1{\unskip}     \fi
\ifx \shownote     \undefined \def \shownote      #1{#1}          \fi
\ifx \showarticletitle \undefined \def \showarticletitle #1{#1}   \fi
\ifx \showURL      \undefined \def \showURL       {\relax}        \fi
% invisible to TeX
\providecommand\bibfield[2]{#2}
\providecommand\bibinfo[2]{#2}
\providecommand\natexlab[1]{#1}
\providecommand\showeprint[2][]{arXiv:#2}

\bibitem[Abate et~al\mbox{.}(1994)]%
        {abate_waitingtime_1994}
\bibfield{author}{\bibinfo{person}{Joseph Abate}, \bibinfo{person}{Gagan~L. Choudhury}, {and} \bibinfo{person}{Ward Whitt}.} \bibinfo{year}{1994}\natexlab{}.
\newblock \showarticletitle{Waiting-Time Tail Probabilities in Queues with Long-Tail Service-Time Distributions}.
\newblock \bibinfo{journal}{\emph{Queueing Systems}} \bibinfo{volume}{16}, \bibinfo{number}{3-4} (\bibinfo{date}{Sept.} \bibinfo{year}{1994}), \bibinfo{pages}{311--338}.
\newblock
\showISSN{0257-0130, 1572-9443}
\urldef\tempurl%
\url{https://doi.org/10.1007/BF01158960}
\showDOI{\tempurl}


\bibitem[Abate and Whitt(1997)]%
        {abate_asymptotics_1997}
\bibfield{author}{\bibinfo{person}{Joseph Abate} {and} \bibinfo{person}{Ward Whitt}.} \bibinfo{year}{1997}\natexlab{}.
\newblock \showarticletitle{Asymptotics for {$M/G/1$} Low-Priority Waiting-Time Tail Probabilities}.
\newblock \bibinfo{journal}{\emph{Queueing Systems}} \bibinfo{volume}{25}, \bibinfo{number}{1} (\bibinfo{date}{June} \bibinfo{year}{1997}), \bibinfo{pages}{173--233}.
\newblock
\showISSN{1572-9443}
\urldef\tempurl%
\url{https://doi.org/10.1023/A:1019104402024}
\showDOI{\tempurl}


\bibitem[Alcoz et~al\mbox{.}(2020)]%
        {alcoz_sppifo_2020}
\bibfield{author}{\bibinfo{person}{Albert~Gran Alcoz}, \bibinfo{person}{Alexander Dietm{\"u}ller}, {and} \bibinfo{person}{Laurent Vanbever}.} \bibinfo{year}{2020}\natexlab{}.
\newblock \showarticletitle{{{SP-PIFO}}: Approximating Push-in First-out Behaviors Using Strict-Priority Queues}. In \bibinfo{booktitle}{\emph{17th {{USENIX}} Symposium on Networked Systems Design and Implementation ({{NSDI}} 2020)}}. \bibinfo{publisher}{USENIX Association}, \bibinfo{address}{Santa Clara, CA}, \bibinfo{pages}{59--76}.
\newblock
\showISBNx{978-1-939133-13-7}
\urldef\tempurl%
\url{https://www.usenix.org/conference/nsdi20/presentation/alcoz}
\showURL{%
\tempurl}


\bibitem[Baccelli and Br{\'e}maud(2003)]%
        {baccelli_elements_2003}
\bibfield{author}{\bibinfo{person}{Fran{\c c}ois Baccelli} {and} \bibinfo{person}{Pierre Br{\'e}maud}.} \bibinfo{year}{2003}\natexlab{}.
\newblock \bibinfo{booktitle}{\emph{Elements of Queueing Theory: {{Palm}} Martingale Calculus and Stochastic Recurrences} (\bibinfo{edition}{2} ed.)}.
\newblock Number~26 in \bibinfo{series}{Stochastic {{Modelling}} and {{Applied Probability}}}. \bibinfo{publisher}{Springer}, \bibinfo{address}{Berlin, Germany}.
\newblock
\showISBNx{978-3-540-66088-0}
\showLCCN{QA274.8 .B33 2003}
\urldef\tempurl%
\url{https://doi.org/10.1007/978-3-662-11657-9}
\showDOI{\tempurl}


\bibitem[Bingham et~al\mbox{.}(1987)]%
        {bingham_regular_1987}
\bibfield{author}{\bibinfo{person}{Nicholas~H. Bingham}, \bibinfo{person}{Charles~M. Goldie}, {and} \bibinfo{person}{Jef~L. Teugels}.} \bibinfo{year}{1987}\natexlab{}.
\newblock \bibinfo{booktitle}{\emph{Regular Variation}}.
\newblock Number~27 in \bibinfo{series}{Encyclopedia of {{Mathematics}} and Its {{Applications}}}. \bibinfo{publisher}{Cambridge University Press}, \bibinfo{address}{Cambridge, UK}.
\newblock
\showISBNx{978-0-521-30787-1}
\showLCCN{QA331.5 .B54 1987}


\bibitem[Borst et~al\mbox{.}(2003)]%
        {borst_impact_2003}
\bibfield{author}{\bibinfo{person}{Sem~C. Borst}, \bibinfo{person}{Onno~J. Boxma}, \bibinfo{person}{Rudesindo {N{\'u}{\~n}ez-Queija}}, {and} \bibinfo{person}{Bert Zwart}.} \bibinfo{year}{2003}\natexlab{}.
\newblock \showarticletitle{The Impact of the Service Discipline on Delay Asymptotics}.
\newblock \bibinfo{journal}{\emph{Performance Evaluation}} \bibinfo{volume}{54}, \bibinfo{number}{2} (\bibinfo{date}{Oct.} \bibinfo{year}{2003}), \bibinfo{pages}{175--206}.
\newblock
\showISSN{01665316}
\urldef\tempurl%
\url{https://doi.org/10.1016/S0166-5316(03)00071-3}
\showDOI{\tempurl}


\bibitem[Boxma and Zwart(2007)]%
        {boxma_tails_2007}
\bibfield{author}{\bibinfo{person}{Onno~J. Boxma} {and} \bibinfo{person}{Bert Zwart}.} \bibinfo{year}{2007}\natexlab{}.
\newblock \showarticletitle{Tails in Scheduling}.
\newblock \bibinfo{journal}{\emph{ACM SIGMETRICS Performance Evaluation Review}} \bibinfo{volume}{34}, \bibinfo{number}{4} (\bibinfo{date}{March} \bibinfo{year}{2007}), \bibinfo{pages}{13--20}.
\newblock
\showISSN{0163-5999}
\urldef\tempurl%
\url{https://doi.org/10.1145/1243401.1243406}
\showDOI{\tempurl}


\bibitem[Br{\'e}maud(2020)]%
        {bremaud_probability_2020}
\bibfield{author}{\bibinfo{person}{Pierre Br{\'e}maud}.} \bibinfo{year}{2020}\natexlab{}.
\newblock \bibinfo{booktitle}{\emph{Probability Theory and Stochastic Processes}}.
\newblock \bibinfo{publisher}{Springer}, \bibinfo{address}{Cham, Switzerland}.
\newblock
\showISBNx{978-3-030-40182-5 978-3-030-40183-2}
\urldef\tempurl%
\url{https://doi.org/10.1007/978-3-030-40183-2}
\showDOI{\tempurl}


\bibitem[Charlet and Van~Houdt(2024)]%
        {charlet_tail_2024}
\bibfield{author}{\bibinfo{person}{Nils Charlet} {and} \bibinfo{person}{Benny Van~Houdt}.} \bibinfo{year}{2024}\natexlab{}.
\newblock \bibinfo{title}{Tail Optimality and Performance Analysis of the {{Nudge-M}} Scheduling Algorithm}.
\newblock
\newblock
\showeprint[arxiv]{2403.06588}~[cs, math]
\urldef\tempurl%
\url{http://arxiv.org/abs/2403.06588}
\showURL{%
\tempurl}


\bibitem[Chen and Dong(2021)]%
        {chen_scheduling_2021}
\bibfield{author}{\bibinfo{person}{Yan Chen} {and} \bibinfo{person}{Jing Dong}.} \bibinfo{year}{2021}\natexlab{}.
\newblock \bibinfo{title}{Scheduling with Service-Time Information: The Power of Two Priority Classes}.
\newblock
\newblock
\showeprint[arxiv]{2105.10499}~[cs, math]
\urldef\tempurl%
\url{http://arxiv.org/abs/2105.10499}
\showURL{%
\tempurl}


\bibitem[Fajardo and Drekic(2015)]%
        {fajardo_controlling_2015}
\bibfield{author}{\bibinfo{person}{Val~Andrei Fajardo} {and} \bibinfo{person}{Steve Drekic}.} \bibinfo{year}{2015}\natexlab{}.
\newblock \showarticletitle{Controlling the Workload of {$M/G/1$} Queues via the {$q$}-Policy}.
\newblock \bibinfo{journal}{\emph{European Journal of Operational Research}} \bibinfo{volume}{243}, \bibinfo{number}{2} (\bibinfo{date}{June} \bibinfo{year}{2015}), \bibinfo{pages}{607--617}.
\newblock
\showISSN{0377-2217}
\urldef\tempurl%
\url{https://doi.org/10.1016/j.ejor.2014.12.036}
\showDOI{\tempurl}


\bibitem[Fajardo and Drekic(2017)]%
        {fajardo_waiting_2017}
\bibfield{author}{\bibinfo{person}{Val~Andrei Fajardo} {and} \bibinfo{person}{Steve Drekic}.} \bibinfo{year}{2017}\natexlab{}.
\newblock \showarticletitle{Waiting Time Distributions in the Preemptive Accumulating Priority Queue}.
\newblock \bibinfo{journal}{\emph{Methodology and Computing in Applied Probability}} \bibinfo{volume}{19}, \bibinfo{number}{1} (\bibinfo{date}{March} \bibinfo{year}{2017}), \bibinfo{pages}{255--284}.
\newblock
\showISSN{1573-7713}
\urldef\tempurl%
\url{https://doi.org/10.1007/s11009-015-9476-1}
\showDOI{\tempurl}


\bibitem[Friedman and Henderson(2003)]%
        {friedman_fairness_2003}
\bibfield{author}{\bibinfo{person}{Eric~J. Friedman} {and} \bibinfo{person}{Shane~G. Henderson}.} \bibinfo{year}{2003}\natexlab{}.
\newblock \showarticletitle{Fairness and Efficiency in Web Server Protocols}.
\newblock \bibinfo{journal}{\emph{ACM SIGMETRICS Performance Evaluation Review}} \bibinfo{volume}{31}, \bibinfo{number}{1} (\bibinfo{date}{June} \bibinfo{year}{2003}), \bibinfo{pages}{229--237}.
\newblock
\showISSN{0163-5999}
\urldef\tempurl%
\url{https://doi.org/10.1145/885651.781056}
\showDOI{\tempurl}


\bibitem[Friedman and Hurley(2003)]%
        {friedman_protective_2003}
\bibfield{author}{\bibinfo{person}{Eric~J. Friedman} {and} \bibinfo{person}{Gavin Hurley}.} \bibinfo{year}{2003}\natexlab{}.
\newblock \bibinfo{booktitle}{\emph{Protective Scheduling}}.
\newblock \bibinfo{type}{Technical {{Report}}} 1373. \bibinfo{institution}{Cornell University}, \bibinfo{address}{Ithaca, NY}. \bibinfo{pages}{11} pages.
\newblock
\urldef\tempurl%
\url{https://hdl.handle.net/1813/9250}
\showURL{%
\tempurl}


\bibitem[Gavish and Schweitzer(1977)]%
        {gavish_markovian_1977}
\bibfield{author}{\bibinfo{person}{Bezalel Gavish} {and} \bibinfo{person}{Paul~J. Schweitzer}.} \bibinfo{year}{1977}\natexlab{}.
\newblock \showarticletitle{The {{Markovian}} Queue with Bounded Waiting Time}.
\newblock \bibinfo{journal}{\emph{Management Science}} \bibinfo{volume}{23}, \bibinfo{number}{12} (\bibinfo{date}{Aug.} \bibinfo{year}{1977}), \bibinfo{pages}{1349--1357}.
\newblock
\showISSN{0025-1909, 1526-5501}
\urldef\tempurl%
\url{https://doi.org/10.1287/mnsc.23.12.1349}
\showDOI{\tempurl}


\bibitem[Gittins et~al\mbox{.}(2011)]%
        {gittins_multiarmed_2011}
\bibfield{author}{\bibinfo{person}{John~C. Gittins}, \bibinfo{person}{Kevin~D. Glazebrook}, {and} \bibinfo{person}{Richard~R. Weber}.} \bibinfo{year}{2011}\natexlab{}.
\newblock \bibinfo{booktitle}{\emph{Multi-Armed Bandit Allocation Indices} (\bibinfo{edition}{2} ed.)}.
\newblock \bibinfo{publisher}{Wiley}, \bibinfo{address}{Chichester, UK}.
\newblock
\showISBNx{978-0-470-67002-6}
\showLCCN{QA279 .G55 2011}


\bibitem[Grosof(2023)]%
        {grosof_optimal_2023}
\bibfield{author}{\bibinfo{person}{Isaac Grosof}.} \bibinfo{year}{2023}\natexlab{}.
\newblock \emph{\bibinfo{title}{Optimal Scheduling in Multiserver Queues}}.
\newblock \bibinfo{thesistype}{Ph.\,D. Dissertation}. \bibinfo{school}{Carnegie Mellon University}, \bibinfo{address}{Pittsburgh, PA}.
\newblock
\urldef\tempurl%
\url{https://isaacg1.github.io/assets/isaac-thesis.pdf}
\showURL{%
\tempurl}


\bibitem[Grosof et~al\mbox{.}(2018)]%
        {grosof_srpt_2018}
\bibfield{author}{\bibinfo{person}{Isaac Grosof}, \bibinfo{person}{Ziv Scully}, {and} \bibinfo{person}{Mor {Harchol-Balter}}.} \bibinfo{year}{2018}\natexlab{}.
\newblock \showarticletitle{{{SRPT}} for Multiserver Systems}.
\newblock \bibinfo{journal}{\emph{Performance Evaluation}}  \bibinfo{volume}{127--128} (\bibinfo{date}{Nov.} \bibinfo{year}{2018}), \bibinfo{pages}{154--175}.
\newblock
\showISSN{01665316}
\urldef\tempurl%
\url{https://doi.org/10.1016/j.peva.2018.10.001}
\showDOI{\tempurl}


\bibitem[Grosof et~al\mbox{.}(2019)]%
        {grosof_load_2019}
\bibfield{author}{\bibinfo{person}{Isaac Grosof}, \bibinfo{person}{Ziv Scully}, {and} \bibinfo{person}{Mor {Harchol-Balter}}.} \bibinfo{year}{2019}\natexlab{}.
\newblock \showarticletitle{Load Balancing Guardrails: Keeping Your Heavy Traffic on the Road to Low Response Times}.
\newblock \bibinfo{journal}{\emph{Proceedings of the ACM on Measurement and Analysis of Computing Systems}} \bibinfo{volume}{3}, \bibinfo{number}{2}, Article \bibinfo{articleno}{42} (\bibinfo{date}{June} \bibinfo{year}{2019}), \bibinfo{numpages}{31}~pages.
\newblock
\showISSN{2476-1249, 2476-1249}
\urldef\tempurl%
\url{https://doi.org/10.1145/3341617.3326157}
\showDOI{\tempurl}


\bibitem[Grosof et~al\mbox{.}(2022)]%
        {grosof_optimal_2022}
\bibfield{author}{\bibinfo{person}{Isaac Grosof}, \bibinfo{person}{Ziv Scully}, \bibinfo{person}{Mor {Harchol-Balter}}, {and} \bibinfo{person}{Alan {Scheller-Wolf}}.} \bibinfo{year}{2022}\natexlab{}.
\newblock \showarticletitle{Optimal Scheduling in the Multiserver-Job Model under Heavy Traffic}.
\newblock \bibinfo{journal}{\emph{Proceedings of the ACM on Measurement and Analysis of Computing Systems}} \bibinfo{volume}{6}, \bibinfo{number}{3}, Article \bibinfo{articleno}{51} (\bibinfo{date}{Dec.} \bibinfo{year}{2022}), \bibinfo{numpages}{32}~pages.
\newblock
\showISSN{2476-1249}
\urldef\tempurl%
\url{https://doi.org/10.1145/3570612}
\showDOI{\tempurl}


\bibitem[Grosof et~al\mbox{.}(2021)]%
        {grosof_nudge_2021}
\bibfield{author}{\bibinfo{person}{Isaac Grosof}, \bibinfo{person}{Kunhe Yang}, \bibinfo{person}{Ziv Scully}, {and} \bibinfo{person}{Mor {Harchol-Balter}}.} \bibinfo{year}{2021}\natexlab{}.
\newblock \showarticletitle{Nudge: Stochastically Improving upon {{FCFS}}}.
\newblock \bibinfo{journal}{\emph{Proceedings of the ACM on Measurement and Analysis of Computing Systems}} \bibinfo{volume}{5}, \bibinfo{number}{2}, Article \bibinfo{articleno}{21} (\bibinfo{date}{June} \bibinfo{year}{2021}), \bibinfo{numpages}{29}~pages.
\newblock
\urldef\tempurl%
\url{https://doi.org/10.1145/3460088}
\showDOI{\tempurl}


\bibitem[{Harchol-Balter}(2013)]%
        {harchol-balter_performance_2013}
\bibfield{author}{\bibinfo{person}{Mor {Harchol-Balter}}.} \bibinfo{year}{2013}\natexlab{}.
\newblock \bibinfo{booktitle}{\emph{Performance Modeling and Design of Computer Systems: Queueing Theory in Action}}.
\newblock \bibinfo{publisher}{Cambridge University Press}, \bibinfo{address}{Cambridge, UK}.
\newblock
\showISBNx{978-1-107-02750-3}
\showLCCN{QA76.545 .H37 2013}


\bibitem[{Harchol-Balter} et~al\mbox{.}(2003)]%
        {harchol-balter_sizebased_2003}
\bibfield{author}{\bibinfo{person}{Mor {Harchol-Balter}}, \bibinfo{person}{Bianca Schroeder}, \bibinfo{person}{Nikhil Bansal}, {and} \bibinfo{person}{Mukesh Agrawal}.} \bibinfo{year}{2003}\natexlab{}.
\newblock \showarticletitle{Size-Based Scheduling to Improve Web Performance}.
\newblock \bibinfo{journal}{\emph{ACM Transactions on Computer Systems}} \bibinfo{volume}{21}, \bibinfo{number}{2} (\bibinfo{date}{May} \bibinfo{year}{2003}), \bibinfo{pages}{207--233}.
\newblock
\showISSN{0734-2071, 1557-7333}
\urldef\tempurl%
\url{https://doi.org/10.1145/762483.762486}
\showDOI{\tempurl}


\bibitem[Hong and Scully(2024)]%
        {hong_performance_2024}
\bibfield{author}{\bibinfo{person}{Yige Hong} {and} \bibinfo{person}{Ziv Scully}.} \bibinfo{year}{2024}\natexlab{}.
\newblock \showarticletitle{Performance of the {{Gittins}} Policy in the {{G}}/{{G}}/1 and {{G}}/{{G}}/{$k$}, with and without Setup Times}.
\newblock \bibinfo{journal}{\emph{Performance Evaluation}}  \bibinfo{volume}{163}, Article \bibinfo{articleno}{102377} (\bibinfo{date}{Jan.} \bibinfo{year}{2024}), \bibinfo{numpages}{26}~pages.
\newblock
\showISSN{01665316}
\urldef\tempurl%
\url{https://doi.org/10.1016/j.peva.2023.102377}
\showDOI{\tempurl}


\bibitem[Kingman(1993)]%
        {kingman_poisson_1993}
\bibfield{author}{\bibinfo{person}{J.~F.~C. Kingman}.} \bibinfo{year}{1993}\natexlab{}.
\newblock \bibinfo{booktitle}{\emph{Poisson Processes}}.
\newblock Number~3 in \bibinfo{series}{Oxford Studies in Probability}. \bibinfo{publisher}{Oxford University Press}, \bibinfo{address}{Oxford}.
\newblock
\showISBNx{978-0-19-853693-2}
\showLCCN{QA274.42 .K56 1993}


\bibitem[Kontorovich and Pinelis(2019)]%
        {kontorovich_exact_2019}
\bibfield{author}{\bibinfo{person}{Aryeh Kontorovich} {and} \bibinfo{person}{Iosif Pinelis}.} \bibinfo{year}{2019}\natexlab{}.
\newblock \showarticletitle{Exact Lower Bounds for the Agnostic Probably-Approximately-Correct ({{PAC}}) Machine Learning Model}.
\newblock \bibinfo{journal}{\emph{The Annals of Statistics}} \bibinfo{volume}{47}, \bibinfo{number}{5} (\bibinfo{date}{Oct.} \bibinfo{year}{2019}), \bibinfo{pages}{2822--2854}.
\newblock
\showISSN{0090-5364}
\urldef\tempurl%
\url{https://doi.org/10.1214/18-AOS1766}
\showDOI{\tempurl}


\bibitem[Lenstra and Shmoys(2020)]%
        {lenstra_elements_2020}
\bibfield{author}{\bibinfo{person}{Jan~Karel Lenstra} {and} \bibinfo{person}{David~B. Shmoys}.} \bibinfo{year}{2020}\natexlab{}.
\newblock \bibinfo{title}{Elements of Scheduling}.
\newblock
\newblock
\showeprint[arxiv]{2001.06005}~[cs]
\urldef\tempurl%
\url{http://arxiv.org/abs/2001.06005}
\showURL{%
\tempurl}


\bibitem[Marin et~al\mbox{.}(2020)]%
        {marin_sizebased_2020}
\bibfield{author}{\bibinfo{person}{Andrea Marin}, \bibinfo{person}{Sabina Rossi}, {and} \bibinfo{person}{Carlo Zen}.} \bibinfo{year}{2020}\natexlab{}.
\newblock \showarticletitle{Size-Based Scheduling for {{TCP}} Flows: Implementation and Performance Evaluation}.
\newblock \bibinfo{journal}{\emph{Computer Networks}}  \bibinfo{volume}{183}, Article \bibinfo{articleno}{107574} (\bibinfo{date}{Dec.} \bibinfo{year}{2020}), \bibinfo{numpages}{15}~pages.
\newblock
\showISSN{1389-1286}
\urldef\tempurl%
\url{https://doi.org/10.1016/j.comnet.2020.107574}
\showDOI{\tempurl}


\bibitem[Nair et~al\mbox{.}(2010)]%
        {nair_tailrobust_2010}
\bibfield{author}{\bibinfo{person}{Jayakrishnan Nair}, \bibinfo{person}{Adam Wierman}, {and} \bibinfo{person}{Bert Zwart}.} \bibinfo{year}{2010}\natexlab{}.
\newblock \showarticletitle{Tail-Robust Scheduling via Limited Processor Sharing}.
\newblock \bibinfo{journal}{\emph{Performance Evaluation}} \bibinfo{volume}{67}, \bibinfo{number}{11} (\bibinfo{date}{Nov.} \bibinfo{year}{2010}), \bibinfo{pages}{978--995}.
\newblock
\showISSN{0166-5316}
\urldef\tempurl%
\url{https://doi.org/10.1016/j.peva.2010.08.012}
\showDOI{\tempurl}


\bibitem[{N{\'u}{\~n}ez-Queija}(2002)]%
        {nunez-queija_queues_2002}
\bibfield{author}{\bibinfo{person}{Rudesindo {N{\'u}{\~n}ez-Queija}}.} \bibinfo{year}{2002}\natexlab{}.
\newblock \showarticletitle{Queues with Equally Heavy Sojourn Time and Service Requirement Distributions}.
\newblock \bibinfo{journal}{\emph{Annals of Operations Research}} \bibinfo{volume}{113}, \bibinfo{number}{1/4} (\bibinfo{date}{July} \bibinfo{year}{2002}), \bibinfo{pages}{101--117}.
\newblock
\showISSN{02545330}
\urldef\tempurl%
\url{https://doi.org/10.1023/A:1020905810996}
\showDOI{\tempurl}


\bibitem[Nuyens et~al\mbox{.}(2008)]%
        {nuyens_preventing_2008}
\bibfield{author}{\bibinfo{person}{Misja Nuyens}, \bibinfo{person}{Adam Wierman}, {and} \bibinfo{person}{Bert Zwart}.} \bibinfo{year}{2008}\natexlab{}.
\newblock \showarticletitle{Preventing Large Sojourn Times Using {{SMART}} Scheduling}.
\newblock \bibinfo{journal}{\emph{Operations Research}} \bibinfo{volume}{56}, \bibinfo{number}{1} (\bibinfo{date}{Feb.} \bibinfo{year}{2008}), \bibinfo{pages}{88--101}.
\newblock
\showISSN{0030-364X, 1526-5463}
\urldef\tempurl%
\url{https://doi.org/10.1287/opre.1070.0504}
\showDOI{\tempurl}


\bibitem[Nuyens and Zwart(2006)]%
        {nuyens_largedeviations_2006}
\bibfield{author}{\bibinfo{person}{Misja Nuyens} {and} \bibinfo{person}{Bert Zwart}.} \bibinfo{year}{2006}\natexlab{}.
\newblock \showarticletitle{A Large-Deviations Analysis of the {$GI/GI/1$} {{SRPT}} Queue}.
\newblock \bibinfo{journal}{\emph{Queueing Systems}} \bibinfo{volume}{54}, \bibinfo{number}{2} (\bibinfo{date}{Oct.} \bibinfo{year}{2006}), \bibinfo{pages}{85--97}.
\newblock
\showISSN{1572-9443}
\urldef\tempurl%
\url{https://doi.org/10.1007/s11134-006-8767-1}
\showDOI{\tempurl}


\bibitem[Pinedo(2016)]%
        {pinedo_scheduling_2016}
\bibfield{author}{\bibinfo{person}{Michael Pinedo}.} \bibinfo{year}{2016}\natexlab{}.
\newblock \bibinfo{booktitle}{\emph{Scheduling: Theory, Algorithms, and Systems} (\bibinfo{edition}{5} ed.)}.
\newblock \bibinfo{publisher}{Springer}, \bibinfo{address}{Cham, Switzerland}.
\newblock
\showISBNx{978-3-319-26578-0}


\bibitem[Scully(2022)]%
        {scully_new_2022}
\bibfield{author}{\bibinfo{person}{Ziv Scully}.} \bibinfo{year}{2022}\natexlab{}.
\newblock \emph{\bibinfo{title}{A New Toolbox for Scheduling Theory}}.
\newblock \bibinfo{thesistype}{Ph.\,D. Dissertation}. \bibinfo{school}{Carnegie Mellon University}, \bibinfo{address}{Pittsburgh, PA}.
\newblock
\urldef\tempurl%
\url{https://ziv.codes/pdf/scully-thesis.pdf}
\showURL{%
\tempurl}


\bibitem[Scully et~al\mbox{.}(2020a)]%
        {scully_gittins_2020}
\bibfield{author}{\bibinfo{person}{Ziv Scully}, \bibinfo{person}{Isaac Grosof}, {and} \bibinfo{person}{Mor {Harchol-Balter}}.} \bibinfo{year}{2020}\natexlab{a}.
\newblock \showarticletitle{The {{Gittins}} Policy Is Nearly Optimal in the {{M}}/{{G}}/{$k$} under Extremely General Conditions}.
\newblock \bibinfo{journal}{\emph{Proceedings of the ACM on Measurement and Analysis of Computing Systems}} \bibinfo{volume}{4}, \bibinfo{number}{3}, Article \bibinfo{articleno}{43} (\bibinfo{date}{Nov.} \bibinfo{year}{2020}), \bibinfo{numpages}{29}~pages.
\newblock
\showISSN{2476-1249, 2476-1249}
\urldef\tempurl%
\url{https://doi.org/10.1145/3428328}
\showDOI{\tempurl}


\bibitem[Scully and {Harchol-Balter}(2018)]%
        {scully_soap_2018a}
\bibfield{author}{\bibinfo{person}{Ziv Scully} {and} \bibinfo{person}{Mor {Harchol-Balter}}.} \bibinfo{year}{2018}\natexlab{}.
\newblock \showarticletitle{{{SOAP}} Bubbles: Robust Scheduling under Adversarial Noise}. In \bibinfo{booktitle}{\emph{56th {{Annual Allerton Conference}} on {{Communication}}, {{Control}}, and {{Computing}}}}. \bibinfo{publisher}{IEEE}, \bibinfo{address}{Monticello, IL}, \bibinfo{pages}{144--154}.
\newblock
\showISBNx{978-1-5386-6596-1}
\urldef\tempurl%
\url{https://doi.org/10.1109/ALLERTON.2018.8635963}
\showDOI{\tempurl}


\bibitem[Scully and {Harchol-Balter}(2021)]%
        {scully_gittins_2021}
\bibfield{author}{\bibinfo{person}{Ziv Scully} {and} \bibinfo{person}{Mor {Harchol-Balter}}.} \bibinfo{year}{2021}\natexlab{}.
\newblock \showarticletitle{The {{Gittins}} Policy in the {{M}}/{{G}}/1 Queue}. In \bibinfo{booktitle}{\emph{19th {{International Symposium}} on {{Modeling}} and {{Optimization}} in {{Mobile}}, {{Ad Hoc}}, and {{Wireless Networks}} ({{WiOpt}} 2021)}}. \bibinfo{publisher}{IFIP}, \bibinfo{address}{Philadelphia, PA}, \bibinfo{pages}{248--255}.
\newblock
\urldef\tempurl%
\url{https://doi.org/10.23919/WiOpt52861.2021.9589051}
\showDOI{\tempurl}


\bibitem[Scully et~al\mbox{.}(2018)]%
        {scully_soap_2018}
\bibfield{author}{\bibinfo{person}{Ziv Scully}, \bibinfo{person}{Mor {Harchol-Balter}}, {and} \bibinfo{person}{Alan {Scheller-Wolf}}.} \bibinfo{year}{2018}\natexlab{}.
\newblock \showarticletitle{{{SOAP}}: One Clean Analysis of All Age-Based Scheduling Policies}.
\newblock \bibinfo{journal}{\emph{Proceedings of the ACM on Measurement and Analysis of Computing Systems}} \bibinfo{volume}{2}, \bibinfo{number}{1}, Article \bibinfo{articleno}{16} (\bibinfo{date}{April} \bibinfo{year}{2018}), \bibinfo{numpages}{30}~pages.
\newblock
\showISSN{2476-1249, 2476-1249}
\urldef\tempurl%
\url{https://doi.org/10.1145/3179419}
\showDOI{\tempurl}


\bibitem[Scully and van Kreveld(2024)]%
        {scully_when_2024}
\bibfield{author}{\bibinfo{person}{Ziv Scully} {and} \bibinfo{person}{Lucas van Kreveld}.} \bibinfo{year}{2024}\natexlab{}.
\newblock \showarticletitle{When Does the {{Gittins}} Policy Have Asymptotically Optimal Response Time Tail in the {{M}}/{{G}}/1?}
\newblock \bibinfo{journal}{\emph{Operations Research}} \bibinfo{volume}{72}, \bibinfo{number}{2} (\bibinfo{date}{Feb.} \bibinfo{year}{2024}).
\newblock
\showISSN{0030-364X, 1526-5463}
\urldef\tempurl%
\url{https://doi.org/10.1287/opre.2022.0038}
\showDOI{\tempurl}


\bibitem[Scully et~al\mbox{.}(2020b)]%
        {scully_characterizing_2020}
\bibfield{author}{\bibinfo{person}{Ziv Scully}, \bibinfo{person}{Lucas van Kreveld}, \bibinfo{person}{Onno~J. Boxma}, \bibinfo{person}{Jan-Pieter Dorsman}, {and} \bibinfo{person}{Adam Wierman}.} \bibinfo{year}{2020}\natexlab{b}.
\newblock \showarticletitle{Characterizing Policies with Optimal Response Time Tails under Heavy-Tailed Job Sizes}.
\newblock \bibinfo{journal}{\emph{Proceedings of the ACM on Measurement and Analysis of Computing Systems}} \bibinfo{volume}{4}, \bibinfo{number}{2}, Article \bibinfo{articleno}{30} (\bibinfo{date}{June} \bibinfo{year}{2020}), \bibinfo{numpages}{33}~pages.
\newblock
\showISSN{2476-1249, 2476-1249}
\urldef\tempurl%
\url{https://doi.org/10.1145/3392148}
\showDOI{\tempurl}


\bibitem[Shanthikumar and Sumita(1987)]%
        {shanthikumar_convex_1987}
\bibfield{author}{\bibinfo{person}{J.~George Shanthikumar} {and} \bibinfo{person}{Ushio Sumita}.} \bibinfo{year}{1987}\natexlab{}.
\newblock \showarticletitle{Convex Ordering of Sojourn Times in Single-Server Queues: Extremal Properties of {{FIFO}} and {{LIFO}} Service Disciplines}.
\newblock \bibinfo{journal}{\emph{Journal of Applied Probability}} \bibinfo{volume}{24}, \bibinfo{number}{3} (\bibinfo{date}{Sept.} \bibinfo{year}{1987}), \bibinfo{pages}{737--748}.
\newblock
\showISSN{0021-9002, 1475-6072}
\urldef\tempurl%
\url{https://doi.org/10.2307/3214103}
\showDOI{\tempurl}


\bibitem[Sivaraman et~al\mbox{.}(2016)]%
        {sivaraman_programmable_2016}
\bibfield{author}{\bibinfo{person}{Anirudh Sivaraman}, \bibinfo{person}{Suvinay Subramanian}, \bibinfo{person}{Mohammad Alizadeh}, \bibinfo{person}{Sharad Chole}, \bibinfo{person}{Shang-Tse Chuang}, \bibinfo{person}{Anurag Agrawal}, \bibinfo{person}{Hari Balakrishnan}, \bibinfo{person}{Tom Edsall}, \bibinfo{person}{Sachin Katti}, {and} \bibinfo{person}{Nick McKeown}.} \bibinfo{year}{2016}\natexlab{}.
\newblock \showarticletitle{Programmable Packet Scheduling at Line Rate}. In \bibinfo{booktitle}{\emph{Proceedings of the 2016 {{ACM SIGCOMM Conference}} ({{SIGCOMM}} 2016)}}. \bibinfo{publisher}{ACM}, \bibinfo{address}{Florianopolis, Brazil}, \bibinfo{pages}{44--57}.
\newblock
\showISBNx{978-1-4503-4193-6}
\urldef\tempurl%
\url{https://doi.org/10.1145/2934872.2934899}
\showDOI{\tempurl}


\bibitem[Stanford et~al\mbox{.}(2014)]%
        {stanford_waiting_2014}
\bibfield{author}{\bibinfo{person}{David~A. Stanford}, \bibinfo{person}{Peter Taylor}, {and} \bibinfo{person}{Ilze Ziedins}.} \bibinfo{year}{2014}\natexlab{}.
\newblock \showarticletitle{Waiting Time Distributions in the Accumulating Priority Queue}.
\newblock \bibinfo{journal}{\emph{Queueing Systems}} \bibinfo{volume}{77}, \bibinfo{number}{3} (\bibinfo{date}{July} \bibinfo{year}{2014}), \bibinfo{pages}{297--330}.
\newblock
\showISSN{0257-0130, 1572-9443}
\urldef\tempurl%
\url{https://doi.org/10.1007/s11134-013-9382-6}
\showDOI{\tempurl}


\bibitem[Stolyar and Ramanan(2001)]%
        {stolyar_largest_2001}
\bibfield{author}{\bibinfo{person}{Alexander~L. Stolyar} {and} \bibinfo{person}{Kavita Ramanan}.} \bibinfo{year}{2001}\natexlab{}.
\newblock \showarticletitle{Largest Weighted Delay First Scheduling: Large Deviations and Optimality}.
\newblock \bibinfo{journal}{\emph{The Annals of Applied Probability}} \bibinfo{volume}{11}, \bibinfo{number}{1} (\bibinfo{date}{Feb.} \bibinfo{year}{2001}), \bibinfo{pages}{1--48}.
\newblock
\showISSN{1050-5164}
\urldef\tempurl%
\url{https://doi.org/10.1214/aoap/998926986}
\showDOI{\tempurl}


\bibitem[Van~Houdt(2022)]%
        {vanhoudt_stochastic_2022}
\bibfield{author}{\bibinfo{person}{Benny Van~Houdt}.} \bibinfo{year}{2022}\natexlab{}.
\newblock \showarticletitle{On the Stochastic and Asymptotic Improvement of First-Come First-Served and Nudge Scheduling}.
\newblock \bibinfo{journal}{\emph{Proceedings of the ACM on Measurement and Analysis of Computing Systems}} \bibinfo{volume}{6}, \bibinfo{number}{3} (\bibinfo{date}{Dec.} \bibinfo{year}{2022}), \bibinfo{pages}{1--22}.
\newblock
\showISSN{2476-1249}
\urldef\tempurl%
\url{https://doi.org/10.1145/3570610}
\showDOI{\tempurl}


\bibitem[Wierman and Zwart(2012)]%
        {wierman_tailoptimal_2012}
\bibfield{author}{\bibinfo{person}{Adam Wierman} {and} \bibinfo{person}{Bert Zwart}.} \bibinfo{year}{2012}\natexlab{}.
\newblock \showarticletitle{Is Tail-Optimal Scheduling Possible?}
\newblock \bibinfo{journal}{\emph{Operations Research}} \bibinfo{volume}{60}, \bibinfo{number}{5} (\bibinfo{date}{Oct.} \bibinfo{year}{2012}), \bibinfo{pages}{1249--1257}.
\newblock
\showISSN{0030-364X, 1526-5463}
\urldef\tempurl%
\url{https://doi.org/10.1287/opre.1120.1086}
\showDOI{\tempurl}


\bibitem[Wolff(1982)]%
        {wolff_poisson_1982}
\bibfield{author}{\bibinfo{person}{Ronald~W. Wolff}.} \bibinfo{year}{1982}\natexlab{}.
\newblock \showarticletitle{Poisson Arrivals See Time Averages}.
\newblock \bibinfo{journal}{\emph{Operations Research}} \bibinfo{volume}{30}, \bibinfo{number}{2} (\bibinfo{year}{1982}), \bibinfo{pages}{223--231}.
\newblock
\showISSN{0030364X, 15265463}
\urldef\tempurl%
\url{http://www.jstor.org/stable/170165}
\showURL{%
\tempurl}


\bibitem[Zwart and Boxma(2000)]%
        {zwart_sojourn_2000}
\bibfield{author}{\bibinfo{person}{Bert Zwart} {and} \bibinfo{person}{Onno~J. Boxma}.} \bibinfo{year}{2000}\natexlab{}.
\newblock \showarticletitle{Sojourn Time Asymptotics in the {{M}}/{{G}}/1 Processor Sharing Queue}.
\newblock \bibinfo{journal}{\emph{Queueing Systems}} \bibinfo{volume}{35}, \bibinfo{number}{1/4} (\bibinfo{year}{2000}), \bibinfo{pages}{141--166}.
\newblock
\showISSN{02570130}
\urldef\tempurl%
\url{https://doi.org/10.1023/A:1019142010994}
\showDOI{\tempurl}


\end{thebibliography}

\appendix

\section{Strong tail optimality for heavy-tailed job size distributions}
\label{sec:heavy-tailed_strong_optimality}
In this appendix, we show that many scheduling policies that are known to be weakly tail-optimal in an M/G/1 with heavy-tailed job size distribution are, in fact, strongly tail-optimal. This confirms conjectures by \citet{boxma_tails_2007} and \citet{wierman_tailoptimal_2012}.

The specific class of heavy-tailed distributions we consider are \emph{regularly varying} distributions \citep{bingham_regular_1987}. These are the distributions~$S$ such that there exists a constant $-\alpha$ such that for all $k > 0$,
\[
    \label{eq:regularly_varying}
    \lim_{t \to \infty} \frac{\P{S > kt}}{\P{S > t}} = k^{-\alpha}.
\]
The constant $-\alpha$ is called the \emph{index} of regular variation of~$S$. For the rest of this appendix, we consider an M/G/1 with regularly varying job size distribution~$S$ with index~$-\alpha$.

There are a number of results showing that under various scheduling policies~$\pi$,
\[
    \label{eq:heavy-tailed_strong_optimality}
    \lim_{t \to \infty} \frac{\P{T_\pi > t}}{\P{S > (1 - \rho) t}} = 1,
\]
where $\rho = \lambda \E{S}$ is the system load. Policies~$\pi$ which satisfy \cref{eq:heavy-tailed_strong_optimality} include classic policies like Processor Sharing, Shortest Remaining Processing Time, and Least Attained Service \citep{boxma_tails_2007, nunez-queija_queues_2002, nuyens_preventing_2008, zwart_sojourn_2000}; as well as more complex policies like Shortest Expected Remaining Processing Time, Randomized Multi-Level Feedback, and the Gittins policy \citep{scully_characterizing_2020, scully_when_2024}.

Motivated by \cref{eq:regularly_varying, eq:heavy-tailed_strong_optimality}, let
\[
    C_\pi = \lim_{t \to \infty} \frac{\P{T_\pi > t}}{\P{S > (1 - \rho) t}}.
\]
Because a job's response time is at least its size, by \cref{eq:regularly_varying}, we have $C_\pi \geq (1 - \rho)^\alpha$ for all policies. Any policy~$\pi$ satisfying \cref{eq:heavy-tailed_strong_optimality} has $C_\pi = 1$ and is thus weakly tail-optimal. The question is whether any policy can achieve $C_\pi < 1$. Below, we use a result of \citet{wierman_tailoptimal_2012} to answer negatively.

\begin{theorem}
    Consider an M/G/1 whose job size distribution has regularly varying tail. Any scheduling policy whose response time distribution satisfies \cref{eq:heavy-tailed_strong_optimality} is strongly tail-optimal, and thus $\inf_\pi C_\pi = 1$.
\end{theorem}

\begin{proof}
    Fix a scheduling policy~$\pi$. The key is a result of \citet{wierman_tailoptimal_2012} which gives a necessary condition to have $C_\pi < \infty$. We show that the condition also implies $C_\pi \geq 1$, as desired.

    We first state a version of the necessary condition. Consider a tagged job arriving at time~$0$, and let $R_\pi(t)$ be the total time during $[0, t]$ for which $\pi$ serves jobs that arrive after the tagged job. If $C_\pi < \infty$, then for all $\delta > 0$ \citep[Proposition~1]{wierman_tailoptimal_2012},
    \[
        \label{eq:heavy-tailed_necessary_condition}
        \lim_{t \to \infty} \P[\big]{R_\pi(t) > (\rho - \delta) t \given S > (1 - \rho + \delta) t} = 1.
    \]

    Notice that $R_\pi(t) > (\rho - \delta)$ and $S > (1 - \rho + \delta) t$ together imply $R_\pi(t) + S > t$. When this occurs, the tagged job does not receive enough service to complete by time~$t$, so its response time satisfies $T_\pi > t$. Therefore, for any $\delta > 0$,
    \[
        1
        &= \lim_{t \to \infty} \P[\big]{R_\pi(t) > (\rho - \delta) t \given S > (1 - \rho + \delta) t} \\
        &\leq \lim_{t \to \infty} \P[\big]{T_\pi > t \given S > (1 - \rho + \delta) t} \\
        &\leq \lim_{t \to \infty} \frac{\P{T_\pi > t}}{\P{S > (1 - \rho + \delta) t}} \\
        &= \gp*{\frac{1 - \rho + \delta}{1 - \rho}}^\alpha \lim_{t \to \infty} \frac{\P{T_\pi > t}}{\P{S > (1 - \rho) t}},
    \]
    where the last line follows from \cref{eq:regularly_varying}. Taking the $\delta \to 0$ limit then implies all policies have $C_\pi \geq 1$. The fact that $\inf_\pi C_\pi = 1$ follows from the fact that multiple policies $\pi$ achieve $C_\pi = 1$ \citep{zwart_sojourn_2000, nunez-queija_queues_2002, nuyens_preventing_2008, boxma_tails_2007, scully_characterizing_2020, scully_when_2024}
\end{proof}

We further conjecture that the regularly varying requirement on $S$ can be relaxed to requiring $S$ be \emph{intermediate regularly varying}, namely
\[
    \limsup_{\delta \to 0} \limsup_{t \to \infty} \frac{\P{S > t}}{\P{S > (1 - \delta) t}} = 1.
\]
This property suffices for the computation in our proof above, and it suffices for many of the prior works showing $C_\pi = 1$ for various policies~$\pi$ \citep{nunez-queija_queues_2002, scully_characterizing_2020, scully_when_2024}. The only step of the proof that requires (non-intermediate) regular variation is \cref{eq:heavy-tailed_necessary_condition}, the result of \citet[Proposition~1]{wierman_tailoptimal_2012}. If their result could be generalized to give the same necessary condition for $C_\pi < \infty$ even when $S$ is intermediate regularly varying, it would imply $\inf_\pi C_\pi = 1$ in that setting, too.

\section{Reduction to the batch scheduling problem for unknown sizes}
\label{sec:reduction}
In this section, we extend the batch reduction to the case of unknown sizes with only size-label information for jobs. Specifically, we expand the definitions from \cref{sec:batch} and show how the unknown-size case differs from the known size case.

\begin{definition}
    A \emph{batch instance} $\calI = \gp[\big]{(a_1, l_1), \ldots, (a_n, l_n)}$ is a finite batch of arrival times and labels.
\end{definition}

Unlike the batch instance in \cref{sec:batch}, which consists of fully deterministic arrival time, job size pairs, we additionally need to define a \emph{job size model}, which describes the distribution of job sizes given instance information.

\begin{definition}
    A \emph{job size model} $\calS$ is a function that maps batch instances $\calI = \gp[\big]{(a_1, l_1), \ldots, (a_n, l_n)}$ to a joint distribution $\gp{S_1, \dots, S_n}$ on job sizes.
\end{definition}

We consider two job size models. The first model is $\calS_{\mathrm{ind}}$, which is the model where the joint distribution of job sizes for an instance $\calI$ is simply the product of the conditional distributions~$(S \given l_i)$. That is,
\[
    \gp[\big]{\calS_{\mathrm{ind}}(\calI)}_i &\sim (S \given l_i), & \gp[\big]{\calS_{\mathrm{ind}}(\calI)}_i \text{ independent of } \gp[\big]{\calS_{\mathrm{ind}}(\calI)}_j &\text{ if } i \neq j.
\]

The second job size model is the busy period model $\calS_{\mathrm{busy}}$. $\calS_{\mathrm{busy}}$ maps from batches to joint distributions of job sizes in the following way. Let $\mathit{BP}$ denote the busy period distribution, where the triples $(A_i, L_i, S_i) \sim \mathit{BP}$. Then $\calS_{\mathrm{busy}}$ is a mapping such that if $\calB$ is a random variable $(A_i, L_i)$, then $(\calB, \calS_{\mathrm{busy}}(\calB)) \sim \mathit{BP}$. That is, the joint distribution of sizes conditional on the batch instance coming from a busy period is also distributed such that it makes the triple distributed with the busy period distribution.

We can now define the cost function given a batch and job size model.

\begin{definition}
    The $\theta$-cost of a policy $\pi$ for an instance $\calI$ under a given job size model $\calS$ is given by
    \[
        K_\pi(\theta, \calI, \calS) = \sum_{i=1}^n\E{\exp\gp{\theta (D_{\pi, i} - a_i)}},
    \]
    where the expectation is over the joint job size distribution $\calS(\calI)$ and any randomness in the policy.
\end{definition}

Under the job size model $\calS_{\mathrm{ind}}$, we have the following optimality result:
\begin{theorem}\label{thm:nonpreemptive_batch_policy}
    For any batch instance $\calI = \gp[\big]{(a_i, l_i)}_{i=1}^n$, the \cheat{\theta} policy minimizes $K_\pi(\theta, \calI, \calS_{\mathrm{ind}})$ in the nonpreemptive batch scheduling problem. Specifically, for any nonpreemptive policy $\pi$ under this job size model,
    \[
        K_{\cheat{\theta}}(\theta, \calI, \calS_{\mathrm{ind}}) \leq \min_{\pi} K_\pi(\theta, \calI, \calS_{\mathrm{ind}}).
    \]
\end{theorem}

\begin{proof}
    We observe that \cheat{\theta} serves jobs in~$\calI$ in the same order as the Weighted Discounted Shortest Expected Processing Time (WDSEPT) rule, with the weights equal to $\exp\gp{-\theta a_i}$ and negative discount rate $-\theta$ for each job~$i$. This is because boosted arrival time $a_i - b_\theta(l_i)$ is a monotonic function, namely the negative log, of WDSEPT's priority index:
    \[
        a_i - b_\theta(l_i) = -\frac{1}{\theta} \log\gp*{\exp(-\theta a_i) \frac{\E{\exp(\theta S) \given L = l_i}}{\E{\exp(\theta S) \given L = l_i} - 1}}.
    \]
    The proof is an interchange argument identical to that of \citet[Theorem 10.1.3]{pinedo_scheduling_2016}, with the signs for the discount rate and objective flipped. Specifically, with negative discounting, we define discounted completion time as $\exp(\theta D_i) - 1$ instead of $1 - \exp(\theta D_i)$ to keep the sign positive.
\end{proof}

In \cref{thm:optimal_batch_policy}, we concluded that \cheat{\theta}'s optimality for all batch instances could be translated to a bound in the M/G/1. Why does this not work in \cref{thm:nonpreemptive_batch_policy}? The issue is that our optimality argument in the M/G/1 relies on sampling instances from busy periods. In particular, letting $\calB$ be the random batch instance resulting from a busy period, by reasoning to that in the proof \cref{thm:optimal_batch_policy}, we have
\[
    \E{\exp\gp{\theta T_\pi}} = \frac{\E{K_\pi(\theta, \calB, \calS_{\mathrm{busy}})}}{|\calB|},
\]
where the expectations on the right-hand side are over the distribution of~$\calB$. The key difference is that the job size model is $\calS_{\mathrm{busy}}$, not $\calS_{\mathrm{ind}}$, so \cref{thm:nonpreemptive_batch_policy} does not apply.

Why is there not a similar issue in the full-information case? Because when a job's size is its label, the instance $\calI$ already contains all the job size information. In particular, $\calS_{\mathrm{ind}}(\calI) = \calS_{\mathrm{busy}}(\calI)$ for all instances~$\calI$, with both models simply yielding a deterministic vector of sizes extracted from the labels in the instance~$\calI$.

\section{Optimizing the boost function in the partial-information setting}
\label{sec:partial-information_boost_optimality}
In this appendix, we show that in the partial-information setting, \boost{\gamma} has a lower tail constant than \boost{} with any other boost function:
\[
    C_{\boost{\gamma}} \leq C_{\boost{}}.
\]
To keep the argument simple, we assume a finite set of labels $\bbL = \{1, \dots, n\}$. The result can be generalized to arbitrary sets of labels using a calculus of variations argument.

For brevity, we adopt the notation
\[
    p_i &= \P{L = i}, & b_i &= b(i), & s_i &= \E{\exp({\gamma S)} \given L = i}.
\]
Choosing a boost function amounts to choosing a vector of boosts $(b_1, \dots, b_n)$. To make the dependence of $C_{\boost{}}$ on the $b_i$ explicit, we let
\[
    C(b_1, \dots, b_n) = C_{\boost{}}.
\]
By \cref{def:boost_theta}, we can write $C_{\boost{\gamma}} = C\gp{b^*_1, \dots, b^*_n}$, where
\[
    b^*_i = \frac{1}{\gamma} \log \frac{s_i}{s_i - 1}.
\]
Optimality of \boost{\gamma} among Boost policies amounts to showing that $(b^*_1, \dots, b^*_n)$ is a minimizer of~$C(\cdot)$. To show this, it suffices to show the following two claims:
\* $C(\cdot)$ is convex.
\* $\nabla C(b^*_1, \dots, b^*_n) = 0$.
\*/

We start with convexity of $C(\cdot)$. Writing $p_i = \P{L = i}$, we can rewrite the tail constant from \cref{thm:general_transform} as
\[
    C(b_1, \dots, b_n)
    &= C_W \gp[\bigg]{\sum_{i = 1}^n p_i s_i \exp(-\gamma b_i)} \exp\gp[\bigg]{\sum_{i = 1}^n \lambda p_i b_i (s_i - 1)} \\
    &= C_W \sum_{i = 1}^n p_i s_i \exp\gp[\bigg]{-\gamma b_i + \sum_{j = 1}^n \lambda p_j b_j (s_j - 1)}.
\]
We observe that $C(\cdot)$ a sum of compositions of linear functions with $\exp$, so it is convex. In fact, $C(\cdot)$ is log-convex, because it is a product of a mixture of log-convex functions, which is log-convex \citep[proof of Lemma~A.3]{kontorovich_exact_2019}, and another log-convex function.

To show the gradient at $(b^*_1, \dots, b^*_n)$ is zero, we take the derivative of $C(b_1, \dots, b_n)$ with respect to $b_k$, obtaining
\[
    \frac{\partial}{\partial b_k} C(b_1, \dots, b_n)
    = C_W \sum_{i = 1}^n p_i s_i \gp[\big]{\lambda p_k (s_k - 1) - \gamma \1(i = k)} \exp\gp[\bigg]{-\gamma b_i + \sum_{j = 1}^n \lambda p_j b_j (s_j - 1)}.
\]
This derivative is zero if and only if
\[
     \lambda p_k (s_k - 1) \sum_{i = 1}^n p_i s_i \exp(-\gamma b_i) = \gamma p_k s_k \exp(-\gamma b_k).
\]
Plugging in $b_i = b^*_i$ and noticing $s_i \exp(-\gamma b^*_i) = s_i - 1$, it suffices to show
\[
    \lambda \sum_{i = 1}^n p_i (s_i - 1) = \gamma.
\]
But the right-hand is $\lambda (\E{\exp(\gamma S)} - 1)$, so this is simply the definition of~$\gamma$ from~\cref{eq:model:gamma}.

\section{Comparison of Nudge-M and Boost tail constants}
\label{sec:nudge_vs_boost}
In this appendix, we show that the \boost{\gamma} policy for unknown job sizes with two labels is state-of-the-art among policies for this setting. Specifically, we compare to the known state-of-the-art, Nudge-M \citep{charlet_tail_2024}, and show that for any choice for the parameter $K$, there is a Boost policy which achieves a better tail constant. In particular, this means that the optimal Boost policy in this setting, namely \boost{\gamma} (see \cref{sec:partial-information_boost_optimality}), has better tail constant than the optimal Nudge-M policy.

Recall that Nudge-M partitions the jobs into type-1 and type-2 jobs, allowing any type-1 job to pass any type-2 job that has arrived as one of the last M jobs. This can be thought of as the setting where each job is given label~$1$ or label~$2$, and a scheduling decision is made based on this label. From here onward we will refer to type-1 jobs as jobs with label~$1$ and type-2 jobs as jobs with label~$2$.

For brevity, we adopt the notation
\[
    p_1 &= \P{L = 1}, &
    p_2 &= \P{L = 2}, \\
    s_1 &= \E{e^{\gamma S} \given L = 1}, &
    s_2 &= \E{e^{\gamma S} \given L = 2}, &
    s &= \E{e^{\gamma S}}, \\
    b_1 &= b(1), &
    b_2 &= b(2).
\]
Note that with just two labels, $p_2 = 1-p_1$. We assume both $p_1$ and $p_2$ are nonzero.

The tail constant of Nudge-M for any choice of parameter $K$ is given by rewriting the expression in \citet[Theorem 1]{charlet_tail_2024} as
\[
    \label{eq:nudgem_tail_constant}
    \frac{C_\nudgem}{C_\fcfs} = \frac{p_1 s_1}{s}\frac{\gp{p_1 s_1 + p_2}^K}{s^K} + \frac{p_2 s_2}{s}\gp{p_1 s_1 + p_2}^K.
\]
Recall that the tail constant of \boost{} is given by
\[
    \frac{C_{\boost{}}}{C_\fcfs} = \frac{1}{s}\E{e^{\gamma \gp{S - B}}}\E{e^{\gamma V}}.
\]
Under the two-label case above, the term $\E{e^{\gamma \gp{S - B}}}$ can be written as
\[
    \E{e^{\gamma \gp{S - B}}} &= p_1\E{e^{\gamma S}\given L = 1}e^{-\gamma b_1} + p_2\E{e^{\gamma S} \given L = 2}e^{-\gamma b_2}\\
    &= p_1s_1 e^{-\gamma b_1} + p_2 s_2 e^{-\gamma b_2},
\]
where~$b_1$ and~$b_2$ are, respectively, the boosts given to jobs of label~$1$ and label~$2$. The crossing work can be similarly written, using \cref{lem:work_crossing_expression}, as the sum of an arrival stream of label-$1$ jobs and of label-$2$ jobs, i.e., as
\[
    \E{e^{\gamma V}} &= \exp\gp[\big]{\lambda p_1 b_1 \E{e^{\gamma S} - 1 \given L = 1} + \lambda p_2 b_2 \E{e^{\gamma S} - 1 \given L = 2}}\\
    &= \exp\gp[\big]{\lambda p_1 \gp{b_1 - b_2} \gp{s_1 - 1} + \lambda b_2 \E{e^{\gamma S} - 1}}\\
    &= \exp\gp[\big]{\lambda p_1 \gp{b_1 - b_2}\gp{s_1 - 1} + \gamma b_2}.
\]
The tail constant for \boost{} can thus be written as
\[
    \frac{C_{\boost{}}}{C_\fcfs} &= \frac{1}{s}\gp[\big]{p_1s_1 e^{-\gamma b_1} + p_2 s_2 e^{-\gamma b_2}} \exp\gp[\big]{\lambda p_1 \gp{b_1 - b_2} \gp{s_1 - 1} + \gamma b_2}\\
    &= \gp*{\frac{p_1 s_1}{s}e^{-\gamma \gp{b_1 - b_2}} + \frac{p_2 s_2}{s}} \exp\gp[\big]{\lambda p_1 \gp{b_1 - b_2} \gp{s_1 - 1}}.
\]
Consider an arbitrary~$K$ for Nudge-M. For such an~$K$, consider $b_1 - b_2$ such that
\[
    \label{eq:boost_diff_for_nudgem}
    b_1 - b_2 = \frac{K\log\gp{s}}{\gamma}.
\]
There are many boost functions that satisfy this, e.g. $b_1 = \frac{K\log\gp{s}}{\gamma}$ and $b_2 = 0$.

For a boost function satisfying \cref{eq:boost_diff_for_nudgem}, we observe that $s^K = e^{\gamma (b_1 - b_2)}$. Using this observation, we can rewrite the expression in \cref{eq:nudgem_tail_constant} as
\[
    \frac{C_\nudgem}{C_\fcfs} &= \frac{p_1 s_1}{s}\frac{\gp{p_1 s_1 + p_2}^K}{e^{\gamma (b_1 - b_2)}} + \frac{p_2 s_2}{s}\gp{p_1 s_1 + p_2}^K\\
    &= \gp*{\frac{p_1 s_1}{s}e^{-\gamma (b_1 - b_2)} + \frac{p_2 s_2}{s}}\gp{p_1 s_1 + p_2}^K\\
    &= \gp*{\frac{p_1 s_1}{s}e^{-\gamma (b_1 - b_2)} + \frac{p_2 s_2}{s}}\gp{1 + p_1(s_1 - 1)}^K.
\]
Therefore, to compare the tail constant of \boost{} and that of Nudge-M, it suffices to compare the expressions $\exp\gp[\big]{\lambda p_1 \gp{b_1 - b_2} \gp{s_1 - 1}}$ and $\gp{1 + p_1(s_1 - 1)}^K$. We rewrite the first expression as
\[
    \exp\gp[\big]{\lambda p_1 \gp{b_1 - b_2} \gp{s_1 - 1}} &= \exp\gp*{\frac{\gamma (b_1 - b_2) p_1 \gp{s_1 - 1}}{s - 1}}\\
    &= \exp\gp{\gamma (b_1 - b_2)}^{\frac{p_1\gp{s_1 - 1}}{s - 1}},
\]
and rewrite the second expression as
\[
    \gp{1 + p_1(s_1 - 1)}^K &= \gp{1 + p_1(s_1 - 1)}^{\frac{\log\gp{e^{\gamma (b_1 - b_2)}}}{\log\gp{s}}}\\
    &= \exp\gp*{\frac{\log\gp{1 + p_1\gp{s_1 - 1}}\log\gp{e^{\gamma (b_1 - b_2)}}}{\log{s}}}\\
    &= \exp\gp{\gamma (b_1 - b_2)}^{\frac{\log\gp{1 + p_1\gp{s_1 - 1}}}{\log\gp{1 + \gp{s - 1}}}}.
\]
Finally, we observe that
\[
    \exp\gp{\gamma (b_1 - b_2)}^{\frac{p_1\gp{s_1 - 1}}{s - 1}} < \exp\gp{\gamma (b_1 - b_2)}^{\frac{\log\gp{1 + p_1\gp{s_1 - 1}}}{\log\gp{1 + \gp{s - 1}}}},
\]
since $s - 1 \geq p_1\gp{s - 1}$ and $\frac{\log\gp{1 + x}}{x}$ is monotonically decreasing in~$x$. This shows that for any $K$, \boost{} with $b_1 - b_2 = \frac{K\log\gp{s}}{\gamma}$ performs better than Nudge-M, implying that the optimal Boost policy, namely \boost{\gamma} (\cref{sec:partial-information_boost_optimality}), performs better than the optimal Nudge-M policy.

\begin{figure}
    \centering
    \includegraphics[width=\subfigwidth]{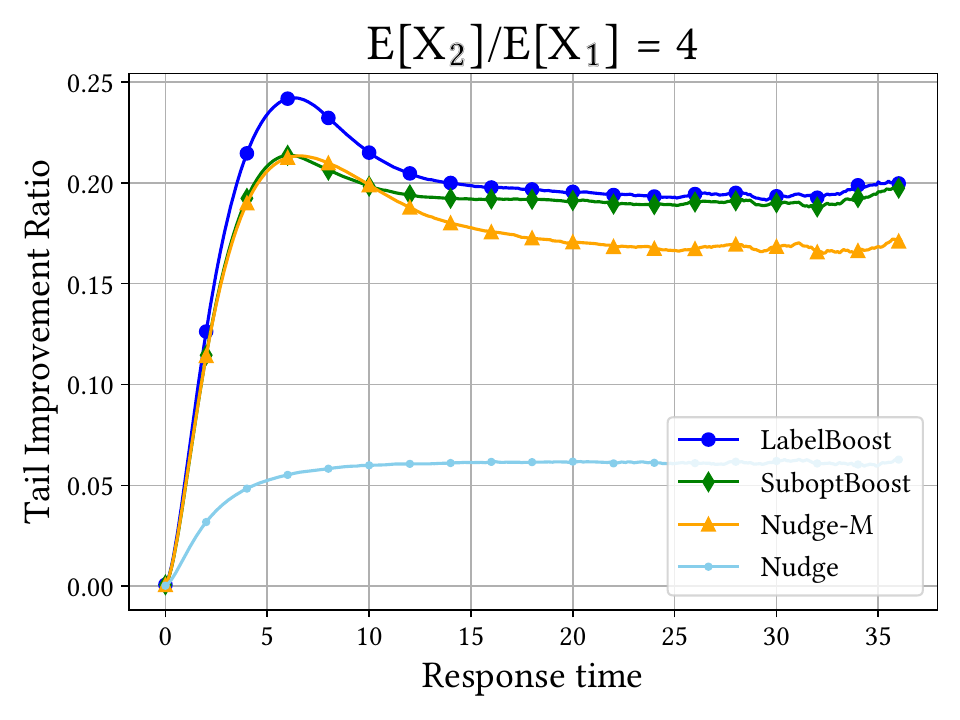}
    \Description{Graph showing tail improvement ratio of $\gamma$-Boost, SuboptBoost, Nudge-M, and Nudge for the two-class Hyperexponential distribution: $x$ axis is response time (range 0 to 40), $y$ axis is tail improvement ratio (range 0 to 0.25). For all policies, each of their curves starts at y = 0, peaks around x = 6, then reduces as it converges to an asymptote. Boost: peaks at about y = 0.24, asymptotes to about y = 0.19. Subopt Boost: peaks at about y = 0.22, asymptotes to about y = 0.19, but under Boost. Nudge-M: peaks at about y = 0.22, asymptotes to about y = 0.17. Nudge: peaks at about y = 0.06, asymptotes to about y = 0.06.}
    \caption{Empirical TIR of \boost{\gamma}'s performance compared to the Nudge-M policy in the two-label setting. Parameters are set to match those of \citet[Fig.~5]{charlet_tail_2024}: label-$1$ and label-$2$ jobs are both exponential, with $\E{X_2}/\E{X_1} = 4$, the probability of a label-$1$ job $p_1 = 2/3$, arrival rate $\lambda = 0.7$. For Nudge-M, we use the optimal parameter $K = 5$. Observe that both \boost{\gamma}, denoted as LabelBoost, and \boost{} with $b_1 - b_2 = \frac{K \log\gp{s}}{\gamma}$, denoted as SuboptBoost, outperform Nudge-M.}
    \label{fig:nudge_m_comp}
\end{figure}

\received{January 2023}
\received[revised]{April 2024}
\received[accepted]{April 2024}

\end{document}